\PassOptionsToPackage{unicode}{hyperref}
\PassOptionsToPackage{hyphens}{url}
\PassOptionsToPackage{dvipsnames,svgnames,x11names}{xcolor}

\documentclass[11pt]{article}

\usepackage{bm}
\usepackage{color}
\def\bSig\mathbf{\Sigma}

\usepackage{amsmath}
\usepackage{amssymb}
\usepackage{amsthm}
\usepackage[linesnumbered,ruled,vlined]{algorithm2e}
\usepackage{subcaption}
\usepackage{float}
\usepackage{booktabs}
\usepackage{multirow}

\usepackage[hidelinks,colorlinks=true,linkcolor=blue,citecolor=blue]{hyperref}

\usepackage{graphicx}
\usepackage{subcaption}
\usepackage{natbib}
\setcitestyle{authoryear,open={(},close={)}}
\usepackage{booktabs}
\usepackage{amsfonts}
\usepackage{amssymb}
\usepackage{dsfont}
\usepackage{eucal}
\usepackage{multirow}
\usepackage[dvipsnames]{xcolor}
\usepackage{caption}
\newtheorem{theorem}{Theorem} 
\newtheorem{proposition}{Proposition} 
\newtheorem{corollary}{Corollary} 
\newtheorem{assumption}{Assumption}

\newtheorem{example}{Example}
\usepackage{apptools}
\AtAppendix{\counterwithin{lemma}{section}}
\AtAppendix{\counterwithin{corollary}{section}}
\AtAppendix{\counterwithin{proposition}{section}}
\AtAppendix{\counterwithin{theorem}{section}}
\AtAppendix{\counterwithin{figure}{section}}
\AtAppendix{\counterwithin{table}{section}}

\usepackage{xr-hyper}

\makeatletter
\newcommand*{\addFileDependency}[1]{%
  \typeout{(#1)}%
  \@addtofilelist{#1}%
  \IfFileExists{#1}{}{\typeout{No file #1.}}
}
\makeatother

\DeclareMathOperator*{\argmax}{arg\,max}

\def\bSig\mathbf{\Sigma}

\usepackage{iftex}
\ifPDFTeX
  \usepackage[T1]{fontenc}
  \usepackage[utf8]{inputenc}
  \usepackage{textcomp} 
\else 
  \usepackage{unicode-math}
  \defaultfontfeatures{Scale=MatchLowercase}
  \defaultfontfeatures[\rmfamily]{Ligatures=TeX,Scale=1}
\fi
\usepackage{lmodern}
\ifPDFTeX\else  
\fi
\IfFileExists{upquote.sty}{\usepackage{upquote}}{}
\IfFileExists{microtype.sty}{
  \usepackage[]{microtype}
  \UseMicrotypeSet[protrusion]{basicmath} 
}{}
\makeatletter
\@ifundefined{KOMAClassName}{
  \IfFileExists{parskip.sty}{%
    \usepackage{parskip}
  }{
    \setlength{\parindent}{0pt}
    \setlength{\parskip}{6pt plus 2pt minus 1pt}}
}{
  \KOMAoptions{parskip=half}}
\makeatother
\usepackage{xcolor}
\makeatletter
\ifx\paragraph\undefined\else
  \let\oldparagraph\paragraph
  \renewcommand{\paragraph}{
    \@ifstar
      \xxxParagraphStar
      \xxxParagraphNoStar
  }
  \newcommand{\xxxParagraphStar}[1]{\oldparagraph*{#1}\mbox{}}
  \newcommand{\xxxParagraphNoStar}[1]{\oldparagraph{#1}\mbox{}}
\fi
\ifx\subparagraph\undefined\else
  \let\oldsubparagraph\subparagraph
  \renewcommand{\subparagraph}{
    \@ifstar
      \xxxSubParagraphStar
      \xxxSubParagraphNoStar
  }
  \newcommand{\xxxSubParagraphStar}[1]{\oldsubparagraph*{#1}\mbox{}}
  \newcommand{\xxxSubParagraphNoStar}[1]{\oldsubparagraph{#1}\mbox{}}
\fi
\makeatother

\usepackage{longtable,booktabs,array}
\usepackage{calc} 
\usepackage{etoolbox}
\makeatletter
\patchcmd\longtable{\par}{\if@noskipsec\mbox{}\fi\par}{}{}
\makeatother
\IfFileExists{footnotehyper.sty}{\usepackage{footnotehyper}}{\usepackage{footnote}}
\makesavenoteenv{longtable}
\usepackage{graphicx}
\makeatletter
\def\maxwidth{\ifdim\Gin@nat@width>\linewidth\linewidth\else\Gin@nat@width\fi}
\def\maxheight{\ifdim\Gin@nat@height>\textheight\textheight\else\Gin@nat@height\fi}
\makeatother
\setkeys{Gin}{width=\maxwidth,height=\maxheight,keepaspectratio}
\makeatletter
\def\fps@figure{htbp}
\makeatother

\makeatletter
\@ifpackageloaded{caption}{}{\usepackage{caption}}
\AtBeginDocument{%
\ifdefined\contentsname
  \renewcommand*\contentsname{Table of contents}
\else
  \newcommand\contentsname{Table of contents}
\fi
\ifdefined\listfigurename
  \renewcommand*\listfigurename{List of Figures}
\else
  \newcommand\listfigurename{List of Figures}
\fi
\ifdefined\listtablename
  \renewcommand*\listtablename{List of Tables}
\else
  \newcommand\listtablename{List of Tables}
\fi
\ifdefined\figurename
  \renewcommand*\figurename{Figure}
\else
  \newcommand\figurename{Figure}
\fi
\ifdefined\tablename
  \renewcommand*\tablename{Table}
\else
  \newcommand\tablename{Table}
\fi
}
\@ifpackageloaded{float}{}{\usepackage{float}}
\floatstyle{ruled}
\@ifundefined{c@chapter}{\newfloat{codelisting}{h}{lop}}{\newfloat{codelisting}{h}{lop}[chapter]}
\floatname{codelisting}{Listing}

\makeatother
\makeatletter
\@ifpackageloaded{caption}{}{\usepackage{caption}}
\@ifpackageloaded{subcaption}{}{\usepackage{subcaption}}
\makeatother

\ifLuaTeX
  \usepackage{selnolig}  
\fi
\usepackage[]{natbib}
\usepackage{bookmark}

\IfFileExists{xurl.sty}{\usepackage{xurl}}{} 
\hypersetup{
  pdftitle={Title},
  pdfauthor={Author 1; Author 2},
  pdfkeywords={3 to 6 keywords, that do not appear in the title},
  colorlinks=true,
  linkcolor={blue},
  filecolor={Maroon},
  citecolor={Blue},
  urlcolor={Blue},
  pdfcreator={LaTeX via pandoc}}

\usepackage{authblk}
\title{\bf  {Bayesian Variable Selection with the Quasi-Posterior}}
\author[1]{Beniamino Hadj-Amar}
\author[2]{Jack Jewson}
\affil[1]{Department of Epidemiology and Biostatistics University of South Carolina, Columbia, SC}
\affil[2]{Department of Econometrics and Business Statistics \\ Monash University, Victoria, Australia}
\affil[ ]{\textit {\textcolor{blue}{hadjamar@mailbox.sc.edu, jack.jewson@monash.edu}}}
\date{March 2026}

\begin{document}




\def\spacingset#1{\renewcommand{\baselinestretch}%
{#1}\small\normalsize} \spacingset{1}

\setcounter{Maxaffil}{0}
\renewcommand\Affilfont{\itshape\small}

\spacingset{1.42} 

\maketitle
\begin{abstract}
The Bayesian approach provides powerful methods for variable selection. The ability to incorporate sparsity through prior beliefs and account for parameter uncertainty allows Bayesian variable selection to consistently identify which of the variables are active and exhibit strong finite-sample performance. However, Bayesian methods require the correct specification of full likelihoods for the data, and there is increasing awareness of the problems that model misspecification causes for variable selection. Current approaches to mitigate misspecification either require complex models, detracting from the interpretability of the variable selection task, or move outside rigorous Bayesian uncertainty quantification and provide no recognised method for variable selection. This paper establishes the model quasi-posterior as a principled tool for variable selection. We prove that the model quasi-posterior shares desirable properties of Bayesian variable selection without requiring full likelihood specification. Instead, the quasi-posterior combines a prior with a quasi-likelihood and requires only specification of mean and variance functions, and is therefore robust to other aspects of the data. Marginalising the quasi-likelihood is analytically possible for linear regression, and Laplace approximations are used beyond this to ensure computational tractability. Extensive simulation studies illustrate improved variable selection accuracy across diverse data-generating scenarios when compared with likelihood-based Bayesian variable selection and lasso-penalized methods. We further demonstrate practical relevance through applications to real datasets from social science and genomics.
\end{abstract}


\spacingset{1.45}

\section{Introduction}
\label{sec:intro}


Consider wanting to identify the subset of possible predictors $x\in \mathcal{X}\subseteq \mathbb{R}^{p}$ that are relevant for predicting a response $y\in \mathcal{Y}\subseteq \mathbb{R}$. Such a task is generally  referred to as variable selection, a specific type of model selection. 
The Bayesian paradigm provides desirable tools for variable selection.
The prior encodes beliefs about the number of active predictors and the posterior quantifies the confidence that any given predictor is truly active. 
As a result, Bayesian variable selection naturally penalizes overly complex models, favoring simpler models that are well supported by the data. This enables consistent identification of the true model among a set of nested candidates \citep[e.g.][]{dawid2011posterior}, and has been shown to yield strong finite-sample performance in a variety of settings \citep[e.g.][]{rossell2018tractable}.


This article considers the class of generalised linear models (GLMs) \citep{mccullagh1989generalized} that use a linear predictor to define
\begin{equation}
    {E}\left[Y\mid X = x\right] = \mu(x^T\beta),\label{Equ:GLM}
\end{equation}
where $\beta\in\mathbb{R}^p$ are model parameters capturing the extent to which each predictor affects the response $y$ and $\mu^{-1}(\cdot):\mathbb{R}\mapsto \mathcal{Y}$ is known as the GLM link function.  
In order to conduct variable selection, we introduce $\gamma = (\gamma_1, \ldots, \gamma_p)\in \{0, 1\}^p$ with $\gamma_j = \mathbb{I}(\beta_j \neq 0)$. Here, $\gamma_j = 0$ means that variable $x_j$ is not related to response $y$, i.e. is inactive. We refer to each unique configuration of predictors $\gamma$, as a model. 

A considerable drawback of the Bayesian paradigm, however,  is the need to specify a full likelihood for $y\mid X$ in order to define the model posterior $\pi(\gamma\mid y, X)$. This requires \eqref{Equ:GLM} to be furnished with a distribution $F(\cdot\mid X)$, with density or mass function $f$, that describes in much greater detail how $y$ depends on $X$ than what is specified by \eqref{Equ:GLM}. 
Recent works have demonstrated that misspecifying the model distribution \( F \) relative to the true data-generating distribution \( F_0 \) can negatively affect the performance of Bayesian variable selection. \citet{rossell2018tractable, rossell2023additive} 
observed a decrease in finite-sample power to detect truly active covariates when the data have heavier tails than those assumed by $F$.
Conversely, \cite{grunwald2017inconsistency, heide2020safe} demonstrated an increased rate of false positives relative to the predictive risk minimising model when the data contain inliers. 
\cite{fontana2023scalable} show that Bayesian variable selection tends to select unnecessarily complex models when \eqref{Equ:GLM} is also misspecified. Similar phenomena have also  been observed in Bayesian mixture models \citep{miller2018robust, cai2021finite} and autoregressive time series models \citep{miller2018robust}. 

In this paper, we consider the quasi-posterior \citep{lin2006quasi, agnoletto2025bayesian}, a generalized Bayesian posterior \citep{bissiri2016general} based on the quasi-likelihood \citep{mccullagh1983quasi, wedderburn1974quasi}, for use in variable selection. 
We demonstrate that the resulting model quasi-posterior
retains 
desirable properties of Bayesian variable selection procedures without the need to correctly specify a full likelihood model $F$.
Instead, the quasi-likelihood enables calibrated inference on the regression parameters $\beta$ and corresponding model $\gamma$ that only requires specification of the mean and variance of the conditional distribution $y\mid x$. 
Furthermore, unlike other generalised Bayesian methods, there is no need to tune a `learning-rate' parameter to achieve this. 


Existing approaches to address model misspecification in Bayesian inference can be broadly categorized into two approaches: (i) modifying the model, or (ii) modifying the updating rule. Methods falling under  (i) include those that specify more complex models with greater flexibility to adapt to the data, e.g. Bayesian non-parametric regression models \citep[e.g.][]{de2004anova, muller2004nonparametric}, or models specifically designed to be robust to particular types of misspecifications, such as outliers \citep[e.g.][]{berger1994overview, rossell2018tractable}. However, more complex models require larger sample sizes and greater computational resources to fit, and are generally less interpretable than simpler structures such as \eqref{Equ:GLM}. Further, even models designed to be robust  are still inevitably misspecified to some degree. In contrast, category (ii) takes a different approach, modifying not the model itself, but the rule by which information from the data is incorporated into the prior. This includes methods that raise the likelihood to a power when conducting posterior inferences \citep{walker2001bayesian,zhang2006information, grunwald2017inconsistency, miller2018robust, lyddon2019general, syring2019calibrating, wu2023comparison} or update the prior using the exponential of a loss function rather than a likelihood \citep{bissiri2016general, syring2017gibbs, jewson2018principles, cherief2020mmd, knoblauch2022generalized}. However, departing from Bayes' rule as the basis for posterior updating removes the principled probabilistic foundation of these methods. As a result, the validity of the uncertainty quantification resulting from these methods often relies on carefully tuned `learning rate' hyperparameters which has so far prevented their applicability to variable and model selection problems.

 To overcome these limitations, we establish the model quasi-posterior as a principled tool for Bayesian variable selection.
We show that the model quasi-posterior converges asymptotically to the same distribution as the Bayesian posterior for $\gamma$ when the likelihood model is correctly specified.
To the best of our knowledge this represents the first instance of generalised Bayesian variable selection with  asymptotically valid uncertainty quantification.
%
From a computational perspective, we show that the marginal quasi-likelihood admits a closed-form expression in 
linear regression models, and use Laplace approximations in other cases. We further identify a broad class of models for which the log quasi-likelihood is concave, guaranteeing the asymptotic accuracy of the Laplace approximation and enabling stable and scalable optimisation. Building on these results, we use an efficient random-scan Gibbs sampler for exploring the model quasi-posterior, combined with a caching strategy that substantially reduces repeated optimisation and Hessian evaluations. 
Extensive simulation studies illustrate that the model quasi-posterior achieves superior finite-sample variable selection performance under a range of model misspecification scenarios, including heavy-tailed regression, inliers, and overdispersed count data.
We further illustrate its practical relevance through applications to real datasets from social science and genomics.

The remainder of the paper is organised as follows. Section~\ref{Sec:QP} introduces the quasi-likelihood and quasi-posterior framework for generalized Bayesian inference. Section~\ref{Sec:Inference} formalizes quasi-posterior model selection, defining the model quasi-posterior and the associated marginal quasi-likelihood.
Section~\ref{Sec:Theory} establishes the asymptotic properties of the proposed quasi-posterior methodology. Section~\ref{sec:posterior_inference} describes the posterior inference strategy for exploring the model space and performing variable selection. Section~\ref{sec:simul_studies} reports simulation studies illustrating the finite-sample advantages of the quasi-posterior under various forms of model misspecification, and Section~\ref{Sec:real_data} presents real-data applications. Section ~\ref{Sec:conclusion} concludes the paper with a discussion. All methods are implemented in the \textit{R} package \texttt{modelSelectionQP}
which will be made publicly available on GitHub upon acceptance.

\section{Generalized Bayesian inference using quasi-likelihoods} \label{Sec:QP}


In this section, we review the quasi-posterior framework, 
which requires only the specification of first and second moments of the data rather than a fully specified likelihood. 
%
%
Let \( y = (y_1, \ldots, y_n) \in \mathbb{R}^n \) denote the vector of responses, where each \( y_i \in \mathbb{R} \), and let \( n \) denote the number of observed data points. Let us also define \( x_i \in \mathbb{R}^p \) to be the associated \( p \)-dimensional covariate vector  for each \( i = 1, \ldots, n \), and $X\in\mathbb{R}^{n\times p}$ a matrix whose $i$th row is $x_i$.  We assume that the data are generated from an unknown distribution \( F_0 \)
such that
\(
y_i \mid x_i \sim F_0,
\) 
with conditional mean and variance given by
\[
{E}_{F_0}[y_i \mid x_i] = \mu(x_i^\top \beta^{\ast}), \qquad var_{F_0}[y_i \mid x_i] = \psi^{\ast} V\left( \mu(x_i^\top \beta^{\ast}) \right),
\]
respectively, where \( \mu: \mathbb{R} \to \mathcal{Y} \) is the inverse link function, \( \beta^{\ast} \in \mathbb{R}^p \) is the true regression parameter, \( V: \mathcal{Y} \to \mathbb{R}_+ \) is a variance function, and \( \psi^{\ast} > 0 \) is an unknown dispersion parameter. 
Common examples of \( \mu \) and \( V \) are illustrated in Examples~\ref{Ex:LinearRegression} and~\ref{Ex:Count} below; see also \citet{agnoletto2025bayesian}. 

Traditionally, conducting frequentist or Bayesian inference for $\beta^{\ast}$ from observations $y$ and $x$ requires the specification of a likelihood $f(y_i\mid x_i, \beta)$ that describes the distribution of response $y_i$ given predictors $x_i$ for parameter values $\beta$. 
%
%
Specifically, \( f \) must not only capture the mean and variance structure of the true data-generating process \( F_0 \), but also match the infinitely many other probability statements it implies \citep{goldstein1990influence}. This is rarely achievable in practice—meaning \( f \) is almost certainly misspecified. 

\subsection{Quasi-likelihood and quasi-posterior}


Rather than specifying a full model for \( F_0 \), \citet{wedderburn1974quasi, mccullagh1983quasi} showed that frequentist inference on \( \beta^{\ast} \) can be conducted by using
$\hat{\beta} := \arg\max_{\beta \in \mathbb{R}^p} Q_n(y, X; \beta, \psi)$
where
\begin{equation}
Q_n(y, X; \beta, \psi) = \frac{1}{n} \sum_{i=1}^n \ell_{\psi}(y_i; x_i, \beta),\quad \quad\ell_{\psi}(y_i; x_i, \beta) := \int_a^{\mu(x_i^\top \beta)} \frac{y_i - t}{\psi V(t)} \, dt \label{Equ:QP_loss}
\end{equation}
is the \textit{log quasi-likelihood} that only requires the specification of the mean $\mu(\cdot)$ and variance $V(\cdot)$ functions of $F_0$. Here, \( a \in \mathcal{Y} \) is an arbitrary constant that does not affect \( \hat{\beta} \).  While estimating $\psi^{\ast}$ is not required for obtaining $\hat{\beta}$, it is required to calculate the sampling distribution of $\hat{\beta}$ and therefore conduct inference. A consistent estimator of \( \psi^{\ast} \) is \citep{mccullagh1989generalized}
\begin{equation}
\hat{\psi} = \frac{1}{n - p} \sum_{i=1}^n \frac{(y_i - \mu(x_i^\top \hat{\beta}))^2}{V(\mu(x_i^\top \hat{\beta}))}, \label{Equ:psi_est}
\end{equation}
\cite{van2012quasi, liu2018penalized} further combined the quasi-likelihood with lasso penalties to estimate sparse $\hat{\beta}$ and conduct variable selection.

The \textit{quasi-posterior} \citep{lin2006quasi, agnoletto2025bayesian} combines the exponential of the log quasi-likelihood~\eqref{Equ:QP_loss} with a prior distribution \( \pi(\beta) \) over regression coefficients \( \beta \), yielding
\begin{align}
    \tilde{\pi}(\beta \mid y, X) :&= \frac{\pi(\beta) \exp\left\{ n Q_n(y, X; \beta, \hat{\psi}) \right\}}{\int_{\mathbb{R}^p}\pi(\beta) \exp\left\{ n Q_n(y, X; \beta, \hat{\psi}) \right\} d\beta}. \label{Equ:Quasi-posterior}
\end{align}
\citet{greco2008robust} and \citet{ventura2016pseudo} discussed a general class of robust pseudo-posteriors, 
of which the quasi-posterior is a special case. \citet{ventura2010default} further developed 
objective Bayesian priors suitable for quasi-posterior inference.  \citet{agnoletto2025bayesian} provided a rigorous formulation of the quasi-posterior as a special case of the generalised Bayesian update framework proposed by \citet{bissiri2016general}. They show that the quasi-posterior is asymptotically Gaussian and that it is \textit{well calibrated}—in the sense that it achieves correct asymptotic frequentist coverage for parameters $\beta$—when the dispersion parameter \( \psi^{\ast} \) is estimated using~\eqref{Equ:psi_est}. 
Bayes linear methods \citep{goldstein1999bayes, astfalck2024generalised} provide an alternative approach that only requires specifying means and variances. However, the quasi-posterior provides a more familiar output that naturally extends to model and variable selection.

\subsection{Motivating examples}

We now present two examples that  motivate using the quasi-posterior for variable selection.

\begin{example}[Linear regression]
In Bayesian linear regression, the practitioner not only specifies the conditional mean \( {E}_F[y \mid x] = x^\top\beta \) (i.e., \( \mu(s) = s \)) and decide whether the data is homoskedastic (e.g., \( var_F[y \mid x] = \psi^{\ast} \), so \( V(s) = 1 \)), but they must also choose a completely known distribution for the residuals \( y - x^\top\beta \). This choice could be Gaussian, double-exponential, Student-\( t \) (with varying degrees of freedom), skew-normal, or another alternative. Prior work has shown that posterior inference over model indicators \( \gamma \) is sensitive to this specification \citep{rossell2018tractable, grunwald2017inconsistency}. Alternatively, one can conduct quasi-posterior inference for \( \beta \) using the log quasi-likelihood
\begin{align}
\ell_{\psi}(y_i; x_i, \beta) = \frac{y_i x_i^\top\beta - \frac{1}{2}(x_i^\top\beta)^2}{\psi}, 
\end{align}
(see Section \ref{Sec:Ex1_derv} for the derivation) which allows for inference on $\beta$ that is invariant to the assumed distribution of the residuals. 
\label{Ex:LinearRegression}
\end{example}


\begin{example}[Count data]
In modeling count data, practitioners typically assume that responses \( y \in \mathbb{Z}_+ \) follow either a Poisson or Negative Binomial distribution. The Poisson model assumes a rate parameter \( \lambda_i = \exp(x_i^\top\beta) \), which imposes the restrictive condition that ${E}_F[y \mid x] = var_F[y \mid x] = \exp(x_i^\top\beta)$.
When the data exhibit overdispersion (i.e., the variance exceeds the mean), a common alternative is the Negative Binomial model, which allows ${E}_F[y \mid x] = \exp(x_i^\top\beta)$, and $var_F[y \mid x] = \exp(x_i^\top\beta) + \frac{\exp(x_i^\top\beta)^2}{\theta},$
where \( \theta > 0 \) is an overdispersion parameter. However, \citet{holmes2017assigning} and \citet{agnoletto2025bayesian} show that Bayesian inference under a Poisson model can lead to posterior overconcentration when applied to overdispersed data, while the Negative Binomial model yields inconsistent mean estimates under misspecification \citep[e.g.][Sec.~18.3.1]{wooldridge2010econometric}. 
Alternatively, using \( \mu(s) = \exp(s) \) and \( V(s) = s \), the log quasi-likelihood becomes
\begin{align}
    \ell_{\psi}(y_i; x_i, \beta) = \frac{y_i x_i^\top\beta - \exp(x_i^\top \beta)}{\psi},
\end{align}
(see Section \ref{Sec:Ex2_derv} for the derivation) which enables robust inference for both over- and underdispersed count data without the structural assumptions imposed by a full likelihood. 
\label{Ex:Count}
\end{example}


While the quasi-posterior is, in principle, more general, both Examples~\ref{Ex:LinearRegression} and~\ref{Ex:Count} reduce to familiar forms known as tempered~\citep{walker2001bayesian, zhang2006information, grunwald2012safe, holmes2017assigning, syring2019calibrating, lyddon2019general} or coarsened~\citep{miller2018robust} posteriors, where the Gaussian and Poisson likelihoods, respectively, are raised to a power—often referred to as the \textit{learning rate}—in this case, \( \frac{1}{\psi} \). 
Calibrating the learning rate has recently been established as a tool to improve Bayesian inference under model misspecification with \cite{grunwald2012safe} proposing methods to set this based on minimise predictive risk, \cite{lyddon2019general} setting this to match the information in a Bayesian boostrap and \cite{syring2019calibrating} using it to calibrate the coverage of credibilty intervals. See \cite{wu2023comparison} for a details comparison. In particular, \cite{grunwald2017inconsistency} shows that the method of \cite{grunwald2012safe} helps improve the predictive risk of Bayesian ridge regression. However, the focus there is on shrinkage and prediction rather than model or variable selection.
A key distinction of the quasi-posterior 
is that it has a natural estimate of the learning rate via~\eqref{Equ:psi_est}, whose implications for posterior uncertainty quantification are well understood \citep{agnoletto2025bayesian}. 




\section{Quasi-posterior model selection}{\label{Sec:Inference}}


In this section, we formalize the quasi-posterior framework for model selection. 
Throughout, we use the following conventions: 
for model $\gamma\in\{0, 1\}^p$  define \( |\gamma|_0 = \sum_{j=1}^p \gamma_j \) as the number of active (included) predictors. Let \( \beta_{\gamma} \in \mathbb{R}^{|\gamma|_0} \) and \( X_{\gamma} \in \mathbb{R}^{n \times |\gamma|_0} \) denote the subvector and submatrix corresponding to the active coefficients and covariates, respectively. We say that model \( \gamma^{(1)} \) is nested in model \( \gamma^{(2)} \), written \( \gamma^{(1)} \subseteq \gamma^{(2)} \), if \( |\gamma^{(1)}|_0 \leq |\gamma^{(2)}|_0 \) and \( \gamma^{(1)}_j = 1 \Rightarrow \gamma^{(2)}_j = 1 \) for all \( j = 1, \ldots, p \). Finally, for variable \( j \), we denote by \( \gamma_{-j} = \gamma \setminus \{ \gamma_j \} \in \{0,1\}^{p-1} \) the variable inclusion indicators without $\gamma_j$.

\subsection{Prior specification}

We adopt a spike-and-slab prior \citep{mitchell1988bayesian, george1993variable, smith1996nonparametric}, see \cite{tadesse2021handbook} for a review, of the form
\begin{equation}
    \beta_j \mid \gamma_j \sim (1 - \gamma_j) \, \delta_0(d\beta_j) + \gamma_j \, \mathcal{N}(\beta_j; 0, s^2), \qquad \gamma_j \sim \text{Bernoulli}(w),
    \label{Equ:spike_and_slab}
\end{equation}
where \( \delta_0 \) denotes a point mass at zero (the ``spike''), and \( \mathcal{N}(0, s^2) \) is the ``slab'' — a Normal distribution with mean zero and variance \( s^2 \), allowing for nonzero coefficients. This leads to the joint prior
$\pi(\beta, \gamma) = \prod_{j=1}^p \pi(\beta_j \mid \gamma_j)\pi(\gamma_j)$
which encourages sparsity by shrinking many coefficients exactly to zero, while allowing the remaining ones to vary according to the slab prior. As part of our prior specification, we assign a Beta hyperprior \( w \sim \text{Beta}(a, b) \) to the sparsity level, which induces  a Beta-Binomial prior over model configurations \( \gamma \) \citep{scott2010bayes}.
Although we focus on spike–and–slab priors with independent Gaussian slabs, our framework is flexible and can incorporate alternative slab specifications, such as 
$g$-priors \citep{zellner1986assessing}, without requiring any modification to the inference procedure.

\subsection{The model quasi-posterior} \label{sec:model_quasi_post}

We now present the \textit{model quasi-posterior}, which serves as a generalized posterior distribution over model configurations 
\( \gamma \). Specifically, we define
\begin{equation}
\tilde{\pi}(\gamma\mid y, X) \propto \tilde{f}(y\mid X, \gamma)\pi(\gamma),
\label{Equ:model quasi-posterior}
\end{equation}
where the marginal quasi-likelihood is 
\begin{equation}
\tilde{f}(y\mid X, \gamma) := \int_{\mathbb{R}^{|\gamma|_0}} \pi(\beta_{\gamma}) \exp\left\{nQ_{n}(y, X_{\gamma}; \beta_{\gamma}, \hat{\psi})\right\} d\beta_{\gamma},
\label{Equ:Quasi-marginal-likelihood}
\end{equation}
and \( \pi(\gamma) \) denotes the Beta-Binomial prior introduced above. Here, \( \hat{\psi} \) is fixed across models and denotes the estimate obtained from \eqref{Equ:psi_est}, computed under the full model that includes all predictors.  We further define $\tilde{B}_{\gamma, \gamma^{\prime}} = {\tilde{f}(y\mid X, \gamma)}/{\tilde{f}(y\mid X, \gamma^{\prime})}$ as the \textit{quasi-Bayes-factor} for comparing model $\gamma$ with $\gamma^{\prime}$.

The model quasi-posterior \eqref{Equ:model quasi-posterior}, constructed through the marginal quasi-likelihood \eqref{Equ:Quasi-marginal-likelihood}, was {considered} in \cite{lin2006quasi}. However, the methodology was justified via a complexity penalization argument that only applied to Example \ref{Ex:LinearRegression}. {No rigorous justification for its uncertainty quantification was provided, and posterior sampling details or empirical evaluation were not included to support its practical utility} 

A key innovation 
of \eqref{Equ:model quasi-posterior} is the construction of a generalized posterior over models $\gamma$ using the marginal quasi-likelihood in \eqref{Equ:Quasi-marginal-likelihood}.  In the standard Bayesian framework, the marginal likelihood plays a central role in model and variable selection, and its use is justified by the probabilistic rules of marginalization and conditioning. In contrast, the quasi-posterior replaces the likelihood with the contributions \( \exp\left\{ \ell_{\hat{\psi}}(y_i; x_{i}, \beta) \right\} \) from \eqref{Equ:QP_loss}, which do not form a proper probability model for \( y_i \mid x_i \). As such, it is not immediately clear whether marginalization over \( \beta_{\gamma} \) remains valid in this generalized setting. To address this concern, we later show in Section~\ref{Sec:Theory} that \( \tilde{\pi}(\gamma \mid y, X) \) has the same asymptotic distribution as  \( \pi(\gamma \mid y, X) \) when the model \( f \) is correctly specified. This result establishes the theoretical validity of inference using the model quasi-posterior $\tilde{\pi}(\gamma\mid y, X)$ and constitutes one of the main contributions of this work.

\subsection{Marginal quasi-likelihood: closed-form and Laplace approximation} \label{sec:quasi_marginal}
The marginal quasi-likelihood \( \tilde{f}(y\mid X, \gamma) \) plays a key role in defining the model quasi-posterior \eqref{Equ:model quasi-posterior}, and its tractability is crucial for practical implementation. Although it generally involves an intractable integral, Proposition~\ref{Prop:LinearRegression}  below confirms that a closed-form expression is available in the case of Example \ref{Ex:LinearRegression} (see Supplementary Material, Section \ref{sec:supp_theory}, for the proof). Notably, the marginal quasi-likelihood in this case coincides with the standard marginal likelihood in a Gaussian linear model with residual variance fixed at $\hat{\psi}$.

\begin{proposition}[Closed form marginal quasi-likelihood in linear regression]
\label{Prop:LinearRegression}
    Under the specification of Example \ref{Ex:LinearRegression} and prior \eqref{Equ:spike_and_slab} the marginal quasi-likelihood is
    \begin{align}
        \tilde{f}(y\mid X, \gamma) &= \frac{\hat{\psi}^{|\gamma|_0/2}}{s^{|\gamma|_0}|U_{\gamma}|^{1/2}}\exp\left\{\frac{1}{2\hat{\psi}}m_{\gamma}^\top U_{\gamma}m_{\gamma}\right\}.
        \label{Equ:quasi-marginal-likelihood-Linearregression}
    \end{align}
    where $U_{\gamma} = \left(X_{\gamma}^\top X_{\gamma} + \frac{\hat{\psi}}{s^2}  I_{|\gamma|_0}\right)$ and $m_{\gamma} = y^\top X_{\gamma} U_{\gamma}^{-1}$.
\end{proposition} 

Outside the linear setting, however, the marginal quasi-likelihood is typically unavailable in closed form. To address this, we follow \cite{rossell2018tractable, rossell2023additive} and approximate \( \tilde{f}(y\mid X, \gamma) \)  using its Laplace approximation
\begin{align}
    \tilde{f}^{\textsc{LA}}(y\mid X, \gamma) := \left(2\pi\right)^{|\gamma|_0/2}\frac{\pi(\tilde{\beta}_{\gamma})\exp\left\{nQ_n(y, X; \tilde{\beta}_{\gamma}, \hat{\psi})\right\}}{n^{|\gamma|_0/2}\left|H_n(y, X; \tilde{\beta}_{\gamma}, \hat{\psi}) - \frac{1}{n}\nabla^2_{\beta_{\gamma}}\log \pi(\tilde{\beta}_{\gamma})\right|^{1/2}},
    \label{eq:laplace_approx}
\end{align}
where $\tilde{\beta}_{\gamma} := \argmax_{\beta_{\gamma}\in\mathbb{R}^{|\gamma|_0}} nQ_n(y, X; \beta_{\gamma}, \hat{\psi}) + \log\pi(\beta_{\gamma})$ is the maximum a posteriori (MAP) estimate for parameter $\beta_{\gamma}$ of model $\gamma$, and $H_n(y, X; \beta_{\gamma}, \hat{\psi}) := -\nabla^2_{\beta_{\gamma}} nQ_n(y, X; \beta_{\gamma}, \hat{\psi})$ is the Hessian matrix of the log quasi-likelihood. 

\section{Theoretical guarantees}{\label{Sec:Theory}}


In this section, we establish the main theoretical results of our proposed methodology. 
Specifically, 
we show $\sqrt{n}$-consistency of the quasi–log-likelihood maximiser of any model $\gamma$, and use this to prove that the model quasi-posterior shares the same asymptotic distribution as the traditional Bayesian model posterior under correct specification. 
The proofs of all of our results are deferred to the Supplementary Material, Section~\ref{sec:supp_theory}.

The regularity conditions, discussed in Section \ref{sec:regularity}, are largely standard. A key requirement 
is the correct specification of the mean and variance functions; i.e. that 
${E}_{F_0}\left[y\mid x\right] = \mu(x^\top\beta^{\ast})$ and $var_{F_0}\left[y\mid x\right] = \psi^{\ast}V(\mu(x^\top\beta^{\ast}))$. 
{This condition substantially relaxes the requirement of full likelihood specification. A correctly specified full likelihood requires correct specification of the mean and variance functions, along with infinitely many additional features of the conditional distribution.} 
Further, the mean and variance functions correspond to interpretable summaries of how $x$ affects the distribution of $y$, which may have been well studied in prior work or can be reliably elicited from experts \citep[e.g.][]{o2004probability} and assessed visually.  For example, \cite{agnoletto2025bayesian} propose diagnostics for assessing 
$\mu(\cdot)$ and $V(\cdot)$,  
which we extend in  the Supplementary Material, Sections~\ref{sec:variance_function_diagnostics} and \ref{Sec:homo_diagnostics}, and apply in 
Section~\ref{Sec:real_data}.




\subsection{Parameter consistency}

A key step in establishing the asymptotic behaviour of the model quasi-posterior is to show that, for any model $\gamma$, the quasi log-likelihood maximiser $\hat{\beta}_{\gamma} := \argmax_{\beta_{\gamma}\in\mathbb{R}^{|\gamma|_0}} Q_n(y, X; \beta_{\gamma}, \psi)$ achieves $O_p(1/\sqrt{n})$ consistency for its population parameter $\beta^{\ast}_{\gamma} := \argmax_{\beta_{\gamma}\in\mathbb{R}^{|\gamma|_0}} Q(F_0; \beta_{\gamma}, \psi)$, where $Q(F_0; \beta_{\gamma}, \psi) = \lim_{n\rightarrow\infty}\frac{1}{n}\sum_{i=1}^nQ(F_0; x_{i\gamma}, \beta_{\gamma}, \psi)$, and $Q(F_0; x_{i\gamma}, \beta_{\gamma}, \psi):= {E}_{F_0}\left[\ell_{\psi}(y; x_{i\gamma}, \beta)\right]$.
\cite{mccullagh1983quasi} established  $O_p(1/\sqrt{n})$ consistency 
of the full vector $\hat{\beta}$, 
without considering the model–specific subvector $\beta_{\gamma}$ relevant for variable selection. This result extends trivially to any 
over-specified model $\gamma$ with $\gamma^{\ast} \subseteq \gamma$, where $\gamma^{\ast}$ denotes the model associated with generating parameter $\beta^{\ast}$. However, additional work is 
required for under-specified models with $\gamma^{\ast} \not\subseteq \gamma$ as these are by definition misspecified i.e. $\beta^{\ast}_{\gamma}\neq \beta^{\ast}_{\gamma^{\ast}}$. We apply the general conditions of Theorem 5 of \cite{miller2021asymptotic} to prove $o_p(1)$ consistency in Theorem \ref{Thm:ConsistencyMiller} and use Theorem 5.23 of \cite{van2000asymptotic} to extend this to $O_p(1/\sqrt{n})$ consistency in Theorem \ref{Thm:ConsistencyRR23_Osqrtn} below.

\begin{theorem}[Consistency]
    Assume A1.1 and A2.1-6 and fix $\gamma\in \{0, 1\}^p$. Then $||\hat{\beta}_{\gamma} -\beta_{\gamma}^{\ast}||_2 = O_p(1/\sqrt{n})$, where $||\cdot||_2$ is the $L_2$-norm.
    \label{Thm:ConsistencyRR23_Osqrtn}
\end{theorem} 
Proposition \ref{Prop:psi_hat_Osqrtn} further proves that the estimated dispersion parameter 
$\hat{\psi} = \psi^{\ast} + O_p(n^{-1/2})$.



\subsection{Model selection}

We now derive the main theoretical contribution of the paper. Theorem \ref{thm:variable_selection_quasi_posterior_RR23} establishes the asymptotic behaviour of the Laplace approximation to the quasi-Bayes factor. 

\begin{theorem}[Asymptotic behaviour of variable selection quasi-posterior]
\label{thm:variable_selection_quasi_posterior_RR23}
Assume A1 and A2 and fix $\gamma\in \{0, 1\}^p$. Then
\begin{enumerate}
    \item[(i)] If $\gamma^{\ast}\not\subset \gamma$, then
    \begin{equation}
        \log(\tilde{B}^{\textsc{LA}}_{\gamma\gamma^{\ast}}) := \log\frac{\tilde{f}^{\textsc{LA}}(y\mid X, \gamma)}{\tilde{f}^{\textsc{LA}}(y\mid X, \gamma^{\ast})} = n\left\{Q(F_0;\beta_{\gamma}, \psi^{\ast}) - Q(F_0;\beta_{\gamma^{\ast}}, \psi^{\ast})\right\} + O_p(\sqrt{n}), \nonumber
    \end{equation}
    with $\left\{Q(F_0;\beta_{\gamma}, \psi^{\ast}) - Q(F_0;\beta_{\gamma^{\ast}}, \psi^{\ast})\right\} < 0$.
    \item[(ii)] If $\gamma^{\ast}\subset \gamma$, then
    \begin{align}
        \log(\tilde{B}^{\textsc{LA}}_{\gamma\gamma^{\ast}}) := \log\frac{\tilde{f}^{\textsc{LA}}(y\mid X, \gamma)}{\tilde{f}^{\textsc{LA}}(y\mid X, \gamma^{\ast})}  =&  \, \,  C(\gamma, \gamma^{\ast}) + U_n + o_p(1), 
        \label{Equ:log_B_nested}
    \end{align}
    where $\,\,\, U_n\sim \frac{1}{2}\chi^2_{|\gamma|_0 - |\gamma^{\ast}|_0}$ and
    \begin{align*}
    C(\gamma, \gamma^{\ast}) := \frac{1}{2}(|\gamma|_0 - |\gamma^{\ast}|_0)\log\left(\frac{2\pi}{n}\right) + \log\left(\frac{\pi(\beta^{\ast}_{\gamma}\mid \gamma)}{\pi(\beta^{\ast}_{\gamma^{\ast}}\mid \gamma^{\ast})}\right)+ \frac{1}{2}\log\left(\frac{H(F_0; \beta^{\ast}_{\gamma^{\ast}}, \psi^{\ast}))}{H(F_0; \beta^{\ast}_{\gamma}, \psi^{\ast})}\right).
\end{align*}
 
\end{enumerate}
\end{theorem}

 Theorem \ref{thm:variable_selection_quasi_posterior_RR23} establishes the validity of the model quasi-posterior for variable selection by showing that Laplace approximated \textit{quasi-Bayes-factor}, $\tilde{B}^{\textsc{LA}}_{\gamma\gamma^{\ast}}$, exhibits the same asymptotic behaviour as the standard Bayes factor under correct model specification \citep{dawid2011posterior}. In particular, (i) shows that for any under-specified model $\gamma$ with $\gamma^{\ast} \not\subset \gamma$—that is, when $\gamma$ omits at least one truly active variable—the quasi-Bayes factor $\tilde{B}^{\textsc{LA}}_{\gamma\gamma^{\ast}}$ converges to zero at an exponential rate in $n$. Further, and perhaps more importantly, (ii) reveals that when a model $\gamma$ is over-specified ($\gamma^{\ast} \subset \gamma$), meaning it contains all active variables plus at least one inactive one, the quasi-Bayes factor $\tilde{B}^{\textsc{LA}}_{\gamma\gamma^{\ast}}$ reflects a trade-off between a complexity penalty $C(\gamma,\gamma^{\ast})$, which diverges to $-\infty$ as $n$ increases, and a strictly positive $\chi^2$  random variable that is $O_p(1)$. In this sense, the prior specification has the same asymptotic effect on the model quasi-posterior as it does on the standard model-selection posterior. Not only does this guarantee consistent model selection among nested models, but it also indicates that the model quasi-posterior inherits the desirable finite-sample performance of standard Bayesian variable selection for correctly specified models.

While replacing $\hat{\psi}$ with $\psi^{\ast}$ in the log quasi-likelihood \eqref{Equ:psi_est} is sufficient to obtain the asymptotic $\chi^2$ distribution in \eqref{Equ:log_B_nested} (see Theorem \ref{thm:chi_squared}), the accuracy of $\hat{\psi}$ will affect the finite sample behaviour of $\log(\tilde{B}^{\textsc{LA}}_{\gamma\gamma^{\ast}})$. Because $1/\hat{\psi}$ acts as a multiplicative factor on the difference in log quasi-likelihoods, underestimation of $\hat{\psi}$ leads to unnecessarily complex models receiving greater quasi-posterior mass while overestimation will increase the posterior mass of models that are simpler than $\gamma^{\ast}$. 
We discuss this further in Section \ref{sec:simul_studies}.

Lastly, we are not the first to use $\chi^2$ asymptotics to justify the validity of uncertainty quantification in generalised Bayesian frameworks. A salient example is \cite{ribatet2012bayesian}, where the composite likelihood is adjusted so that its difference is asymptotically $\chi^2$.





\subsection{Asymptotic accuracy of the Laplace approximation} \label{sec:accuracy_laplace}

Finally, we examine the implications of using Laplace approximations as the basis of our inferential strategy.
We rely on the Laplace-approximated model quasi-posterior because it is computationally efficient and, crucially, Theorem \ref{thm:variable_selection_quasi_posterior_RR23} characterises the behaviour of the criteria we actually use (see also \cite{rossell2018tractable, rossell2023additive} for related arguments). In any case, under the conditions of Theorem~\ref{thm:variable_selection_quasi_posterior_RR23}, Proposition~\ref{Prop:Lapalce_accuracy} formalises that the Laplace approximation is asymptotically equivalent to the marginal quasi-likelihood. 


\begin{proposition}[Accuracy of Laplace Approximations]
\label{Prop:Lapalce_accuracy}
Assume A1 and A2.1-5 and fix $\gamma\in \{0, 1\}^p$. Then, as $n\rightarrow\infty$
\begin{align*}
    \tilde{f}(y\mid X, \gamma, \hat{\psi}) \times\frac{\left|H(F_0, \beta^{\ast}_{\gamma}, \hat{\psi})\right|^{1/2}}{\pi(\beta_{\gamma}^{\ast}\mid \gamma)\exp\left\{nQ_n(y, X, \hat{\beta}_{\gamma}, \hat{\psi})\right\}}\left(\frac{2\pi}{n}\right)^{|\gamma|_0}\rightarrow 1.
\end{align*}
\end{proposition}
Both the theory and implementation of the quasi-posterior variable selection benefit from concavity of the log quasi-likelihood. Concavity ensures efficient computation of $\tilde{\beta}_{\gamma}$ in \eqref{eq:laplace_approx} and guarantees Condition A2.5. Proposition \ref{Prop:Concavity} in the Supplementary Material provides conditions under which the log quasi-likelihood is concave.



\section{Posterior inference} \label{sec:posterior_inference}

This section describes our posterior inference strategy, which consists of exploring the model space, selecting active variables, and sampling the corresponding regression coefficients. A closed-form expression for the model posterior \eqref{Equ:model quasi-posterior} would require computing a normalizing constant obtained by summing over all $2^p$ models, reflecting every possible inclusion–exclusion configuration of the $p$ predictors. Such a normalisation is infeasible for even moderate values of $p$.
Instead, we use Markov chain Monte Carlo (MCMC) methods to explore the model space by drawing samples from $\tilde{\pi}(\gamma\mid y, X)$.

\subsection{Sampling from the model quasi-posterior}{\label{Sec:GibbsModelPosterior}}

We employ a \textit{random-scan} Gibbs sampler to explore 
\( \tilde{\pi}(\gamma \mid y, X) \), alternating between randomly updating individual inclusion indicators 
\( \gamma_j \) 
and sampling the sparsity level \( w \) from their full conditional posterior distributions. The full MCMC procedure is summarized in Algorithm~\ref{Alg:Gibbs_BetaBinom}. At each iteration, a variable index \( j \in \{1, \ldots, p\} \) is selected uniformly at random and updated according to \( \tilde{\pi}(\cdot \mid y, X, \gamma_{-j}, w) \). 
In particular, since \( \gamma_j \in \{0, 1\} \), we sample \( \gamma_j \mid y, X, \gamma_{-j}, w \) from a Bernoulli distribution with inclusion probability
\begin{align}
    P(\gamma_{j} = 1\mid y, X, \gamma_{-j}, w) &= \frac{\tilde{\pi}(\gamma_j = 1, \gamma_{-j}\mid y, X, w)}{\tilde{\pi}(\gamma_j = 1, \gamma_{-j}\mid y, X, w) + \tilde{\pi}(\gamma_j = 0, \gamma_{-j}\mid y, X, w)} \notag \\
    &= \left(1 + \tilde{B}_{\gamma^{-}\gamma^{+}}\frac{(1-w)}{w}\right)^{-1}
\end{align}
Here, $\gamma^{-} = \{\gamma_1,\allowbreak\ \ldots,\allowbreak\ \gamma_{j-1},\allowbreak\ 0,\allowbreak\ 
\gamma_{j+1},\allowbreak\ \ldots,\allowbreak\ \gamma_p\}$, and $\gamma^{+} = \{\gamma_1, \ldots, \gamma_{j-1}, 1, \gamma_{j+1}, \ldots, \gamma_p\}$, denote the models obtained by setting the $j$-th inclusion indicator to 0 or 1, respectively, while keeping all other components of $\gamma$ fixed. 

\begin{algorithm}
\small
\caption{Gibbs sampling algorithm for quasi-posterior model selection}\label{Alg:Gibbs_BetaBinom}
\KwIn{Data $y$ and $X$, dispersion $\hat{\psi}$, prior hyperparameters $\{a, b, s\}$, initial state $\{\gamma^{(0)}, w^{(0)}\}$, and number of MCMC samples $N$.}
\KwOut{$\{\gamma^{(t)}, w^{(t)}\}_{t=1}^N$, approximate samples from $\tilde{\pi}(\gamma\mid y, X)$.} 
\For{$t =1, \ldots, N$}{
    Generate $l=(l_1,\ldots,l_p)$ a random permutation of $\{1,\ldots,p\}$. Set $\gamma = \gamma^{(t-1)}$. \\
    \For{$j =1, \ldots, p$}{
        Set $\gamma^{+}_{-l_j} = \gamma_{-l_j}$, $\gamma^{+}_{l_j} = 1$, $\gamma^{-}_{-l_j} = \gamma_{-l_j}$, and $\gamma^{-}_{l_j} = 1$.\\
        Set $\gamma_{l_j} = 1$ with probability $\left(1 + \tilde{B}_{\gamma^{-}\gamma^{+}}\frac{(1-w^{(t-1)})}{w^{(t-1)}}\right)^{-1}$ or $\gamma_{l_j} = 0$ otherwise. \\
    }
Set $\gamma^{(t)} = \gamma$.\\
Sample $w^{(t)}\sim \textrm{Beta}\left(a + |\gamma^{(t)}|_0, b + p - |\gamma^{(t)}|_0\right)$
}
\end{algorithm}


As discussed in Section~\ref{sec:quasi_marginal}, when the marginal quasi-likelihood 
is not available in closed form, we compute 
$\tilde{B}^{\textsc{LA}}_{\gamma^{-}\gamma^{+}}$ using the Laplace approximation in 
\eqref{eq:laplace_approx}. This requires an optimization step to obtain the mode 
$\tilde{\beta}_{\gamma}$. In our implementation, we use the L-BFGS optimizer in 
\texttt{Stan} \citep{carpenter2017stan} to obtain $\tilde{\beta}_{\gamma}$ and rely on 
\texttt{Stan}’s automatic differentiation to evaluate
$H_n(y, X; \tilde{\beta}_{\gamma}, \hat{\psi})
    - n^{-1}\nabla^2_{\beta_{\gamma}} \log \pi(\tilde{\beta}_{\gamma})$.
Although bespoke optimization routines and analytic Hessian calculations could yield  faster inference, we prioritise the usability  and generality of a probabilistic-programming implementation. We implemented a caching scheme in which every visited model  $\gamma$ is stored together with its mode $\tilde{\beta}_{\gamma}$, Hessian evaluation, and  Laplace-approximated marginal quasi-likelihood and so if this model is revisited the required quantities are retrieved from memory rather than recomputed. 
This substantially  reduces the number of optimisation steps required—especially when the posterior concentrates on a  relatively small subset of models, as seen in our empirical studies.

Our implementation of Algorithm~\ref{Alg:Gibbs_BetaBinom} not only generates samples
${\gamma^{(t)}}_{t=1}^N$, but also keeps track of the average probability that
$\gamma_j = 1$ across iterations, providing a Rao–Blackwellised estimate of the posterior
inclusion probability for variable $j$ \citep{gelfand1990sampling}. We diagnose convergence of the sampler by monitoring the cumulative quasi-posterior inclusion probabilities for each variable across iterations (See Figures \ref{fig:NMES_diags} and \ref{fig:DLD_diags}). If a practitioner prefers to fix the sparsity level 
$w$ in advance, the corresponding version of the sampler is provided in the Supplementary Material (Algorithm~\ref{Alg:Gibbs}).  

 \subsection{Variable selection} \label{sec:variable_selection}
 
 To perform variable selection, the marginal quasi-posterior probability of inclusion (qPPI), \( \tilde{\pi}(\gamma_j = 1 \mid y, X) \), is thresholded to determine whether a variable is considered active. A common strategy is the so-called median probability model, which, 
 corresponds to thresholding the qPPI at 0.5 \citep{barbieri2004optimal}. Alternatively, when controlling for false positives is a priority, one can set the threshold such that the average inclusion probability among selected variables is at least \( 1 - \alpha \), achieving a target posterior expected false discovery rate (FDR) level \( \alpha \) \citep{mueller:2004} (see Suppelementary Material, Section~\ref{sec:FDR-control} for further discussion). For a broader treatment of thresholding posterior distributions and their connection to FDR control, see \citet{castillo2020spike}.  

\subsection{Sampling of the regression coefficients}

For any selected model configuration $\tilde{\gamma}$, 
we then sample the corresponding 
regression coefficients from $\tilde{\pi}(\beta_{\tilde{\gamma}} \mid y, X, \tilde{\gamma})$, 
which has the same quasi-posterior form as \eqref{Equ:Quasi-posterior}. This is achieved 
using the No-U-Turn sampler (NUTS; \citealt{hoffman2014no}) implemented in \texttt{Stan}. 
NUTS automatically tunes the step size of the Hamiltonian Monte Carlo (HMC; 
\citealt{duane1987hybrid}) integrator, which discretises the Hamiltonian dynamics to 
propose joint parameter updates and accepts or rejects these proposals using the usual 
Metropolis--Hastings criterion. Finally, for each draw we set $\beta_j = 0$ whenever 
$\tilde{\gamma}_j = 0$.  
Alternatively, posterior inference on $\beta$ could be carried out via Bayesian model averaging by sampling $\beta^{(t)}\sim \tilde{\pi}(\cdot \mid y, X, \gamma^{(t)})$ at each MCMC iteration $t$. 
%
%
%
%
In either case, marginally exploring the model space before sampling the parameters conditional 
on each model, 
(i) restricts inference to 
the lower-dimensional subspace $\beta_{\gamma} \in \mathbb{R}^{|\gamma|_0}$, which yields 
substantial savings when $\beta$ is sparse, and (ii) mitigates the strong 
auto-correlation that typically arises when Gibbs sampling from the joint posterior of $\gamma$ and $\beta$ \citep[e.g.][]{george1993variable}.

\section{Simulation studies}
\label{sec:simul_studies}



In this section, we compare the empirical performance of the model quasi-posterior with standard Bayesian procedures across a range of challenging variable–selection settings and sample sizes. 

\subsection{Parameter settings and performance metrics}

Across all experiments, we run 3{,}000 MCMC iterations and discard the first 1{,}500 as burn–in, yielding estimates of the qPPIs, $\hat{\pi}_j$ for $j=1,\dots,p$. We set Beta-Binomial hyperparameters $a = b =  1$, 
corresponding to an uninformative prior on the sparsity level. The slab variance is further fixed at $s^2 = 9$.
Variable selection is performed by: (i) thresholding at $\hat{\pi}_j \geq 0.5$ (median–probability model), and (ii) a Bayesian FDR approach with target level $\alpha = 0.05$, as discussed in Section \ref{sec:variable_selection}.  Variable selection performance is evaluated using the false discovery rate (FDR), statistical power,  F1 score and the Matthews correlation coefficient (MCC) (see Section \ref{Sec:metrics} for their definition). 
The FDR is the proportion the covariates that the model claims to be active that were truly inactive, a measure of Type I error. Power is the proportion of the truly active covariates that the model finds, a measure of  1 minus the Type II error. 
The F1 score is the harmonic mean between 1-FDR and Power, penalising imbalance between the two, while the  MCC seeks to evaluate overall classification performance that remains informative when the number of active and in-active variables differs greatly, i.e. sparse settings. 

\subsection{Overdispersed count data}{\label{Sec:sim_count}}

We assess the performance of the proposed method in a sparse regression setting with overdispersed counts as introduced in Example \ref{Ex:Count}. 
The number of covariates is fixed at $p=20$ (including an intercept) and sample size varies over $n \in \{25, 35, 50, 100, 200\}$. This reflects a nontrivial scenario where $n$ is never very large relative to $p$. 
The design matrix entries are drawn as $x_{ij} \stackrel{\text{i.i.d.}}{\sim} \mathcal{N}(0,1)$ for $j=2,\dots,p$, with the first column set to one.  The true coefficient vector ${\beta}^\ast \in \mathbb{R}^p$ is sparse, with nonzero values given by $\beta^\ast_1 = 3.5$, $\beta^\ast_7 = 1.5$, $\beta^\ast_9 = -0.3$, and $\beta^\ast_{10} = 0.3$. 
For each observation $i$, the mean is defined as  $\mu_i = \exp\{{x}_i^\top {\beta}^\ast\}$, 
and responses are generated following the scheme in \citet{agnoletto2025bayesian} 
where
\[
Y_i^\ast \,\big|\, \mu_i,\psi^\ast \sim \mathrm{Gamma}\!\left(\frac{\mu_i}{\psi^\ast},\ \psi^\ast\right),
\]  
with  $\psi^\ast=5.5$, so that ${E}[Y_i^\ast\mid x_i] = \mu_i$ and $var[Y_i^\ast\mid x_i] = \mu_i \psi^{\ast}$. This setup emulates overdispersion relative to the widely applied Poisson model. The observed counts are then obtained as  $Y_i = \lfloor Y_i^\ast \rceil$,  where $\lfloor \cdot \rceil$ denotes rounding to the nearest integer.

We benchmark four variable–selection procedures: our proposed quasi–posterior (QP) approach, with mean and variance function described in Example \ref{Ex:Count}, a Bayesian Poisson regression model (Pois), a Bayesian negative binomial regression model (NB), and a lasso penalised Poisson regression model with lasso hyper-parmerter set via cross-validation (LASSO). 
For comparability, we use Laplace approximations of marginal likelihood and a Gibbs sampler analogous to Algorithm \ref{Alg:Gibbs_BetaBinom} to sample from the traditional Poisson and NB model posteriors $\pi(\gamma\mid y, X)$.

Figure~\ref{fig:countsim-metrics} (top) reports the average classification metrics, computed over 50 simulation replicates and based on thresholding to achieve a Bayesian FDR of at most 0.05. 
We report in the Supplementary Material (Figure~\ref{fig:coutsim-metrics-0.5}) the corresponding plots for the median–probability model thresholding.
In these simulations, the \emph{quasi-posterior} (QP)  consistently yields the most accurate variable-selection performance, with the largest improvements observed at small sample size.
For $n > 25$ the QP maintains power above 90\% and FDR much lower than  5\%, while the NB model exhibits power as low as 50\% and never higher than 90\% for any sample size, and the Poisson and lasso models are unable to achieve FDR of less than 20\%. As a result, the MCC and F1 metrics are uniformly larger for the QP.
The bottom of Figure~\ref{fig:countsim-metrics} compares the estimated inclusion probabilities of the QP with the Poisson and NB models respectively for each covariate in each repeat. These shows that for small samples the QP generally assigned higher probability to active covariates than the NB and lower probability to non-active covariates than the Poisson. We further provide in the Supplementary Material (Figure~\ref{fig:coutsim-metrics-0.5}) the corresponding comparisons for the remaining values of $n$.
Overall, these experiments illustrate a situation where  the quasi-posterior model selection provides valid inference, in terms of acceptably controlling false discovery rate, and out performs standard model based procedures in terms of power, F1 and MCC.

\begin{figure}[ht!]
    \centering
    
    \includegraphics[width=0.95\textwidth]{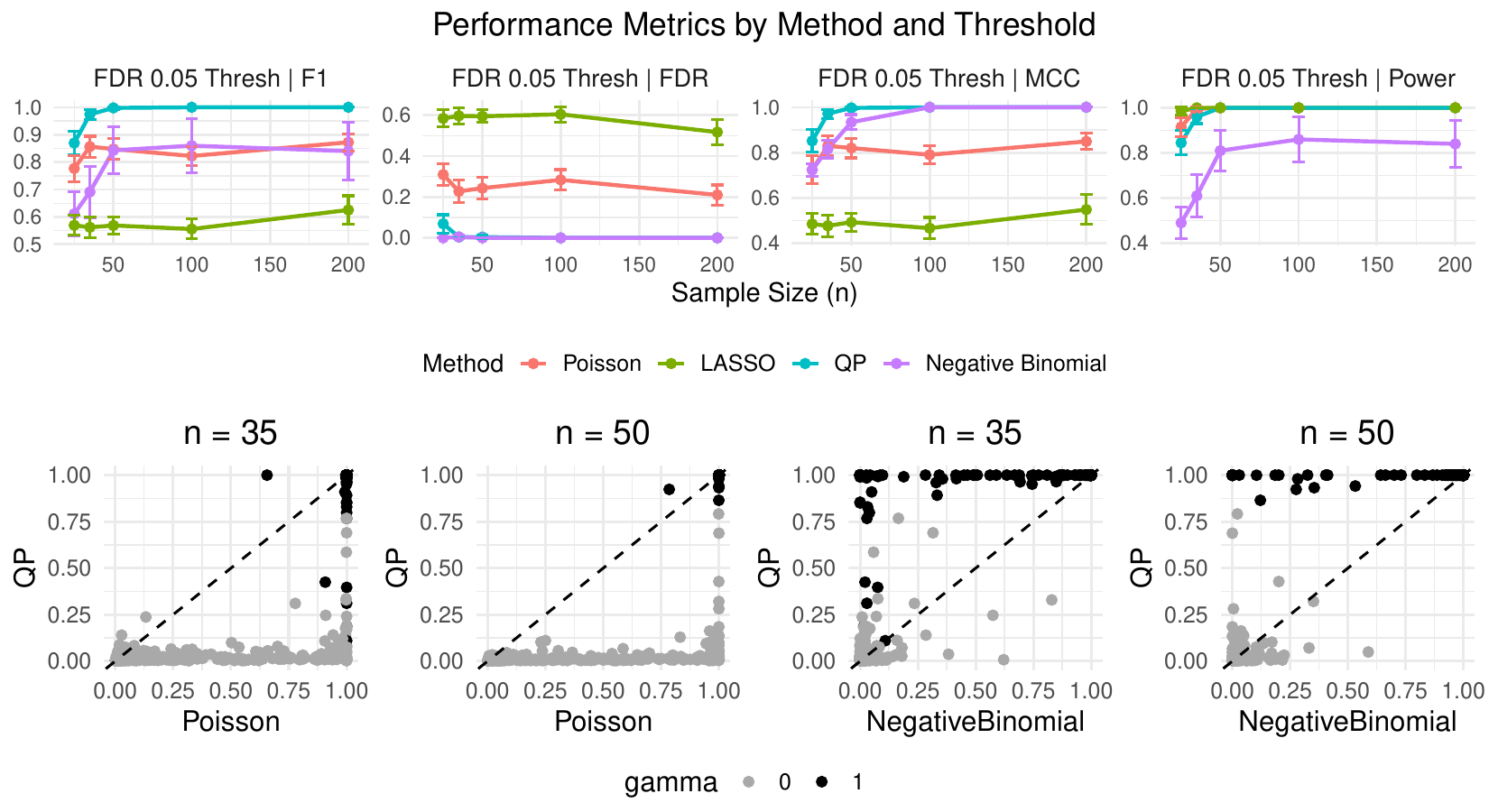}
    \caption{
    Overdispersed count data. 
Top: Variable selection performance of quasi–posterior (QP), negative binomial (NB), Poisson (Pois), and lasso (LASSO) models across sample sizes using the Bayesian FDR control $\alpha = 0.05$. Error bars denote $\pm 2$ standard errors across simulation replicates. 
Bottom: Estimated posterior probabilities of inclusion  across repeats, coloured according to whether the generating $\beta^\ast$ was zero (grey) or non-zero (black).}
    \label{fig:countsim-metrics}
\end{figure}

\subsection{Linear Regression with heavy-tailed errors}{\label{Sec:sim_heavy}}

We now consider a simulation settings aligned with the linear regression framework in Example~\ref{Ex:LinearRegression} with heavy-tailed error distributions. 
Specifically, we consider $p=20$ covariates (including an intercept) and sample sizes $n \in \{25, 50, 75, 100, 125, 150\}$.
We generate features such that $x_{ij} \stackrel{\text{i.i.d.}}{\sim} \mathcal{N}(0,1)$ for $j=2,\dots,p$, with the first column set to one.  The true coefficient vector is ${\beta}^\ast = \{0.3, 0.6, -0.6, 0.3, 0, \ldots, 0\}$. For each observation $i$, the mean is defined as  $\mu_i = {x}_i^\top {\beta}^\ast$, and responses are generated such that,  
\[
Y_i \,\big|\, \mu_i,\psi^\ast \sim \mathrm{Student-}t_{\nu}\!\left(\mu_i, {\frac{\nu-2}{\nu}\psi^\ast}\right),
\]  
with $\nu = 3$ degrees of freedom and $\psi^\ast=1$, so that ${E}[Y_i\mid x_i] = \mu_i$ and $var[Y_i\mid x_i] = \psi^{\ast}$. 

We compare the QP, with mean and variance function described in Example \ref{Ex:LinearRegression}, with traditional Bayesian variable selection in linear regression which assumes a Gaussian error distribution as implemented in the \texttt{mombf} package \citep{rossell2015package}, and lasso penalised ordinary least squares using cross-validation to set the lasso hyper-parameter. \texttt{mombf} implements a Gibbs sampler analogous to Algorithm \ref{Alg:Gibbs_BetaBinom} using the closed form linear regression marginal likelihood to sample form $\pi(\gamma\mid y, X)$. The Student-$t$ with $\nu = 3$ degrees of freedom has heavier tails than a Gaussian and therefore following \cite{rossell2018tractable} we expect traditional Bayesian variable selection to have reduced power to detect truly active variables.

Figure~\ref{fig:multcomp_probs_homo_student} (top) reports the average classification metrics over 50 simulation replicates, based on thresholding to achieve a Bayesian FDR of at most 0.05. 
Figure~\ref{fig:multcomp_probs_homo_student0.5} in the Supplementary Material presents the corresponding plots for the median–probability model thresholding.  
In these simulations both the QP and \texttt{mombf} control the FDR below 5\%, while the QP has power up to 10 percentage point larger than \texttt{mombf}  for small $n$. The QP has a uniformly larger F1 than \texttt{mombf} while the MCC of the two methods is very similar. 
The lasso approach estimates much higher power than both of the Bayesian methods when $n$ is small, but correspondingly achieves FDR above 40\% for all $n$. 
The bottom of Figure~\ref{fig:multcomp_probs_homo_student} plots the posterior inclusion probabilities estimated by the quasi-posterior for each parameter across simulation repeats against those estimated for the traditional posterior by \texttt{mombf} for several values of $n$, with points coloured according to whether the generating $\beta^\ast$ was non-zero. Across the simulations, the QP generally estimates larger probabilities of inclusion than \texttt{mombf}. When $n$ is only slightly greater than $p$ (i.e. $n = 25$) the QP estimates a larger probability of inclusion for many truly active covariates, but also some truly inactive covariates. This behaviour is likely due to the imprecise estimation of $\hat{\psi}$ when $n \approx p$. However, for $n > 25$, the QP method generally assigns larger inclusion probabilities than \texttt{mombf} only for the truly active variables, with both methods agreeing on near-zero probabilities for inactive variables.

\begin{figure}[ht!]
    \centering
    \includegraphics[width=0.95\linewidth]{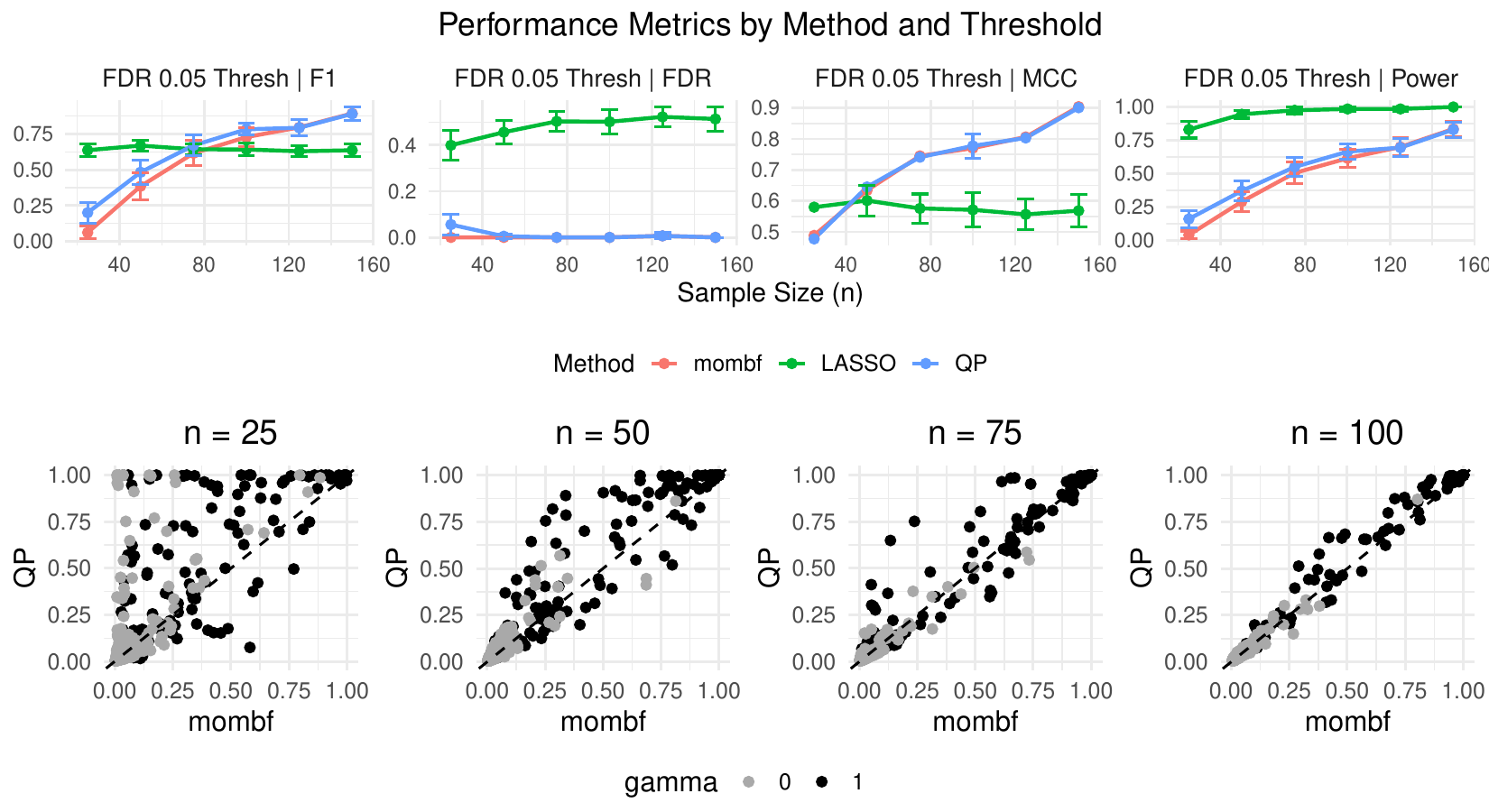}
    \caption{
    Heavy-tailed linear regression.  
Top: Variable selection performance of quasi–posterior (QP), traditional posterior (mombf) and frequentist LASSO models across sample sizes using the Bayesian FDR control $\alpha = 0.05$. Error bars denote $\pm 2$ standard errors across simulation replicates. 
Bottom: Estimated posterior probabilities of inclusion across repeats, coloured according to whether the generating $\beta^\ast$ was zero (grey) or non-zero (black).}
    \label{fig:multcomp_probs_homo_student}
\end{figure}

\subsection{Linear regression with inliers}{\label{Sec:inliers}}
While heavy-tailed data have been shown to affect Bayesian inference's power to detect truly active covariates, 
inliers have been shown to cause covariates to be included in the model that are not necessarily predictively optimal \citep{grunwald2017inconsistency}.
To examine this further, we consider the simulation setting introduced in 
\cite{grunwald2017inconsistency}. In this setting,  $p = 51$ covariates are generated with $x_{i1} = 1$ and $x_{i, 2:p} \stackrel{\text{i.i.d.}}{\sim} 0.5\mathcal{N}_{p-1}(0,I_{p-1}) + 0.5\prod_{j=2}^{p}\delta_{x_{ij} = 0}$, and the true coefficient vector is ${\beta}^{\ast} = \{0, 0.1, 0.1, 0.1, 0.1, 0, \ldots, 0\}$. For each observation $i$, the mean is defined as  $\mu_i = {x}_i^\top {\beta}^{\ast}$ and responses are generated such that,
\[
Y_i \,\big|\, \mu_i,\psi^{\ast} \sim \mathcal{N}\!\left(\mu_i, 2\psi^{\ast}\right)\left(1 - \prod_{j=2}^{p}\mathbb{I}(x_{ij} = 0)\right) + \mu_i\prod_{j=2}^{p}\mathbb{I}(x_{ij} = 0),
\]  
with $\psi^{\ast}=\frac{1}{40}$, so that ${E}[Y_i\mid x_i] = \mu_i$ and $\frac{1}{n}\sum_{i=1}^nvar[Y_i\mid x_i] = \psi^{\ast}$. 
In this design, half of the observations have Gaussian covariates with Gaussian errors, while the other half have all covariates equal to zero and the linear predictor observed without error—i.e., they are inliers. When assuming a Gaussian model, these inliers cause the posterior interval to shrink, 
generating false positives.


We compare the model quasi-posterior, using the mean and variance specifications in Example~\ref{Ex:LinearRegression}, with standard Bayesian variable selection under Gaussian errors, implemented via the \texttt{mombf} package \citep{rossell2015package}, and lasso penalised ordinary least squares using cross-validation to set the lasso hyper-parameter. 
Assumption A1.1, which requires correct specification of the quasi-posterior variance function, is violated in this setting. However, the mean function remains correctly specified, and importantly, although the true variance is not constant, it does not depend on $\beta^{\ast}$.

Figure~\ref{fig:multcomp_selestion_homo_inlier} (top) reports the average classification metrics, computed over 50 simulation replicates and based on thresholding to achieve a Bayesian FDR of at most 0.05. 
The corresponding results for the median–probability model are shown in the Supplementary Material (Figure~\ref{fig:multcomp_selestion_homo_inlier0.5}). In these simulations, for $n \geq 250$, the QP halves the FDR of the traditional Bayesian posterior (\texttt{mombf}), which remains above the 5\% level even at $n = 750$, and it does so without sacrificing power. The lasso approach has FDR greater than 60\% for all values of $n$. 
Figure~\ref{fig:multcomp_selestion_homo_inlier} (bottom) shows that this improvement results from the model quasi-posterior assigning consistently lower inclusion probabilities to truly inactive covariates than \texttt{mombf}. The corresponding results for the remaining values of $n$ are reported in the Supplementary Material (Figure~\ref{fig:multcomp_selestion_homo_inlier0.5}). 
However, for $n < 250$, while the quasi-posterior attains higher power than \texttt{mombf}, it also exhibits a higher FDR. This likely reflects difficulties in estimating $\hat{\psi}$ when the sample size is small relative to $p$ and the variation in the data is low.
To address this issue, we report in the Supplementary Material (Figure \ref{fig:multcomp_metrics_homo_inlier_psiLASSO}) the corresponding results obtained by replacing $\hat{\beta}$ in \eqref{Equ:psi_est} with a lasso-regularised estimate. 
This adjustment leads to lower FDR for small $n$, accompanied by a slight reduction in power. This presents a promising avenue for the quasi-posterior when $n$ is small relative to $p$, with a full theoretical analysis left for future work.

\begin{figure}[ht!]
    \centering
    \includegraphics[width=0.95\linewidth]{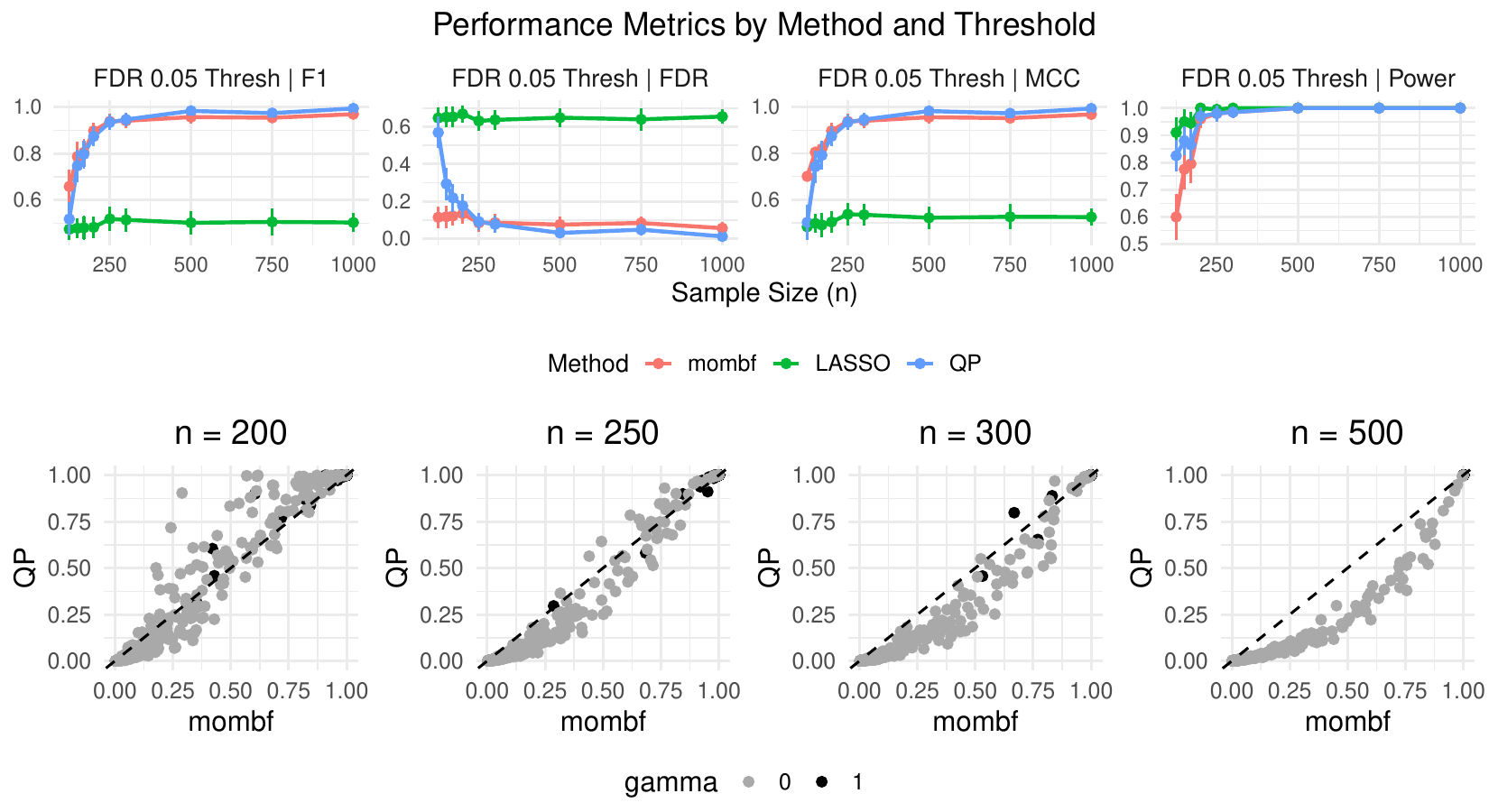}
    \caption{
    Linear regression with inliers. 
    Top: Variable selection performance of quasi–posterior (QP), traditional posterior (mombf) and frequentist lasso models across sample sizes using the Bayesian FDR control $\alpha = 0.05$. Error bars denote $\pm 2$ standard errors across simulation replicates. 
Bottom: Estimated posterior probabilities of inclusion across repeats, coloured according to whether the generating $\beta^\ast$ was zero (grey) or non-zero (black).}
    \label{fig:multcomp_selestion_homo_inlier}
\end{figure}





\section{Real data}{\label{Sec:real_data}}

We illustrate the utility of the model quasi-posterior for inference through applications in social science and genomics.

\subsection{NMES 1988 dataset}

This analysis uses data from the 1987--1988 U.S.\ National Medical Expenditure Survey (NMES), a nationally representative study of health-care utilization among older adults. 
We focus on the subsample of $n = 4{,}406$ Medicare-covered individuals aged 66 and older, a benchmark dataset in the econometrics literature on count-data regression
\citep{deb1997demand}. 
The response variable $y$ records the number of physician office visits. 
The design matrix $X$ contains $p = 16$ predictors, including measures of self-reported health, chronic conditions, functional limitations, region of residence, age, race, gender, marital status, years of education, income, employment status, and insurance coverage. 

We compare the model QP, using the mean and variance functions defined in Example \ref{Ex:Count}, with standard Poisson  and NB regression models. 
Convergence of the Gibbs sampler for estimating posterior inclusion probabilities is assessed in Figure \ref{fig:NMES_diags} in the Supplementary Material. 
The QP and NB results closely align, and the selected variables are consistent with established findings in health-care behaviour \citep{deb1997demand}: poorer self-reported health, chronic conditions, and medical insurance indicators are strongly associated with higher physician-visit counts. The Poisson model on the other hand highlights additional predictors that are not supported under QP or NB. Figure \ref{fig:NMES1988_ppi_comparison} in the supplementary material reports the corresponding posterior inclusion probabilities.

We further assess whether the proposed mean and variance assumptions provide an adequate description of this dataset. 
To this end, we follow a diagnostic similar to that of \citet{agnoletto2025bayesian}, which compares empirical mean–variance patterns with those implied by the fitted models. Full details of the diagnostic procedure and additional results are provided in the Supplementary Material (Section~\ref{sec:variance_function_diagnostics}).
The left panel of Figure \ref{fig:mean_variance-adequacy_NMES1988} compares empirical and model-implied conditional means, estimated by binning, and shows that all three models achieve similar accuracy in estimating the mean structure.
The right panel Figure \ref{fig:mean_variance-adequacy_NMES1988} compares the  mean–variance relationship present in the binned data with what is assumed by each model.  
The Poisson model substantially underestimates variability, the NB captures overdispersion but tends to overstate variance at larger fitted means, while the proposed model QP provides the closest agreement with the empirical variability across the full range of  values. These differences are reflected quantitatively by the mean square error (MSE) between empirical and model-implied variances, which equals 1570.0 for the Poisson model, 263.0 for the NB model, and 62.3 for the model QP. In the Supplementary Material, we further assess out-of-sample predictive performance via cross-validation using a variance-weighted mean-squared error, and find that the proposed model QP provides the most accurate predictions.

\begin{figure}[ht!]
  \centering
  \includegraphics[width=0.95\textwidth]{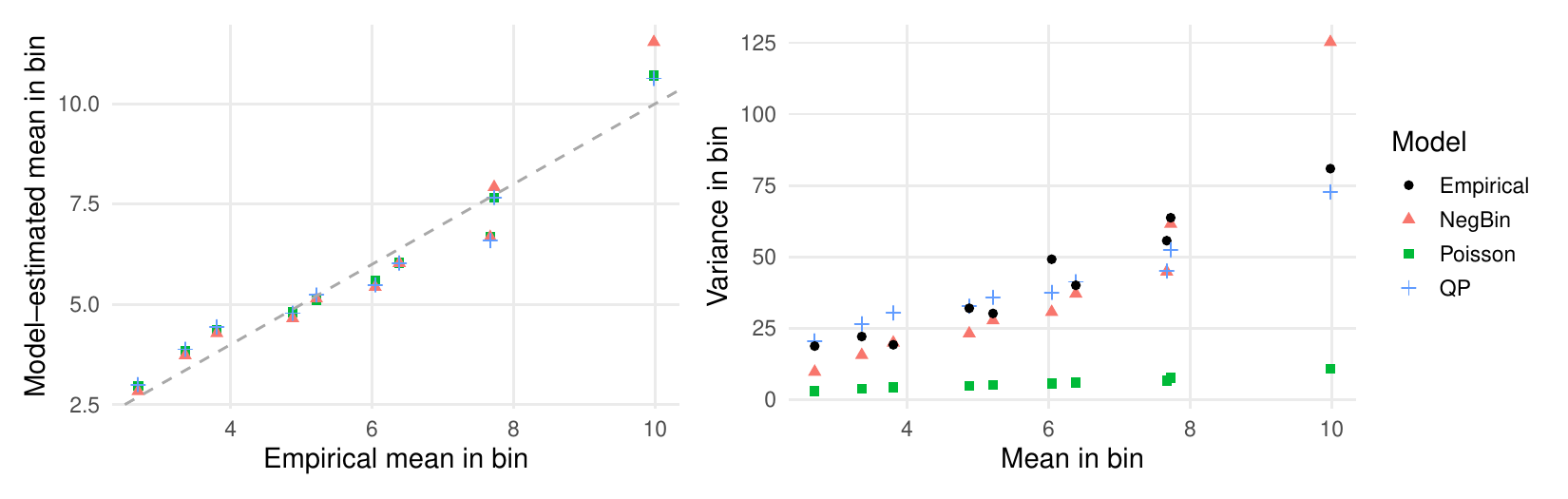}
\caption{
NMES data. (Left) Empirical vs model-implied average for $y$ within bins of the fitted linear predictor.
(Right) Empirical and model-implied mean vs variance for $y$ within bins of the fitted linear predictor.
Each model is fitted using its own selected subset of predictors.
}
\label{fig:mean_variance-adequacy_NMES1988}
\end{figure}

\subsection{DLD dataset}

We further illustrate our approach using the DLD dataset \citep{yuan2016plasma}, which profiled plasma extracellular RNA (exRNA) from $n = 192$ individuals comprising colorectal, prostate, and pancreatic cancer patients and healthy controls. The data include expression measures (reads per million, RPM) for a collection of messenger RNA genes identified as part of the exRNA species diversity analysis. Similarly to \citet{rossell2018tractable}, we study the problem of predicting the expression of \textit{DLD} (dihydrolipoamide dehydrogenase), a gene involved in mitochondrial energy metabolism and linked to several metabolic and immune-system disorders, using the $p = 55$ other genes as as explanatory variables. 

We compare the model QP, using the mean and variance functions described in Example \ref{Ex:LinearRegression}, with traditional Bayesian variable selection methods for linear regression that assume a Gaussian error distribution, as implemented in the \texttt{mombf} package \citep{rossell2015package}.  Both the quasi-posterior and  \texttt{mombf} approaches 
identify C6orf226, ECH1, CSF2RA, and FBX119 as having high posterior inclusion probabilities, while assigning negligible support to the remaining genes (see Supplementary Material for the corresponding posterior inclusion probability plots). We further evaluated the adequacy of the homoscedasticity assumption for the DLD dataset using the diagnostic introduced in the previous example and detailed in Section~\ref{Sec:homo_diagnostics} of the Supplementary Material. 

To compare variable selection performance between QP and \texttt{mombf} across different sample sizes, we conducted a subsampling analysis similar in spirit to the simulation study. Specifically, we repeatedly drew subsamples of size \(n \in \{60, 80, 100, 120, 140, 160\}\) from the full dataset and applied both QP and \texttt{mombf} to each subsample. Since the true active set is unknown for real data, we treat the variables selected using the full dataset (\(n=192\)) as a ground truth. 
The top of Figure~\ref{fig:DLD_multcomp_probs_homo} shows that in these simulations the QP maintains higher power and F1 score 
when $n < 150$, while controlling false discoveries below 5\% for $n\geq 100$. 
The bottom panels further illustrate that for small sample sizes the QP assigns larger posterior inclusion probabilities than \texttt{mombf} to genes identified as active in the full-data analysis. 




\begin{figure}[ht!]
    \centering
    \includegraphics[width=0.95\linewidth]{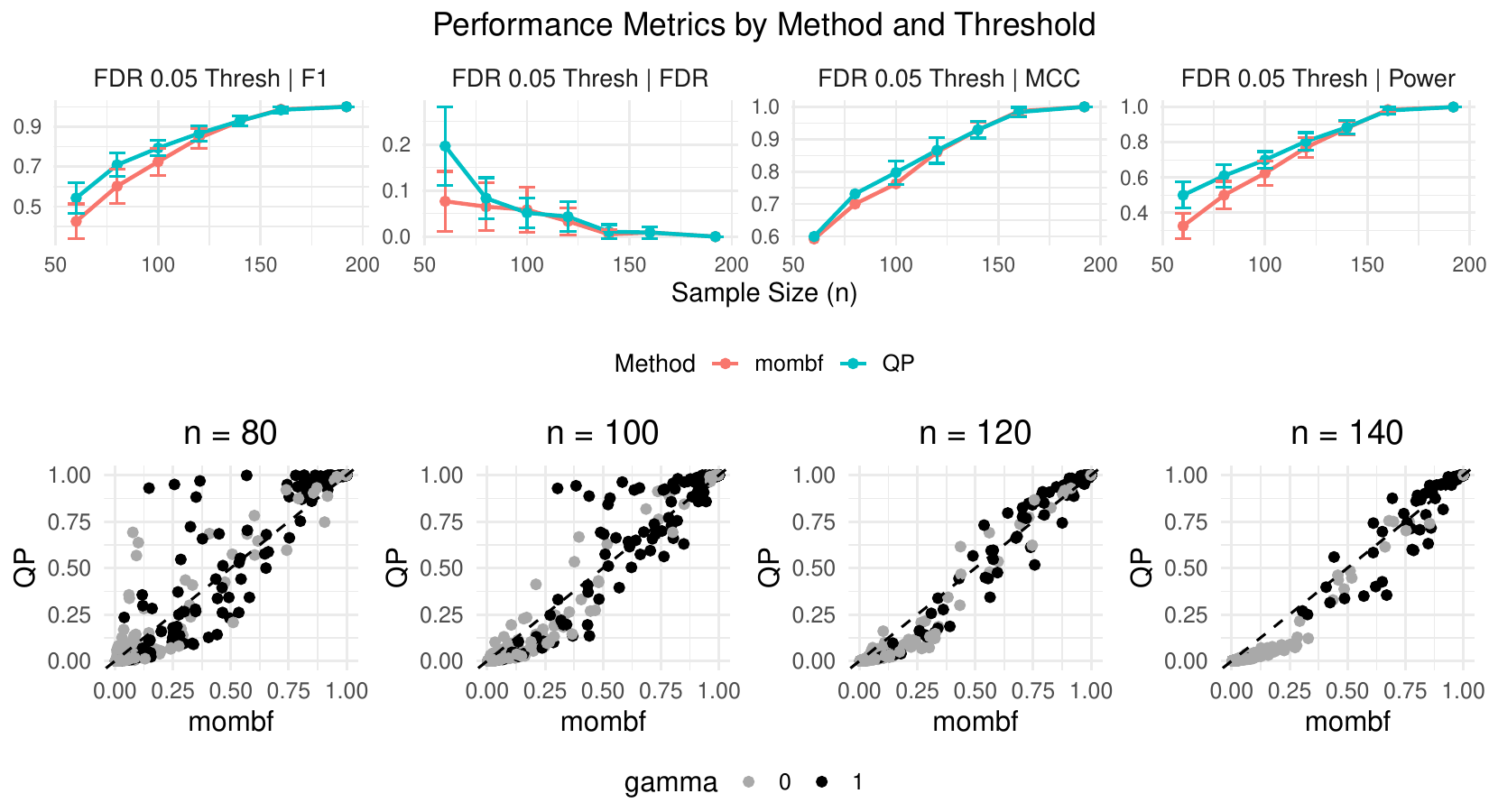}
    \caption{
    DLD data. 
Top: Variable selection performance of quasi–posterior (QP) and traditional posterior (mombf) across sub-sample sizes using the Bayesian FDR control $\alpha = 0.05$. Error bars denote $\pm 2$ standard errors across simulation replicates. 
Bottom: Estimated posterior probabilities of inclusion across repeats, coloured according to whether  $\beta$ was estimated to be zero (grey) or non-zero (black) when the full $n = 192$ data points are considered.}
    \label{fig:DLD_multcomp_probs_homo}
\end{figure}

\section{Discussion}{\label{Sec:conclusion}}


This paper establishes the validity of the model quasi-posterior
for Bayesian variable selection in generalized linear models without requiring full specification of the data-generating distribution. 
We proved that the model quasi-posterior provides valid inference, in the sense that it exhibits the same asymptotic behaviour as the standard Bayesian model posterior under a correctly specified full likelihood, while simulation studies and showed improved variable selection performance under model misspecification.
This enables practitioners to quantify model uncertainty and perform variable selection using interpretable mean and variance structures relating $y$ to $x$, while avoiding the need to fully specify a likelihood function.

To the best of our knowledge this represents the first instance of generalised Bayesian variable selection and suggests there may be a link between correct asymptotic coverage and calibrated variable selection. 
A limitation of this approach is its reliance on correct specification of the mean and variance functions. While this requirement is weaker than assuming a fully specified probability model, it may still be violated in practice. To address this, we provided diagnostic tools for assessing the adequacy of the mean and variance specifications. Further work will  investigate whether generalised Bayesian updates with loss functions with different robustness properties can be calibrated for variable selection.  

Although the present work focuses on moderate-dimensional settings, further development is needed to extend quasi-posterior variable selection to high-dimensional regimes. 
Penalized regression methods may offer a viable approach for estimating the dispersion parameter \( \hat{\psi} \) in this setting. 
Non-local priors \citep{rossell2017nonlocal} may also be used to improve quasi-Bayes factor rates. From a computational perspective, 
birth–death–swap samplers \citep{yang2016computational} or tempered Gibbs sampling \citep{zanella2019scalable} may yield improved sampling efficiency in higher dimensions.

\section*{Acknowledgements}

The authors would like to thank Rami Tabri for useful discussions during the preparation of this manuscript

\bibliography{bib.bib}

\newpage

\appendix

This supplement provides additional theoretical, methodological, and empirical details supporting the results presented in the main paper. 
Section~\ref{sec:supp_theory} contains full notation, regularity conditions, and complete proofs of the main theoretical results, together with auxiliary propositions on concavity, consistency, and the accuracy of Laplace approximations for quasi--posterior model comparison. 
Section~\ref{sec:supp_experiments} presents implementation details for Bayesian false discovery rate control and Gibbs sampling, along with supplementary simulation results. 
Finally, additional analyses for the real data applications are reported, including diagnostic checks, convergence assessments, and out--of--sample predictive comparisons.

\section{Proofs of Theoretical Results}{\label{sec:supp_theory}}

This section contains full definitions of all notation, discussion of technical conditions, preliminary theoretical results and proofs of the paper's main theorems. 

\subsection{Notation}{\label{Sec:Notation}}

First, we review the important notation defined in the main paper. Let \( X \in \mathbb{R}^{n \times p} \) be the design matrix whose rows are the covariate vectors \( x_i \). Let \( \gamma = (\gamma_1, \ldots, \gamma_p) \in \{0,1\}^p \) denote a vector of variable inclusion indicators, where \( \gamma_j = \mathbb{I}(\beta_j \neq 0) \) indicates whether the \( j \)-th predictor is active. Define \( |\gamma|_0 = \sum_{j=1}^p \gamma_j \) as the number of active (included) predictors. Let \( \beta_{\gamma} \in \mathbb{R}^{|\gamma|_0} \) and \( X_{\gamma} \in \mathbb{R}^{n \times |\gamma|_0} \) denote the subvector and submatrix corresponding to the active coefficients and covariates, respectively. We say that model \( \gamma^{(1)} \) is nested in model \( \gamma^{(2)} \), written \( \gamma^{(1)} \subseteq \gamma^{(2)} \), if \( |\gamma^{(1)}|_0 \leq |\gamma^{(2)}|_0 \) and \( \gamma^{(1)}_j = 1 \Rightarrow \gamma^{(2)}_j = 1 \) for all \( j = 1, \ldots, p \). Finally, for variable \( j \), we denote by \( \gamma_{-j} = \gamma \setminus \{ \gamma_j \} \in \{0,1\}^{p-1} \) the variable inclusion indicators without $\gamma_j$. 

 Further, for any model $\gamma \in \{0, 1\}^p$, define 
\begin{align*}
	Q_{n}(y, X_{\gamma}; \beta_{\gamma}, \psi) :&= \frac{1}{n}\sum_{i=1}^n \ell_{\psi}(y_i; x_{i\gamma}, \beta_{\gamma}),\\
    Q(F_0; x_{\gamma}, \beta_{\gamma}, \psi) :&= {E}_{F_0}\left[\ell_{\psi}(y; x_{\gamma}, \beta)\right]\\
    Q(F_0; \beta_{\gamma}, \psi) :&= \lim_{n\rightarrow \infty}\frac{1}{n}\sum_{i=1}^nQ(F_0; x_{i\gamma}, \beta_{\gamma}, \psi),
\end{align*}
where \( x_{i\gamma} \) denotes the \( i \)-th row of the sub matrix matrix \( X_\gamma \). By construction, $\beta^{\ast} := \beta^{\ast}_{\{1, \ldots, 1\}}$ is such that $\beta^{\ast} = \argmax_{\beta\in \mathbb{R}^p} Q(F_0; \beta, \psi)$ for any $\psi > 0$ (see Theorem 1 of \cite{agnoletto2025bayesian}). Whenever $\beta$, $\beta^{\ast}$ or $\hat{\beta}$ appears without a $\gamma$ subscript this corresponds to the full $p$-dimensional vector. Let us then further define
\begin{align*}
    \beta^{\ast}_{\gamma} :&= \argmax_{\beta_{\gamma}\in \mathbb{R}^{|\gamma|_0}} Q(F_0; \beta_{\gamma}, \psi),\\
    \hat{\beta}_{\gamma} :&= \argmax_{\beta_{\gamma}\in\mathbb{R}^{|\gamma|_0}} Q_{n}(y, X_{\gamma}; \beta_{\gamma}, \psi),\\
    \tilde{\beta}_{\gamma} :&= \argmax_{\beta_{\gamma}\in\mathbb{R}^{|\gamma|_0}} nQ_{n}(y, X_{\gamma}; \beta_{\gamma}, \psi) + \log\pi(\beta_{\gamma}\mid \gamma),
\end{align*}
and note that $\beta^{\ast}_{\gamma}$ and $\hat{\beta}_{\gamma}$ are both invariant to the value of $\psi>0$ while $\tilde{\beta}_{\gamma}$ is not. We then define
\begin{align*}
\hat{\psi} &:= \frac{1}{n-p}\sum_{i=1}^n\frac{\left(y_i - \mu(x_i^\top\hat{\beta})\right)^2}{V(\mu(x_i^\top\hat{\beta}))},\\
\psi^{\ast} &:= {E}_{F_0}\left[\frac{\left(y - \mu(x^\top\beta^{\ast})\right)^2}{V(\mu(x^\top\beta^{\ast}))}\right].
\end{align*}
Finally,
\begin{align*}
H_n(y; X, \beta, \psi)
:&= - \nabla_\beta^2 Q_n(y, X; \beta, \psi)
= -\frac{1}{n}\sum_{i=1}^n \nabla_\beta^2 \ell_\psi(y_i; x_i, \beta)\\
H(F_0; \beta, \psi) :&= \lim_{n\rightarrow \infty}\frac{1}{n}\sum_{i=1}^n{E}_{F_0}\left[-\nabla^2_{\beta}\ell_{\psi}(y; x_i, \beta) \right],
\end{align*}
are the observed and limiting expected Hessian matrices of the quasi-log-likelihood. 

\subsection{Regularity Conditions}{\label{sec:regularity}}

The paper's theoretical conditions require the following technical conditions.

\begin{assumption}[Prior and Data Regularity] The following is true about the data generating process and the prior used
\begin{enumerate}
    \item[A1.1] Data $y_i\mid x_i$, $i = 1,\ldots, n$, are independently distributed according to $F_0(dy_i\mid x_i)$  with ${E}_{F_0}\left[y_i\mid x_i\right] = \mu(x_i^\top\beta^{\ast})$ and $var_{F_0}\left[y_i\mid x_i\right] = \psi^{\ast}V\left(\mu(x_i^\top\beta^{\ast})\right)$
    \item[A1.2] Prior $\pi(\beta_{\gamma}\mid \gamma)$ is continuous at $\beta^{\ast}_{\gamma}$ and such that $\pi(\beta^{\ast}_{\gamma}) > 0$ for all $\gamma$.
\end{enumerate}
\end{assumption}

A1.1 states that the data are independent given their features and that the adopted quasi-likelihood has correctly specified the mean and variance of the data generating process $F_0$. Although restrictive, we note that this assumption is significantly weaker than a requirement to produce a full likelihood model that correctly captures $F_0$. 
In Sections \ref{sec:variance_function_diagnostics} and \ref{Sec:homo_diagnostics} we extend and deploy the diagnostics of \cite{agnoletto2025bayesian} to help verify this condition. 
A1.2 states that prior used is continuous and has positive density at the data generating parameters which is immediately satisfied under \eqref{Equ:spike_and_slab} provided $w\in (0, 1)$.


\begin{assumption}[Quasi-likelihood Regularity]
Let $B_{\gamma}\subseteq \Gamma_{\gamma} \subseteq \mathbb{R}^{|\gamma|_0}$ be an open and convex set such that for all $\psi > 0$
\begin{enumerate}
    \item[A2.1] $\beta^{\ast}_\gamma \in B_{\gamma}$, 
    \item[A2.2]  $Q_{n}(y, X_{\gamma}; \beta_{\gamma}, \psi)$ has continuous third derivatives on $B_{\gamma}$ and that the third derivative is uniformly bounded on $B_{\gamma}$
    \item[A2.3] $Q(F_0; \beta_{\gamma}, \psi)<\infty$ and $\frac{1}{n^2}\sum_{i=1}^{\infty}var_{F_0}\left[\ell_{\psi}(y_i; x_{i\gamma}, \beta_{\gamma})\right] < \infty$

    \item[A2.4] Matrix $H(F_0; \beta_{\gamma}^{\ast}, \psi)$ exists and is positive definite
    \item[A2.5] Either of the following two conditions must be satisfied
    \begin{itemize}
        \item[i)]  $Q(F_0; \beta_{\gamma}, \psi) < Q(F_0; \beta^{\ast}_{\gamma}, \psi)$ for all $\beta_{\gamma} \in B_{\gamma}^{\ast}\setminus\{\beta^{\ast}_{\gamma}\}$ and \\$\lim \sup_n \sup_{\beta_{\gamma}\in\mathbb{R}^|\gamma|_0\setminus B_{\gamma}^{\ast}} Q_{n}(y, X_{\gamma}; \beta_{\gamma}, \psi) < Q(F_0; \beta^{\ast}_{\gamma}, \psi)$ for some compact set $B_{\gamma}^{\ast}\subseteq B_{\gamma}$ with $\beta^{\ast}_{\gamma}$ in the interior of $B^{\ast}_{\gamma}$. 
        \item[ii)] $Q_{n}(y, X; \beta, \psi)$ is concave and $\nabla_{\beta}Q(F_0; \beta^{\ast}_{\gamma}, \psi) = 0$
    \end{itemize}
    \item[A2.6] $Q_{n}(y, X_{\gamma}; \beta_{\gamma}, \psi)$ pointwise bounded on $B_{\gamma}$
    \item[A2.7] $\frac{1}{n}\sum_{i=1}^n var_{F_0}\left[\ell_{\psi^{\ast}}(y_i; x_{i\gamma}, \beta^{\ast}_{\gamma})\right] < \infty$ and $\frac{1}{n}\sum_{i=1}^n var_{F_0}\left[\frac{\left(y_i - \mu(x_i^\top\beta^{\ast})\right)^2}{V(\mu(x_i^\top\beta^{\ast}))}\right] < \infty$
\end{enumerate}
    
\end{assumption}

A2.1-5 appear in \cite{agnoletto2025bayesian}.
They argue that A2.2 is satisfied for  any reasonable choice of the mean $\mu(\cdot)$ and variance $V(\cdot)$ functions including Examples \ref{Ex:LinearRegression} and \ref{Ex:Count}. A2.3 requires that the quasi-log-likelihood is well defined at the limit. This was shown by \cite{ccinlar2011probability} 
Proposition 6.13 (Chapter 3) to be sufficient for pointwise convergence of $Q_{n}(y, X; \beta, \psi)$ to $Q(F_0; \beta, \psi)$ as required by Theorem 5 of \cite{miller2021asymptotic}. A2.4 is a standard assumption that allows for asymptotic Taylor expansions. Proposition \ref{Prop:Concavity} shows that A2.4 is straightforward to show for overparametrised models $\gamma^{\ast}\subseteq \gamma$. A2.5 (i) is a Wald-type condition and means that $Q(F_0; \beta^{\ast}, \psi)$ is maximised at $\beta^{\ast}$. Condition (ii) is stronger, but simplifies some of the argument. A2.5 (ii) implies (i) and is shown to be satisfied in Examples \ref{Ex:LinearRegression} and \ref{Ex:Count} in Corollary \ref{Cor:ConcavityCount}.

A2.6 and A2.7 are additional conditions to those imposed by \cite{agnoletto2025bayesian} that are required 
when analysing the ratio of Laplace approximations to the (quasi) marginal likelihoods. A2.6 is required in Theorem 7 of \cite{miller2021asymptotic} which established that $Q_n$ is equi-Lipschitz which itself is a condition of Theorem 5.23 of \cite{van2000asymptotic} when proving $\sqrt{n}$ consistency for $\hat{\beta}_{\gamma}$ in Theorem \ref{Thm:ConsistencyRR23_Osqrtn}. 
A2.7 is needed for $\sqrt{n}$-consistency of $Q_n$ and $\hat{\psi}$ which are both required for Theorem \ref{thm:variable_selection_quasi_posterior_RR23}. 




\subsection{Proofs of preliminary results}

\subsubsection{Derivation of Example \ref{Ex:LinearRegression}}{\label{Sec:Ex1_derv}}

We derive the quasi-log-likelihood under the mean and variance specification of Example \ref{Ex:LinearRegression}.

\begin{proof}
    Let $\mu(x_i^\top\beta) = x_i^\top\beta$ and $V(t) = 1$. Then the quasi likelihood is
\begin{align*}
    \ell_{\psi}(y_i; x_i, \beta) := \int_a^{\mu(x_i^\top\beta)}\frac{y_i - t}{\psi V(t)}dt &= \int_a^{\mu(x_i^\top\beta)}\frac{y_i - t}{\psi }dt\\
    &= \left[\frac{y_i t - \frac{1}{2}t^2}{\psi }\right]_a^{\mu(x_i^\top\beta)}\\
    &= \frac{y_i \mu(x_i^\top\beta) - \frac{1}{2}\mu(x_i^\top\beta)^2}{\psi } + C\\
    &= \frac{y_i x_i^\top\beta - \frac{1}{2}{x_i^2}^\top\beta^2}{\psi } + C.
\end{align*}
\end{proof}

\subsubsection{Derivation of Example \ref{Ex:Count}}{\label{Sec:Ex2_derv}}

We derive the quasi-log-likelihood under the mean and variance specification of Example \ref{Ex:Count}.

\begin{proof}
    Let $\mu(x_i^\top\beta) = \exp\left\{x_i^\top \beta\right\}$ and $V(t) = t$. then the quasi likelihood is
\begin{align*}
    \ell_{\psi}(y_i; x_i, \beta) := \int_a^{\mu(x_i^\top\beta)}\frac{y_i - t}{\psi V(t)}dt &= \int_a^{\mu(x_i^\top\beta)}\frac{y_i}{\psi t } - \frac{1}{\psi}dt\\
    &= \left[\frac{y_i \log(t) - t}{\psi }\right]_a^{\mu(x_i^\top\beta)}\\
    &= \frac{y_i \log \mu(x_i^\top\beta) - \mu(x_i^\top\beta)}{\psi } + C\\
    &= \frac{y_i x_i^\top\beta - \exp\left\{x_i^\top \beta\right\}}{\psi } + C.
\end{align*}
\end{proof}

\subsubsection{Proof of Proposition \ref{Prop:LinearRegression}}

We prove that the quasi-marginal likelihood has the closed form provided by Proposition \ref{Prop:LinearRegression} under the mean and variance specification in Example \ref{Ex:LinearRegression}

\begin{proof}
Following the derivation of $\ell_{\psi}(y_i; x_i, \beta)$ in Example \ref{Ex:LinearRegression}
\begin{align*}
    \tilde{f}(y; X, \gamma, \psi) &= \int \exp\left\{\sum_{i=1}^n \ell_{\psi}(y_i; x_{i, \gamma}, \beta_{\gamma})\right\}\pi(\beta_{\gamma})d\beta_{\gamma}\\
    &= \int \exp\left\{\sum_{i=1}^n \frac{y_i x_{i,\gamma}^\top\beta_{\gamma} - \frac{1}{2}{x_{i,\gamma}^2}^\top\beta_{\gamma}^2}{\psi }\right\}\pi(\beta_{\gamma})d\beta_{\gamma}\\
    &= \int \exp\left\{\frac{1}{\psi}\left(y^\top X_{\gamma}\beta_{\gamma} - \frac{1}{2}\beta_{\gamma}^\top X_{\gamma}^\top X_{\gamma}\beta_{\gamma}\right)\right\}\pi(\beta_{\gamma})d\beta_{\gamma}.
\end{align*}
Using that $\pi(\beta_j|\gamma_j = 1) = \mathcal{N}(\beta_j; 0, s^2)$ we can write
\begin{align*}
    p(y; X, \gamma, \psi) &= \int \exp\left\{-\frac{1}{2\psi}\left(\beta_{\gamma}^\top X_{\gamma}^\top X_{\gamma}\beta_{\gamma} - 2y^\top X_\gamma\beta_\gamma\right)\right\}\mathcal{N}_{|\gamma|_0}(\beta_\gamma; 0, s^2I_{|\gamma|_0})d\beta_{\gamma}\\
    &=\int \frac{1}{(2\pi)^{|\gamma|_0/2}s^{|\gamma|_0}}\exp\left\{-\frac{1}{2\psi}\left(\beta_{\gamma}^\top \left(X_{\gamma}^\top X_{\gamma} + \frac{\psi}{s^2} I_{|\gamma|_0}\right)\beta_{\gamma} - 2y^\top X_{\gamma}\beta_\gamma\right)\right\}d\beta_{\gamma}
\end{align*}
Let $U_{\gamma} = \left(X_{\gamma}^\top X_{\gamma} + \frac{\psi}{s^2}  I_{|\gamma|_0}\right)$ and 
$m_{\gamma} = U_{\gamma}^{-1} X_{\gamma}^\top y$
\begin{align*}
    &p(y; X, \gamma, \psi) =\frac{1}{s^{|\gamma|_0}}\exp\left\{\frac{1}{2\psi}m_{\gamma}^\top U_{\gamma}m_{\gamma}\right\}\int \frac{1}{(2\pi)^{|\gamma|_0/2}}\exp\left\{-\frac{1}{2\psi}\left(\beta_\gamma - m_{\gamma}\right)^\top U_{\gamma}\left(\beta_\gamma - m_{\gamma}\right)\right\}d\beta_\gamma\\
    =&\frac{\psi^{|\gamma|_0/2}}{s^{|\gamma|_0}|U_{\gamma}|^{1/2}}\exp\left\{\frac{1}{2\psi}m_{\gamma}^\top U_{\gamma}m_{\gamma}\right\}\int \frac{|U_{\gamma}|^{1/2}}{(2\pi)^{|\gamma|_0/2}\psi^{|\gamma|_0/2}}\exp\left\{-\frac{1}{2\psi}\left(\beta_\gamma - m_{\gamma}\right)^\top U_{\gamma}\left(\beta_\gamma - m_{\gamma}\right)\right\}d\beta_\gamma\\
    =&\frac{\psi^{|\gamma|_0/2}}{s^{|\gamma|_0}|U_{\gamma}|^{1/2}}\exp\left\{\frac{1}{2\psi}m_{\gamma}^\top U_{\gamma}m_{\gamma}\right\}.
\end{align*}
\end{proof}

\subsubsection{The concavity of the quasi-log-likelihood}{\label{Sec:Concavity}}

Both the theory and practice of the quasi-posterior variable selection approach is more straightforward when the quasi-log-likelihood is concave. Concavity was one option for regularity condition A2.5, and further facilitates efficient methods to find $\tilde{\beta}_{\gamma}$ in \eqref{eq:laplace_approx}.
Proposition \ref{Prop:Concavity} provides conditions under which the quasi-log-likelihood is concave. 

\begin{proposition}[Concavity of log-quasi posterior ]
\label{Prop:Concavity}
    Under prior specification \eqref{Equ:spike_and_slab} and provided $X_{\gamma}$ has full rank and $\psi > 0$
    \begin{enumerate}
        \item[i)] $m_n(\beta_{\gamma}) := -nQ_n(y, X_{\gamma}; \beta_{\gamma}, \psi) - \log(\pi(\beta_{\gamma}))$ is convex in $\beta_{\gamma}$ if mean and variance function $\mu(\cdot)$ and $V(\cdot)$ satisfy
        \begin{align}
            \left(\frac{V^{\prime}(\mu(x_{i\gamma}^\top\beta_{\gamma}))}{V(\mu(x_{i\gamma}^\top\beta_{\gamma}))}\left(\mu^\prime(x_{i\gamma}^\top\beta_{\gamma})\right)^2 - \mu^{\prime\prime}(x_{i\gamma}^\top\beta_{\gamma})\right) = 0.
        \end{align}
        \item[ii)] If $\gamma^{\ast}\subseteq \gamma$ then ${E}_{F_0}\left[\nabla_{\beta_{\gamma}}^2m_n(y, X; \beta_{\gamma}^{\ast}, \psi)\right] = \frac{1}{\hat{\psi}}X_{\gamma}^\top D^{\ast} X_{\gamma}  + \frac{1}{s^2} I > 0$ where
        $D^{\ast}  = diag(d_1^{\ast} , \ldots, d_n^{\ast} )$ with $d_i^{\ast} = \frac{\mu^\prime(x_{i\gamma}^\top\beta_{\gamma}^{\ast})^2}{V(x_{i\gamma}^\top\beta_{\gamma}^{\ast})}$
    \end{enumerate}
\end{proposition}

Note, Part ii) of Proposition \ref{Prop:Concavity} ensures A2.4 holds when the model is overspecified $\gamma^{\ast}\subseteq \gamma$. 

\begin{proof}

\textbf{Proof of Part i)}

First, remember that
\begin{align*}
    m_n(\beta_{\gamma}) := -nQ_n(y, X_{\gamma}; \beta_{\gamma}, \psi) -\log(\pi(\beta_{\gamma})) = -\sum_{i=1}^n \ell_{\psi}(y_i; x_{i\gamma}, \beta_{\gamma}) + \frac{\beta_{\gamma}^\top\beta_{\gamma}}{2s^2}.
\end{align*}
Now Eq. (6) of \cite{mccullagh1983quasi} represent the quasi-log-likelihood via
\begin{align*}
    \frac{\partial \ell_{\psi}(y_i; x_{i\gamma}, \beta_{\gamma})}{\partial \mu(x_{i\gamma}^\top\beta_{\gamma})} = \frac{y_i - \mu(x_{i\gamma}^\top\beta_{\gamma})}{\psi V(\mu(x_{i\gamma}^\top\beta_{\gamma}))},
\end{align*}
which allows us to use the chain rule to write
\begin{align*}
    \frac{\partial \ell_{\psi}(y_i; x_{i\gamma}, \beta_{\gamma})}{\partial \beta_{\gamma j}} &= \frac{\partial \ell_{\psi}(y_i; x_{i\gamma}, \beta)}{\partial \mu(x_{i\gamma}^\top\beta_{\gamma})}\frac{\partial \mu(x_{i\gamma}^\top\beta_{\gamma})}{\partial\beta_{\gamma j}}\\
    &= \frac{y_i - \mu(x_{i\gamma}^\top\beta_{\gamma})}{\psi V(\mu(x_{i\gamma}^\top\beta_{\gamma}))}\frac{\partial \mu(x_{i\gamma}^\top\beta_{\gamma})}{\partial x_{i\gamma}^\top\beta_{\gamma}}x_{i\gamma,j}\\
    \nabla_{\beta_{\gamma}} \ell_{\psi}(y_i; x_{i\gamma}, \beta_\gamma) &= \frac{y_i - \mu(x_{i\gamma}^\top\beta_{\gamma})}{\psi V(\mu(x_{i\gamma}^\top\beta_{\gamma}))}\mu^\prime(x_{i\gamma}^\top\beta_{\gamma})x_{i\gamma}
\end{align*}
and differentiating once more 
\begin{align*}
    \frac{\partial^2 \ell_{\psi}(y_i; x_{i\gamma}, \beta_{\gamma})}{\partial \beta_{\gamma j}\partial \beta_{\gamma k}} &= \frac{ - x_{i\gamma, k}}{\psi V(\mu(x_{i\gamma}^\top\beta_{\gamma}))}\left(\frac{\partial \mu(x_{i\gamma}^\top\beta_{\gamma})}{\partial x_{i\gamma}^\top\beta_{\gamma}}\right)^2x_{i\gamma,j} + \frac{x_{i\gamma,k}(y_i - \mu(x_{i\gamma}^\top\beta_{\gamma}))}{\psi V(\mu(x_{i\gamma}^\top\beta_{\gamma}))}\frac{\partial^2 \mu(x_{i\gamma}^\top\beta_{\gamma})}{\partial (x_{i\gamma}^\top\beta_{\gamma})^2}x_{i\gamma,j}\\
    &\quad- \frac{x_{i\gamma,k}(y_i - \mu\left(x_{i\gamma}^\top\beta_{\gamma}\right))V^{\prime}(\mu(x_{i\gamma}^\top\beta_{\gamma}))\mu^{\prime}(x_{i\gamma}^\top\beta_{\gamma})}{\psi V(\mu(x_{i\gamma}^\top\beta_{\gamma}))^2}\frac{\partial \mu(x_{i\gamma}^\top\beta_{\gamma})}{\partial x_{i\gamma}^\top\beta_{\gamma}}x_{i\gamma,j}.
\end{align*}

As a result, 
\begin{align}
    \nabla_{\beta_{\gamma}}^2 m_n(\beta_{\gamma}) = \frac{1}{\psi}X_{\gamma}^\top D X_{\gamma} + \frac{1}{s^2} I
\end{align}
where $D = diag(d_1, \ldots, d_n)$ with
\begin{align*}
    d_i = \frac{1}{V(\mu(x_{i\gamma}^\top\beta_{\gamma}))}\left(\mu^\prime(x_{i\gamma}^\top\beta_{\gamma})^2 + (y_i - \mu(x_{i\gamma}^\top\beta_{\gamma}))\left(\frac{V^{\prime}(\mu(x_{i\gamma}^\top\beta_{\gamma}))}{V(\mu(x_{i\gamma}^\top\beta_{\gamma}))}\left(\mu^\prime(x_{i\gamma}^\top\beta_{\gamma})\right)^2 - \mu^{\prime\prime}(x_{i\gamma}^\top\beta_{\gamma})\right)\right)
\end{align*}
and so under the assumption that $X_{\gamma}$ has full rank, sufficient conditions for positive definite $\nabla_{\beta_{\gamma}}^2 m_n(\beta_{\gamma})$ are that $d_i > 0$, $i = 1,\ldots, n$, which is guaranteed if
\begin{align}
    \left(\frac{V^{\prime}(\mu(x_{i\gamma}^\top\beta_{\gamma}))}{V(\mu(x_{i\gamma}^\top\beta_{\gamma}))}\left(\mu^\prime(x_{i\gamma}^\top\beta_{\gamma})\right)^2 - \mu^{\prime\prime}(x_{i\gamma}^\top\beta_{\gamma})\right) = 0
\end{align}
for all $i$.

\textbf{Proof of Part ii)}

Further, if $\gamma^{\ast} \subseteq \gamma$ then $\mu(x_i^\top\beta^{\ast}) = \mu(x_{i\gamma}^\top\beta^{\ast}_{\gamma})$ and therefore,
\begin{align*}
    {E}_{F_0}\left[\nabla_{\beta}^2m_n(\beta^{\ast}_{\gamma})\right] = \frac{1}{\psi}X_{\gamma}^\top D^{\ast} X_{\gamma}  + \frac{1}{s^2} I
\end{align*}
where
$D^{\ast}  = diag(d_1^{\ast} , \ldots, d_n^{\ast} )$ with
\begin{align*}
    d_i^{\ast}  &= \frac{1}{V(x_{i\gamma}^\top\beta_{\gamma}^{\ast})}{E}_{F_0}\left[\left(\mu^\prime(x_{i\gamma}^\top\beta_{\gamma}^{\ast})^2 + (y_i - \mu(x_{i\gamma}^\top\beta_{\gamma}^{\ast}))\left(\frac{V^{\prime}(x_{i\gamma}^\top\beta_{\gamma}^{\ast})}{V(x_{i\gamma}^\top\beta_{\gamma}^{\ast})}\left(\mu^\prime(x_{i\gamma}^\top\beta_{\gamma}^{\ast})\right)^2 - \mu^{\prime\prime}(x_{i\gamma}^\top\beta_{\gamma}^{\ast})\right)\right)\right]\\
    &= \frac{1}{V(x_{i\gamma}^\top\beta_{\gamma}^{\ast})}\left(\mu^\prime(x_{i\gamma}^\top\beta_{\gamma}^{\ast})^2 + (\mu(x_i^\top\beta^{\ast}) - \mu(x_{i\gamma}^\top\beta_{\gamma}^{\ast}))\left(\frac{V^{\prime}(x_{i\gamma}^\top\beta_{\gamma}^{\ast})}{V(x_{i\gamma}^\top\beta_{\gamma}^{\ast})}\left(\mu^\prime(x_{i\gamma}^\top\beta_{\gamma}^{\ast})\right)^2 - \mu^{\prime\prime}(x_{i\gamma}^\top\beta_{\gamma}^{\ast})\right)\right)\\
    &= \frac{\mu^\prime(x_{i\gamma}^\top\beta_{\gamma}^{\ast})^2}{V(x_{i\gamma}^\top\beta_{\gamma}^{\ast})} > 0
\end{align*}
\end{proof}

Note that Part ii) of Proposition \ref{Prop:Concavity} can be found in Section S3.1 of \cite{agnoletto2025bayesian}.

Corollary \ref{Cor:ConcavityCount} verifies the conditions of Proposition \ref{Prop:Concavity} for Examples \ref{Ex:LinearRegression} and \ref{Ex:Count}. 

\begin{corollary}[Concavity of the log-quasi posterior]
\label{Cor:ConcavityCount}
Under the specification of Examples \ref{Ex:LinearRegression} and \ref{Ex:Count} and prior \eqref{Equ:spike_and_slab}, $m_n(\beta_{\gamma}) := -nQ_n(y, X_{\gamma}; \beta_{\gamma}, \psi) - \log(\pi(\beta_{\gamma}))$ is convex in $\beta_{\gamma}$, provided $X_{\gamma}$ has full rank and $\psi > 0$.    
\end{corollary}





\begin{proof}
    Under the specification of Example \ref{Ex:LinearRegression} we have that
    \begin{align*}
        \mu(x_{i\gamma}^\top\beta_{\gamma}) &= x_{i\gamma}^\top\beta_{\gamma}\\
        \Rightarrow \mu^{\prime}(x_{i\gamma}^\top\beta_{\gamma}) &= 1\\
        \Rightarrow \mu^{\prime\prime}(x_{i\gamma}^\top\beta_{\gamma}) &= 0\\
        V(\mu(x_{i\gamma}^\top\beta_{\gamma})) &= 1\\
        \Rightarrow V^{\prime}(\mu(x_{i\gamma}^\top\beta_{\gamma})) &= 0
    \end{align*}

    Further, 
    under the specification of Example \ref{Ex:Count} we have that
    \begin{align*}
        \mu(x_{i\gamma}^\top\beta_{\gamma}) &= \exp\left\{x_{i\gamma}^\top\beta_{\gamma}\right\}\\
        \Rightarrow \mu^{\prime}(x_{i\gamma}^\top\beta_{\gamma}) &= \exp\left\{x_{i\gamma}^\top\beta_{\gamma}\right\}\\
        \Rightarrow \mu^{\prime\prime}(x_{i\gamma}^\top\beta_{\gamma}) &= \exp\left\{x_{i\gamma}^\top\beta_{\gamma}\right\}\\
        V(\mu(x_{i\gamma}^\top\beta_{\gamma})) &= \mu(x_{i\gamma}^\top\beta_{\gamma}) = \exp\left\{x_{i\gamma}^\top\beta_{\gamma}\right\}\\
        \Rightarrow V^{\prime}(\mu(x_{i\gamma}^\top\beta_{\gamma})) &= 1
    \end{align*}
    
    As a result, in both of the above scenarios the mean and variance functions verify the sufficient conditions of Proposition \ref{Prop:Concavity} i) that 
    \begin{align*}
        \left(\frac{V^{\prime}(\mu(x_{i\gamma}^\top\beta_{\gamma}))}{V(\mu(x_{i\gamma}^\top\beta_{\gamma}))}\left(\mu^\prime(x_{i\gamma}^\top\beta_{\gamma})\right)^2 - \mu^{\prime\prime}(x_{i\gamma}^\top\beta_{\gamma})\right) = 0.
    \end{align*}
\end{proof}

\subsubsection{Theorem \ref{thm:chi_squared}}

We restate a result from \cite{mccullagh1983quasi} that established the asymptotic distribution of the difference between quasi-log-likelihoods when comparing an over-specified model, $\gamma$ with $\gamma^{\ast}\subset \gamma$, to the correct specification $\gamma^{\ast}$.

\begin{theorem}[Difference in the quasi-log-likelihood for nested models (Equation 11 \citep{mccullagh1983quasi})]
    Assume regularity conditions A1 and A2.1-5. Let $\gamma$ index a model that contains the true model $\gamma^{\ast}$ but contains additional necessary parameters i.e. $\gamma^{\ast}\subset \gamma$. Then
    \begin{equation*}
        2nQ_n(y, X_{\gamma}, \hat{\beta}_{\gamma}, \psi^{\ast}) - 2nQ_n(y, X_{\gamma^{\ast}}, \hat{\beta}_{\gamma^{\ast}}, \psi^{\ast}) \sim \chi^2_{|\gamma|_0 - |\gamma^{\ast}|_0} + O_p(1/\sqrt{n}).
    \end{equation*}
    Further, the $O_p(1/\sqrt{n})$ error term is unaffected by the insertion of the $\hat{\psi}$ in \eqref{Equ:psi_est} in place of $\psi^{\ast}$.
    \label{thm:chi_squared}
\end{theorem}
Note the difference to what would be obtained were correctly specified log-likelihoods used in place of is that the error term would be  $O_p(1/n)$ \citep{mccullagh1983quasi}. 

\subsection{Proof of Theorem \ref{Thm:ConsistencyRR23_Osqrtn}}

Before proving Theorem \ref{Thm:ConsistencyRR23_Osqrtn} it is useful to first establish that $\left|\left|\hat{\beta}_{\gamma} - \beta^{\ast}_{\gamma}\right|\right| = o_p(1)$, i.e. $\hat{\beta}_{\gamma}$ is a consistent estimate of  $\beta^{\ast}_{\gamma}$, for all $\gamma$.

\begin{theorem}[Consistency of $\hat{\beta}_{\gamma}$]
    Assume A1.1 and A2.1-5 and fix $\gamma\in \{0, 1\}^p$. Then $\hat{\beta}_{\gamma}\overset{P}{\longrightarrow}\beta_{\gamma}^{\ast}$ as $n\rightarrow\infty$.
    \label{Thm:ConsistencyMiller}
\end{theorem}

\begin{proof}
Theorem 5 of \cite{miller2021asymptotic} provide generally applicable sufficient conditions under which the minimiser of an in-sample \textit{loss function} $f_n(\cdot)$ converges to the minimiser of the expectation of $f_n(\cdot)$. This proof therefore proceeds by establishing that conditions A1.1 and A2.1-5 are sufficient for the conditions of Theorem 5 of \cite{miller2021asymptotic} for $f_n(\cdot) = - Q_n(y, X_{\gamma}, \cdot, \psi)$ with $\psi > 0$ (as the value of $\psi$ has no influence on $\hat{\beta}_{\gamma}$ or $\beta_{\gamma}^{\ast}$).

Theorem 5 of \cite{miller2021asymptotic} requires that, i) $\beta_{\gamma}^{\ast}\in B_{\gamma}$ where $B_{\gamma}\subseteq \Gamma\subseteq \mathbb{R}^{|\gamma|_0}$ is open and bounded, as assumed in A2.1, 
ii) $Q_n(y, X_{\gamma}, \cdot, \psi)$ has continuous third derivatives on $B_{\gamma}$, as assumed in A2.2,
iii) $Q_n(y, X_{\gamma}, \beta_{\gamma}, \psi)\rightarrow Q(F_0, \beta_{\gamma}, \psi)$ pointwise  which \cite{agnoletto2025bayesian} discusses that A2.3 is sufficient for the weak law of large numbers to imply (see Proposition 6.13 \citep{ccinlar2011probability}), 
iv) $Q^{''}(F_0, \beta^{\ast}_{\gamma}, \psi) = H(F_0; \beta_{\gamma}^{\ast}, \psi)$ is positive definite, as assumed by A2.4,  
v) $Q^{'''}(F_0, \beta^{\ast}_{\gamma}, \psi)$ is uniformly bounded on $B_{\gamma}$, as assumed in A2.2, and
vi) either
\begin{itemize}
    \item[(1)] $Q(F_0, \beta_{\gamma}, \psi) < Q(F_0, \beta_{\gamma}^{\ast}, \psi)$ for all $\beta_{\gamma}\in B_{\gamma}^{\ast}\setminus \{\beta_{\gamma}^{\ast}\}$ and $\lim\sup_n\sup_{\beta_{\gamma}\in\Gamma_{\gamma}\setminus B_{\gamma}^{\ast}}Q_n(y, X_{\gamma}, \beta_{\gamma}, \psi) < Q_n(y, X_{\gamma}, \beta_{\gamma}^{\ast}, \psi)$ for some compact $B_{\gamma}^{\ast}\subseteq B_{\gamma}$ with $\beta_{\gamma}^{\ast}$ in the interior of $B_{\gamma}^{\ast}$
    \item[(2)] $-Q_n(y, X_{\gamma}, \beta_{\gamma}, \psi)$ is convex i.e. $Q_n(y, X_{\gamma}, \beta_{\gamma}, \psi)$ is concave and $Q^{\prime}(F_0, \beta_{\gamma}, \psi) = 0$
\end{itemize}
which is assumed by A2.5. Therefore Theorem 5 of \cite{miller2021asymptotic} proves that under conditions A1.1 and A2.1-5, $\hat{\beta}_{\gamma}\overset{P}{\longrightarrow}\beta_{\gamma}^{\ast}$ as $n\rightarrow\infty$.
\end{proof}

We are now ready to prove Theorem \ref{Thm:ConsistencyRR23_Osqrtn} that $\left|\left|\hat{\beta}_{\gamma} - \beta^{\ast}_{\gamma}\right|\right| = O_p(1/\sqrt{n})$. This proof follows the same structure as the proof of Proposition 2 in \cite{rossell2023additive}.

\begin{proof}

Theorem 5.23 of \cite{van2000asymptotic} provides sufficient conditions under which $\sqrt{n}(\hat{\beta}_{\gamma} - \beta^{\ast}_{\gamma})$ is asymptomatically normal which is sufficient to show that $\left|\left|\hat{\beta}_{\gamma} - \beta^{\ast}_{\gamma}\right|\right| = O_p(1/\sqrt{n})$. As a result, all that is required is to establish that these conditions are satisfied by A1 and A2.1-6.


The conditions of Theorem 5.23 of \cite{van2000asymptotic} are i) that $\ell_{\psi}(y; x_{\gamma}, \beta_{\gamma})$ is differentiable at $\beta_{\gamma}^{\ast}$ for all $y$ and $x$ which is implied by A2.2, ii) that for every $\beta^{(1)}_{\gamma}$ and $\beta^{(2)}_{\gamma}$ in a neighbourhood of $\beta_{\gamma}^{\ast}$ that 
\begin{align*}
    \left|\ell_{\psi}(y; x_{\gamma}, \beta^{(1)}_{\gamma}) - \ell_{\psi}(y; x_{\gamma}, \beta^{(2)}_{\gamma})\right| \leq c_x(y)\left|\left|\beta^{(1)}_{\gamma} - \beta^{(2)}_{\gamma}\right|\right|_2
\end{align*}
with ${E}_{F_0}\left[c_x(y)^2\right] < \infty$ for all $x$. This is the same as requiring that $\ell_{\psi}(y; x_{\gamma}, \beta_{\gamma})$ is equi-Lipschitz in the region of $\beta_{\gamma}^{\ast}$ which Theorem 7 of \cite{miller2021asymptotic} showed to be the case under A2.2 and A2.6, 
iii) that $H(F_0; \beta_{\gamma}^{\ast}, \psi)$ exists and is positive definite which is assumed by A2.4, iv) $Q_n(y, X, \hat{\beta}_{\gamma}, \psi) \geq \sup_{\beta_{\gamma}} Q_n(y, X, \beta_{\gamma}, \psi) - o_p(1/n)$ which holds by the definition of $\hat{\beta}_{\gamma}$, and v) $||\hat{\beta}_{\gamma} -\beta_{\gamma}^{\ast}||_2 = o_p(1)$ which was established in Theorem \ref{Thm:ConsistencyMiller} under A1.1 and A2.1-5. 

Thus $\left|\left|\hat{\beta}_{\gamma} - \beta^{\ast}_{\gamma}\right|\right| = O_p(1/\sqrt{n})$ follows from Theorem 5.23 of \cite{van2000asymptotic} and Theorem \ref{Thm:ConsistencyMiller}.

\end{proof}

\subsection{Proof of Theorem \ref{thm:variable_selection_quasi_posterior_RR23}}

Before proving Theorem \ref{thm:variable_selection_quasi_posterior_RR23} it is important to first establish that $\hat{\psi} - \psi^{\ast} = O_p(1/\sqrt{n})$, i.e. that $\hat{\psi}$ is $\sqrt{n}$-consistent for $\psi^{\ast}$.

\begin{proposition}[$\sqrt{n}$-consistency of $\hat{\psi}$]
\label{Prop:psi_hat_Osqrtn}
Assume A1.1 and A2.1-7. Then $\hat{\psi} - \psi^{\ast} = O_p(1/\sqrt{n})$.
\end{proposition}

\begin{proof}
Firstly, define
\begin{align*}
    R_n(\beta) &= \frac{1}{n-p}\sum_{i=1}^n\frac{\left(y_i - \mu(x_i^\top\beta)\right)^2}{V(\mu(x_i^\top\beta))}\\
    \tilde{R}_n(\beta) &= \frac{1}{n}\sum_{i=1}^n\frac{\left(y_i - \mu(x_i^\top\beta)\right)^2}{V(\mu(x_i^\top\beta))}
\end{align*}
and recall that $\hat{\psi} = R_n(\hat{\beta})$ and 
therefore, the result requires that $R_n(\hat{\beta}) - \psi^{\ast} = O_p(1/\sqrt{n})$.

Firstly, we add and subtract $R_n(\beta^{\ast})$ so that
\begin{align*}
    R_n(\hat{\beta}) - \psi^{\ast} &= R_n(\hat{\beta}) - R_n(\beta^{\ast}) + R_n(\beta^{\ast}) - \psi^{\ast},
\end{align*}
and separately show that the first two and last two terms are both $O_p(1/\sqrt{n})$. 

For the first two terms recall that Theorem \ref{Thm:ConsistencyRR23_Osqrtn} and A1.1 and A2.1-6 provide that $\left|\left|\hat{\beta} - \beta^{\ast}\right|\right| = O_p(1/\sqrt{n})$ and therefore the continuity of $\mu(\cdot)$ and $V(\cdot)$, which is implied by A2.2, provides that $R_n(\hat{\beta}) - R_n(\beta^{\ast}) = O_p(1/\sqrt{n})$ via the continuous mapping theorem.

For the second two terms we further add and subtract $\tilde{R}_n(\beta^{\ast})$ so that we can write 
\begin{align*}
    R_n(\beta^{\ast})  - \psi^{\ast} = R_n(\beta^{\ast}) - \tilde{R}_n(\beta^{\ast}) + \tilde{R}_n(\beta^{\ast}) - \psi^{\ast}.
\end{align*}

Then, under A1.1 , 
     \begin{align*}
         {E}_{F_0}\left[\sqrt{n}\left(\tilde{R}_n(\beta^{\ast}) - \psi^{\ast}\right)\right] &= 0\\
         var_{F_0}\left[\sqrt{n}\left(\tilde{R}_n(\beta^{\ast}) - \psi^{\ast}\right)\right] &= var_{F_0}\left[\frac{1}{\sqrt{n}}\sum_{i=1}^n\frac{\left(y_i - \mu(x_i^\top\beta^{\ast})\right)^2}{V(\mu(x_i^\top\beta^{\ast}))}\right]\\
         &= \frac{1}{n}\sum_{i=1}^n var_{F_0}\left[\frac{\left(y_i - \mu(x_i^\top\beta^{\ast})\right)^2}{V(\mu(x_i^\top\beta^{\ast}))}\right]
     \end{align*}
     where $\frac{1}{n}\sum_{i=1}^n var_{F_0}\left[\frac{\left(y_i - \mu(x_i^\top\beta^{\ast})\right)^2}{V(\mu(x_i^\top\beta^{\ast}))}\right] < \infty$ by A2.7. Then we can apply Chebyshev's inequality to  $A_n = \sqrt{n}\left(\tilde{R}_n(\beta^{\ast}) - \psi^{\ast}\right)$ which provides that 
     \begin{align*}
         P\left(\sqrt{n}\left(\tilde{R}_n(\beta^{\ast}) - \psi^{\ast}\right)\geq M\right)\leq \frac{\frac{1}{n}\sum_{i=1}^n var_{F_0}\left[\frac{\left(y_i - \mu(x_i^\top\beta^{\ast})\right)^2}{V(\mu(x_i^\top\beta^{\ast}))}\right]}{M^2}
     \end{align*}
     which is sufficient to prove that $\tilde{R}_n(\beta^{\ast}) - \psi^{\ast} = O_p(1/\sqrt{n})$.
Lastly, 
\begin{align*}
    \sqrt{n}\left(R_n(\beta^{\ast}) - \tilde{R}_n(\beta^{\ast})\right) &= \frac{\sqrt{n}}{n-p}\sum_{i=1}^n\frac{y_i - \mu(x_i^\top\beta^{\ast})}{V(\mu(x_i^\top\beta^{\ast}))} - \frac{\sqrt{n}}{n}\sum_{i=1}^n\frac{y_i - \mu(x_i^\top\beta^{\ast})}{V(\mu(x_i^\top\beta^{\ast}))}\\
    &= \frac{\sqrt{n}}{n(n-p)}\sum_{i=1}^n\frac{n\left(y_i - \mu(x_i^\top\beta^{\ast})\right) - (n-p)\left(y_i - \mu(x_i^\top\beta^{\ast})\right)}{V(\mu(x_i^\top\beta^{\ast}))}\\
    &= \frac{1}{\sqrt{n}(n-p)}\sum_{i=1}^n\frac{p\left(y_i - \mu(x_i^\top\beta^{\ast})\right)}{V(\mu(x_i^\top\beta^{\ast}))}\\
    &= O_p(1/\sqrt{n}),
\end{align*}
as $p$ is fixed. Combining the three results above is sufficient to establish that $R_n(\hat{\beta}) - R(\beta^{\ast}) = O_p(1/\sqrt{n})$ as required.
\end{proof}

We are now ready to prove Theorem \ref{thm:variable_selection_quasi_posterior_RR23}. 
This proof follows the structure of the proof of Proposition 3 in \cite{rossell2023additive}.

\begin{proof}
    We aim to characterise the asymptotic behavior of Laplace-approximated Bayes factors
    \begin{align}
        \log\left(\tilde{B}^{\textsc{LA}}_{\gamma\gamma^{\ast}}\right) &:= \log\left(\frac{\tilde{f}^{\textsc{LA}}(y\mid X, \gamma; \hat{\psi})}{\tilde{f}^{\textsc{LA}}(y\mid X, \gamma^{\ast}; \hat{\psi})}\right) = \frac{|\gamma|_0 - |\gamma^{\ast}|_0}{2}\log(2\pi) + T_1 + \log(T_2) + \log(T_3)\label{Equ:BayesFactor_expanded}\\
        T_1 &= nQ_n(y, X, \tilde{\beta}_{\gamma}, \hat{\psi}) - nQ_n(y, X, \tilde{\beta}_{\gamma^{\ast}}, \hat{\psi})\nonumber\\
        T_2 &= \frac{\pi(\tilde{\beta}_{\gamma}\mid \gamma)}{\pi(\tilde{\beta}_{\gamma^{\ast}}\mid \gamma^{\ast})}\nonumber\\
        T_3 &= \frac{\left|nH_n(y; X, \tilde{\beta}_{\gamma^{\ast}}, \hat{\psi}) - \nabla^2_{\beta_{\gamma}}\log \pi(\tilde{\beta}_{\gamma^{\ast}})\right|^{1/2}}{\left|nH_n(y; X, \tilde{\beta}_{\gamma}, \hat{\psi}) - \nabla^2_{\beta_{\gamma}}\log \pi(\tilde{\beta}_{\gamma})\right|^{1/2}}.\nonumber
    \end{align}
    We will characterize each term in \eqref{Equ:BayesFactor_expanded} individually, then combine the results. 
    
    
    First, $(2\pi)^{\frac{|\gamma|_0 - |\gamma^{\ast}|_0}{2}}$ is constant since $|\gamma|_0$ and $|\gamma^{\ast}|_0$ are fixed by assumption. Now, note that
    \begin{align*}
        T_3 = n^{\frac{|\gamma^{\ast}|_0 - |\gamma|_0}{2}}\frac{\left|H_n(y; X, \tilde{\beta}_{\gamma^{\ast}}, \hat{\psi}) - \frac{1}{n}\nabla^2_{\beta_{\gamma}}\log \pi(\tilde{\beta}_{\gamma^{\ast}})\right|^{1/2}}{\left|H_n(y; X, \tilde{\beta}_{\gamma}, \hat{\psi}) - \frac{1}{n}\nabla^2_{\beta_{\gamma}}\log \pi(\tilde{\beta}_{\gamma})\right|^{1/2}}
    \end{align*}
    and by Theorem \ref{Thm:ConsistencyRR23_Osqrtn} and A1.2  we have that $\left|\left|\tilde{\beta}_{\gamma^{\ast}} - \beta^{\ast}_{\gamma^{\ast}}\right|\right|_2 = O_p(1/\sqrt{n})$ and $\left|\left|\tilde{\beta}_{\gamma}-\beta^{\ast}_{\gamma}\right|\right|_2 = O_p(1/\sqrt{n})$ and by Proposition \ref{Prop:psi_hat_Osqrtn} we have that $\left|\hat{\psi} - \psi^{\ast}\right| = O_p(1/\sqrt{n})$. 
    As a result,
    \begin{align*}
        &\left|\left|H_n(y; X, \tilde{\beta}_{\gamma}, \hat{\psi}) - \frac{1}{n}\nabla^2_{\beta_{\gamma}}\log \pi(\tilde{\beta}_{\gamma}) - H(F_0; \beta^{\ast}_{\gamma}, \psi^{\ast})\right|\right|_2\\
        =&\left|\left|\frac{\psi^{\ast}}{\hat{\psi}}H_n(y; X, \tilde{\beta}_{\gamma}, \psi^{\ast}) - \frac{1}{n}\nabla^2_{\beta_{\gamma}}\log \pi(\tilde{\beta}_{\gamma}) - H(F_0; \beta^{\ast}_{\gamma}, \psi^{\ast})\right|\right|_2\\
        \leq&\left|\left|\frac{\psi^{\ast}}{\hat{\psi}}H_n(y; X, \tilde{\beta}_{\gamma}, \psi^{\ast}) - H(F_0; \beta^{\ast}_{\gamma}, \psi^{\ast})\right|\right|_2 + \left|\left|\frac{1}{n}\nabla^2_{\beta_{\gamma}}\log \pi(\tilde{\beta}_{\gamma})\right|\right|_2\\
        =&\left|\left|\frac{\psi^{\ast}}{\hat{\psi}}H_n(y; X, \tilde{\beta}_{\gamma}, \psi^{\ast}) - \frac{\psi^{\ast}}{\hat{\psi}}H_n(y, X, \beta^{\ast}_{\gamma}, \psi^{\ast}) + \frac{\psi^{\ast}}{\hat{\psi}}H_n(y, X, \beta^{\ast}_{\gamma}, \psi^{\ast}) - H(F_0; \beta^{\ast}_{\gamma}, \psi^{\ast})\right|\right|_2\\
        & + \left|\left|\frac{1}{n}\nabla^2_{\beta_{\gamma}}\log \pi(\tilde{\beta}_{\gamma})\right|\right|_2\\
        =&\left|\left|\frac{\psi^{\ast}}{\hat{\psi}}H_n(y; X, \tilde{\beta}_{\gamma}, \psi^{\ast}) - \frac{\psi^{\ast}}{\hat{\psi}}H_n(y, X, \beta^{\ast}_{\gamma}, \psi^{\ast}) + \frac{\psi^{\ast}}{\hat{\psi}}H_n(y, X, \beta^{\ast}_{\gamma}, \psi^{\ast}) - H_n(y, X, \beta^{\ast}_{\gamma}, \psi^{\ast})\right.\right.\\
        &\left.\left. + H_n(y, X, \beta^{\ast}_{\gamma}, \psi^{\ast}) - H(F_0; \beta^{\ast}_{\gamma}, \psi^{\ast}) - \right|\right|_2 + \left|\left|\frac{1}{n}\nabla^2_{\beta_{\gamma}}\log \pi(\tilde{\beta}_{\gamma})\right|\right|_2\\
        \leq& \frac{\psi^{\ast}}{\hat{\psi}}\left|\left|H_n(y; X, \tilde{\beta}_{\gamma}, \psi^{\ast}) - H_n(y, X, \beta^{\ast}_{\gamma}, \psi^{\ast})\right|\right|_2 + \left|\left|\frac{\psi^{\ast}}{\hat{\psi}}H_n(y, X, \beta^{\ast}_{\gamma}, \psi^{\ast}) - H_n(y, X, \beta^{\ast}_{\gamma}, \psi^{\ast})\right|\right|_2\\
        & + \left|\left|H_n(y, X, \beta^{\ast}_{\gamma}, \psi^{\ast}) - H(F_0; \beta^{\ast}_{\gamma}, \psi^{\ast})\right|\right|_2 + \left|\left|\frac{1}{n}\nabla^2_{\beta_{\gamma}}\log \pi(\tilde{\beta}_{\gamma})\right|\right|_2\\
        \leq& \frac{\psi^{\ast}}{\hat{\psi}}\left|\left|H_n(y; X, \tilde{\beta}_{\gamma}, \psi^{\ast}) - H_n(y, X, \beta^{\ast}_{\gamma}, \psi^{\ast})\right|\right|_2 + \left|\left|\frac{\psi^{\ast}}{\hat{\psi}}H_n(y, X, \beta^{\ast}_{\gamma}, \psi^{\ast}) - H_n(y, X, \beta^{\ast}_{\gamma}, \psi^{\ast})\right|\right|_2\\
        & + \left|\left|H_n(y, X, \beta^{\ast}_{\gamma}, \psi^{\ast}) - H(F_0; \beta^{\ast}_{\gamma}, \psi^{\ast})\right|\right|_2 + \left|\left|\frac{1}{n}\nabla^2_{\beta_{\gamma}}\log \pi(\tilde{\beta}_{\gamma}) - \frac{1}{n}\nabla^2_{\beta_{\gamma}}\log \pi(\beta^{\ast}_{\gamma})\right|\right|_2 + \left|\left|\frac{1}{n}\nabla^2_{\beta_{\gamma}}\log \pi(\beta^{\ast}_{\gamma})\right|\right|_2.
    \end{align*}
   Once again, we now consider each term above. Theorem 7 of \cite{miller2021asymptotic} prove that $H_n(y; X, \beta_{\gamma}, \psi)$ is equi-Lipschitz under A1.1 -A2.1-6 which combined with Theorem \ref{Thm:ConsistencyRR23_Osqrtn} and Proposition \ref{Prop:psi_hat_Osqrtn} 
   shows that \newline $\frac{\psi^{\ast}}{\hat{\psi}}\left|\left|H_n(y; X, \tilde{\beta}_{\gamma}, \psi^{\ast}) - H_n(y, X, \beta^{\ast}_{\gamma}, \psi^{\ast})\right|\right|_2 = O_p(1/\sqrt{n})$ and $\left|\left|\frac{\psi^{\ast}}{\hat{\psi}}H_n(y, X, \beta^{\ast}_{\gamma}, \psi^{\ast}) - H_n(y, X, \beta^{\ast}_{\gamma}, \psi^{\ast})\right|\right|_2 = O_p(1/\sqrt{n})$, further by the weak law of large numbers $\left|\left|H_n(y, X, \beta^{\ast}_{\gamma}, \psi^{\ast}) - H(F_0; \beta^{\ast}_{\gamma}, \psi^{\ast})\right|\right|_2 = o_p(1)$, and by the continuous mapping theorem and A1.2 $\left|\left|\frac{1}{n}\nabla^2_{\beta_{\gamma}}\log \pi(\tilde{\beta}_{\gamma}) - \frac{1}{n}\nabla^2_{\beta_{\gamma}}\log \pi(\beta^{\ast}_{\gamma})\right|\right|_2 = O_p(1/\sqrt{n})$ and $\left|\left|\frac{1}{n}\nabla^2_{\beta_{\gamma}}\log \pi(\beta^{\ast}_{\gamma})\right|\right|_2 = O(1/n)$. All of the above arguments also apply if we consider model $\gamma^{\ast}$ in place of model $\gamma$. As a result,
    \begin{align*}
          \frac{\left|H_n(y; X, \tilde{\beta}_{\gamma^{\ast}}, \hat{\psi}) - \frac{1}{n}\nabla^2_{\beta_{\gamma}}\log \pi(\tilde{\beta}_{\gamma^{\ast}})\right|^{1/2}}{\left|H_n(y; X, \tilde{\beta}_{\gamma}, \hat{\psi}) - \frac{1}{n}\nabla^2_{\beta_{\gamma}}\log \pi(\tilde{\beta}_{\gamma})\right|^{1/2}}\overset{P}{\longrightarrow}\frac{\left|H(F_0; \beta^{\ast}_{\gamma^{\ast}}, \psi^{\ast})\right|^{1/2}}{\left|H(F_0; \beta^{\ast}_{\gamma}, \psi^{\ast})\right|^{1/2}}
    \end{align*}
    where the right-hand side is fixed as $\beta^{\ast}_{\gamma^{\ast}}$ and $\beta^{\ast}_{\gamma}$ are fixed by assumption and $H(F_0; \beta^{\ast}_{\gamma^{\ast}}, \psi^{\ast})$ and $H(F_0; \beta^{\ast}_{\gamma}, \psi^{\ast})$ are positive definite by A2.3.    
    Therefore
    \begin{align}
        T_3n^{\frac{|\gamma|_0 - |\gamma^{\ast}|_0}{2}}\overset{P}{\longrightarrow}\frac{\left|H(F_0; \beta^{\ast}_{\gamma^{\ast}}, \psi^{\ast})\right|^{1/2}}{\left|H(F_0; \beta^{\ast}_{\gamma}, \psi^{\ast})\right|^{1/2}}\in (0, \infty).\label{equ:T3}
    \end{align}

    For $T_2$, under Theorem \ref{Thm:ConsistencyRR23_Osqrtn} and by A1.2 and the continuous mapping theorem we have that $\pi(\tilde{\beta}_{\gamma}\mid\gamma)= \pi(\beta^{\ast}_{\gamma}\mid \gamma) + O_p(1/\sqrt{n})$ for any $\gamma$, hence
    \begin{align}
        T_2= \frac{\pi(\beta^{\ast}_{\gamma^{\ast}}\mid \gamma^{\ast})}{\pi(\beta^{\ast}_{\gamma}\mid \gamma)} + O_p(1/\sqrt{n})\label{equ:T2}
    \end{align}

    To characterise $T_1$, we must consider separately the case where $\gamma^{\ast}\not\subset\gamma$ and the case where $\gamma^{\ast}\subset \gamma$. This separates the cases when the limiting quasi-log likelihood of model $\gamma$ is strictly less than the limiting quasi-log-likelihood of the generating model $\gamma^{\ast}$ or equal to it.

    \begin{enumerate}
        \item[i)] \textbf{Case} $\gamma^{\ast}\not\subset \gamma$. In this case, for any $\gamma$ we have that
        \begin{align*}
            &\left|Q_n(y, X, \tilde{\beta}_{\gamma}, \hat{\psi}) - Q(F_0, \beta^{\ast}_{\gamma}, \psi^{\ast})\right|\\
            =&\left|\frac{\hat{\psi}}{\psi^{\ast}}Q_n(y, X, \tilde{\beta}_{\gamma}, \psi^{\ast}) - Q(F_0, \beta^{\ast}_{\gamma}, \psi^{\ast})\right|\\
            =&\left|\frac{\hat{\psi}}{\psi^{\ast}}Q_n(y, X, \tilde{\beta}_{\gamma}, \psi^{\ast}) - \frac{\hat{\psi}}{\psi^{\ast}}Q_n(y, X, \beta^{\ast}_{\gamma}, \psi^{\ast}) + \frac{\hat{\psi}}{\psi^{\ast}}Q_n(y, X, \beta^{\ast}_{\gamma}, \psi^{\ast}) - Q_n(y, X, \beta^{\ast}_{\gamma}, \psi^{\ast}) \right.\\
            &\left.+ Q_n(y, X, \beta^{\ast}_{\gamma}, \psi^{\ast}) - Q(F_0, \beta^{\ast}_{\gamma}, \psi^{\ast}) \right|\\
            \leq&\frac{\hat{\psi}}{\psi^{\ast}}\left|Q_n(y, X, \tilde{\beta}_{\gamma}, \psi^{\ast}) - Q_n(y, X, \beta^{\ast}_{\gamma}, \psi^{\ast})\right| + \left|\frac{\hat{\psi}}{\psi^{\ast}}Q_n(y, X, \beta^{\ast}_{\gamma}, \psi^{\ast}) - Q_n(y, X, \beta^{\ast}_{\gamma}, \psi^{\ast})\right|\\
            &+ \left|Q_n(y, X, \beta^{\ast}_{\gamma}, \psi^{\ast}) - Q(F_0, \beta^{\ast}_{\gamma}, \psi^{\ast})\right|
        \end{align*}
     where Theorem 7 of \cite{miller2021asymptotic} prove that $Q_n(y; X, \beta_{\gamma}, \psi)$ is equi-Lipschitz under A1.1 and A2.1-6 which combined with Theorem \ref{Thm:ConsistencyRR23_Osqrtn} shows $\left|Q_n(y; X, \tilde{\beta}_{\gamma}, \psi^{\ast}) - Q_n(y, X, \beta^{\ast}_{\gamma}, \psi^{\ast})\right| = O_p(1/\sqrt{n})$, and Proposition \ref{Prop:psi_hat_Osqrtn} 
     ensures that $\frac{\hat{\psi}}{\psi^{\ast}}\left|Q_n(y, X, \tilde{\beta}_{\gamma}, \psi^{\ast}) - Q_n(y, X, \beta^{\ast}_{\gamma}, \psi^{\ast})\right| = O_p(1/\sqrt{n})$ and $\left|\frac{\hat{\psi}}{\psi^{\ast}}Q_n(y, X, \beta^{\ast}_{\gamma}, \psi^{\ast}) - Q_n(y, X, \beta^{\ast}_{\gamma}, \psi^{\ast})\right| = O_p(1/\sqrt{n})$. 
     Further, 
     \begin{align*}
         {E}_{F_0}\left[\sqrt{n}\left(Q_n(y, X, \beta^{\ast}_{\gamma}, \psi^{\ast}) - Q(F_0, \beta^{\ast}_{\gamma}, \psi^{\ast})\right)\right] &= 0\\
         var_{F_0}\left[\sqrt{n}\left(Q_n(y, X, \beta^{\ast}_{\gamma}, \psi^{\ast}) - Q(F_0, \beta^{\ast}_{\gamma}, \psi^{\ast})\right)\right] &= var_{F_0}\left[\frac{1}{\sqrt{n}}\sum_{i=1}^n\ell_{\psi^{\ast}}(y_i; x_{i\gamma}, \beta^{\ast}_{\gamma})\right]\\
         &= \frac{1}{n}\sum_{i=1}^n var_{F_0}\left[\ell_{\psi^{\ast}}(y_i; x_{i\gamma}, \beta^{\ast}_{\gamma})\right]
     \end{align*}
     where $\frac{1}{n}\sum_{i=1}^n var_{F_0}\left[\ell_{\psi^{\ast}}(y_i; x_{i\gamma}, \beta^{\ast}_{\gamma})\right] < \infty$ by A2.7. Then we can apply Chebyshev's inequality to $A_n = \sqrt{n}\left(Q_n(y, X, \beta^{\ast}_{\gamma}, \psi^{\ast}) - Q(F_0, \beta^{\ast}_{\gamma}, \psi^{\ast})\right)$ which provides that 
     \begin{align*}
         P\left(\left|\sqrt{n}\left(Q_n(y, X, \beta^{\ast}_{\gamma}, \psi^{\ast}) - Q(F_0, \beta^{\ast}_{\gamma}, \psi^{\ast})\right)\right|\geq M\right)\leq \frac{\frac{1}{n}\sum_{i=1}^n var_{F_0}\left[\ell_{\psi^{\ast}}(y_i; x_{i\gamma}, \beta^{\ast}_{\gamma})\right]}{M^2}
     \end{align*}
     which is sufficient to prove that $\left|Q_n(y, X, \beta^{\ast}_{\gamma}, \psi^{\ast}) - Q(F_0, \beta^{\ast}_{\gamma}, \psi^{\ast})\right| = O_p(1/\sqrt{n})$.
     
     Therefore
        \begin{align*}
            \frac{1}{n}T_1 = Q_n(y, X, \tilde{\beta}_{\gamma}, \hat{\psi}) - Q_n(y, X, \tilde{\beta}_{\gamma^{\ast}}, \hat{\psi}) = Q(F_0, \beta^{\ast}_{\gamma}, \psi^{\ast}) - Q(F_0, \beta^{\ast}_{\gamma^{\ast}}, \psi^{\ast})  + O_p(1/\sqrt{n})\label{equ:T1}
        \end{align*}
        with $Q(F_0, \beta^{\ast}_{\gamma}, \psi^{\ast}) - Q(F_0, \beta^{\ast}_{\gamma^{\ast}}, \psi^{\ast}) < 0$. 
        
        Therefore, combining  \eqref{equ:T3} \eqref{equ:T2} and \eqref{equ:T1} gives that 
        \begin{align*}
            &B_{\gamma\gamma^{\ast}} = \frac{|\gamma|_0 - |\gamma^{\ast}|_0}{2}\log(2\pi) + T_1 + \log(T_2) + \log(T_3)\\
            =& \frac{|\gamma|_0 - |\gamma^{\ast}|_0}{2}\log\left(\frac{2\pi}{n}\right) + n\left\{Q(F_0, \beta^{\ast}_{\gamma}, \psi^{\ast}) - Q(F_0, \beta^{\ast}_{\gamma^{\ast}}, \psi^{\ast})\right\}\\
            &+ \log\left(\frac{\pi(\beta^{\ast}_{\gamma^{\ast}}\mid \gamma^{\ast})}{\pi(\beta^{\ast}_{\gamma}\mid \gamma)}\right) + \log\left(\frac{\left|H(F_0; \beta^{\ast}_{\gamma^{\ast}}, \psi^{\ast})\right|^{1/2}}{\left|H(F_0; \beta^{\ast}_{\gamma}, \psi^{\ast})\right|^{1/2}}\right) + O_p(\sqrt{n})\\
            =& n\left\{Q(F_0, \beta^{\ast}_{\gamma}, \psi^{\ast}) - Q(F_0, \beta^{\ast}_{\gamma^{\ast}}, \psi^{\ast})\right\} + O_p(\sqrt{n})
        \end{align*}
        as required.
        \item[ii)] \textbf{Case} $\gamma^{\ast}\subset \gamma$. In this case A1.2 and Theorem \ref{thm:chi_squared} provide that  
        \begin{align}
            T_1 = \frac{1}{2}\chi^2_{|\gamma|_0 - |\gamma^{\ast}|_0} + O_p(1/\sqrt{n})
        \end{align}
        which combined with \eqref{equ:T3} and \eqref{equ:T2} gives that 
        \begin{align*}
            &B_{\gamma\gamma^{\ast}} = \frac{|\gamma|_0 - |\gamma^{\ast}|_0}{2}\log(2\pi) + T_1 + \log(T_2) + \log(T_3)\\
            =& \frac{|\gamma|_0 - |\gamma^{\ast}|_0}{2}\log(2\pi) + \frac{1}{2}\chi^2_{|\gamma|_0 - |\gamma^{\ast}|_0} + \log\left(\frac{\pi(\beta^{\ast}_{\gamma^{\ast}}\mid \gamma^{\ast})}{\pi(\beta^{\ast}_{\gamma}\mid \gamma)}\right) + \frac{|\gamma^{\ast}|_0 - |\gamma^{}|_0}{2}\log(n)\\
            &+\log\left(\frac{\left|H(F_0; \beta^{\ast}_{\gamma^{\ast}}, \psi^{\ast})\right|^{1/2}}{\left|H(F_0; \beta^{\ast}_{\gamma}, \psi^{\ast})\right|^{1/2}}\right) + o_p(1) \\
            =&\frac{|\gamma|_0 - |\gamma^{\ast}|_0}{2}\log\left(\frac{2\pi}{n}\right) + \frac{1}{2}\chi^2_{|\gamma|_0 - |\gamma^{\ast}|_0} + \log\left(\frac{\pi(\beta^{\ast}_{\gamma^{\ast}}\mid \gamma^{\ast})}{\pi(\beta^{\ast}_{\gamma}\mid \gamma)}\right) + \log\left(\frac{\left|H(F_0; \beta^{\ast}_{\gamma^{\ast}}, \psi^{\ast})\right|^{1/2}}{\left|H(F_0; \beta^{\ast}_{\gamma}, \psi^{\ast})\right|^{1/2}}\right) + o_p(1),
        \end{align*}
        as required.
    \end{enumerate}
\end{proof}

Note that in Case ii) an error term of the order $O_p(1/\sqrt{n})$ could be achieved with slightly stronger conditions on the quasi-log-likelihood hessian matrix. However, we avoided this extra technicality here as it is not relevant to the validity of our proposed procedure. 

\subsection{Proof of Proposition \ref{Prop:Lapalce_accuracy}}

We prove Proposition \ref{Prop:Lapalce_accuracy} regarding the accuracy of the Laplace approximation to the quasi-marginal-likelihood.

\begin{proof}
    This results follows directly Theorem 5 of \cite{miller2021asymptotic} and the proof of Theorem \ref{Thm:ConsistencyMiller} which established that A1 and A2.1-5 imply the conditions of Theorem 5 of \cite{miller2021asymptotic} for $f_n(\cdot) = - Q_n(y, X_{\gamma}, \cdot, \psi)$ with $\psi > 0$.
\end{proof}

We further, note that \cite{rossell2023additive} contains the same result under slightly weaker conditions that applies when the quasi-log-likelihood is concave.

\section{Additional implementation and experimentation details}{\label{sec:supp_experiments}}

Section~\ref{sec:supp_experiments} contains additional details on the Bayesian FDR variable selection procedure, the simplified Gibbs sampler for the model quasi-posterior with fixed $w$, and supplementary results from the simulation studies and real data analyses.

\subsection{Controlling the Bayesian False Discovery Rate}\label{sec:FDR-control}


We briefly review how we deploy the thresholding rule of \citet{mueller:2004} in order to control Bayesian expected false discovery rate (FDR) of our variable selection procedure.
Under our framework, each covariate $j$ is associated with a binary inclusion indicator $\gamma_j = \mathbb{I}\{\beta_j \neq 0\}$, and we denote by $\pi_j = \tilde{\pi}(\gamma_j = 1 \mid y, X)$ the marginal (quasi) posterior inclusion probability for covariate $j$.  Let $S = \{j: \gamma_j = 1\}$ be the set of selected covariates, where $|S|_0 = \sum_{j=1}^p \gamma_j$. The posterior expected number of false discoveries is the expected count of included covariates that are actually spurious, namely  $\sum_{j \in S} \tilde{\pi}(\gamma_j = 0 \mid y, X) \;=\; \sum_{j \in S} (1 - \pi_j)$. Thus, the (quasi) posterior expected proportion of false discoveries is given by
\[
\frac{1}{|S|}\sum_{j \in S}(1-\pi_j)
\;=\; 1 - \frac{1}{|S|}\sum_{j \in S}\pi_j,
\quad |S|>0,
\]
that is, one minus the average posterior inclusion probability among the selected covariates.  Naturally, maintaining this quantity below a target level $\alpha$ is equivalent to requiring
$\frac{1}{|S|}\sum_{j \in S}\pi_j \;\ge\; 1-\alpha.$

This motivates a simple selection rule: include as many covariates as possible while keeping their average posterior inclusion probability above $1-\alpha$.  
To operationalise this, we order the inclusion probabilities as $\pi_{(1)} \ge \pi_{(2)} \ge \cdots \ge \pi_{(p)}$ and select the top $k$ covariates, where $k$ is the largest index such that 
$
\frac{1}{k}\sum_{i=1}^k \pi_{(i)} \;\ge\; 1-\alpha$. Hence, this procedure yields an adaptive and data-driven framework to variable selection, since $\pi_{(k)}$ becomes the implicit threshold. 

\subsection{Gibbs sampler for fixed $w$} 

Algorithm \ref{Alg:Gibbs} provides a special case of Algorithm \ref{Alg:Gibbs_BetaBinom} to sample from $\tilde{\pi}(\gamma\mid y, X)$ when $w$ is assumed fixed.

\begin{algorithm}
\small
\caption{Gibbs Sampling from $\tilde{\pi}(\gamma\mid y, X; \hat{\psi})$}\label{Alg:Gibbs}
\KwIn{Data $y$ and $X$, dispersion $\hat{\psi}$, prior hyperparameters $\{w, s\}$, initial state $\gamma^{(0)}$ and number of sample $N$.}
\KwOut{$\{\gamma^{(t)}\}_{t=1}^N$, approximate samples from $\tilde{\pi}(\gamma\mid y, X; \hat{\psi})$.} 
\For{$t =1, \ldots, N$}{
    Generate $l=(l_1,\ldots,l_p)$ a random permutation of $\{1,\ldots,p\}$. Set $\gamma = \gamma^{(t-1)}$. \\
    \For{$j =1, \ldots, p$}{
        Set $\gamma^{+}_{-l_j} = \gamma_{-l_j}$, $\gamma^{+}_{l_j} = 1$, $\gamma^{-}_{-l_j} = \gamma_{-l_j}$, and $\gamma^{-}_{l_j} = 1$.\\
        Set $\gamma_{l_j} = 1$ with probability $\left(1 + \tilde{B}_{\gamma^{-}\gamma^{+}}\frac{(1-w)}{w}\right)^{-1}$ or $\gamma_{l_j} = 0$ otherwise. \\
    }
Set $\gamma^{(t)} = \gamma$.
}
\end{algorithm}

\subsection{Variable Selection Performance metrics}\label{Sec:metrics}

In Sections \ref{sec:simul_studies} and \ref{Sec:real_data} 
variable selection performance is evaluated using the false discovery rate (FDR), statistical power,  F1 score and the Matthews correlation coefficient (MCC). 
Let TP, FP, TN, and FN denote the number of true positives, false positives, true negatives, and false negatives, respectively, where the positive class corresponds to active variables i.e. $\beta_j\neq 0$. The metrics are defined as
\begin{align*}
    \text{FDR} &= \frac{\mathrm{FP}}{\mathrm{TP}+\mathrm{FP}}, \quad \text{Power} = \frac{\mathrm{TP}}{\mathrm{TP}+\mathrm{FN}}, \quad \text{F1} = \frac{2 \, \mathrm{TP}}{2 \, \mathrm{TP} + \mathrm{FP} + \mathrm{FN}},\\
    &\text{MCC} = \frac{\mathrm{TP}\cdot \mathrm{TN} - \mathrm{FP}\cdot \mathrm{FN}}{\sqrt{(\mathrm{TP}+\mathrm{FP})(\mathrm{TP}+\mathrm{FN})(\mathrm{TN}+\mathrm{FP})(\mathrm{TN}+\mathrm{FN})}}.
\end{align*}

\subsection{Additional Simulation Results}

\subsubsection{Overdispersed Count Data}
Figure \ref{fig:coutsim-metrics-0.5} presents additional results  for the overdispersed count data simulated example from Section \ref{Sec:sim_count}.
The top panel shows that under the median model posterior inclusion probability thresholding we see a slight increase in the FDR of the QP relative to the negative-binomial model, but for $n > 25$ we see that the QP has the same FDR and improved power. The Poisson model continues to have very high FDR. The bottom panel shows that across samples sizes the QP assigns lower probability to truly inactive covariates than the Poisson model and higher probability to truly active covariates than the negative binomial model. 

\begin{figure}[ht!]
    \centering
    \includegraphics[width=0.95\textwidth]{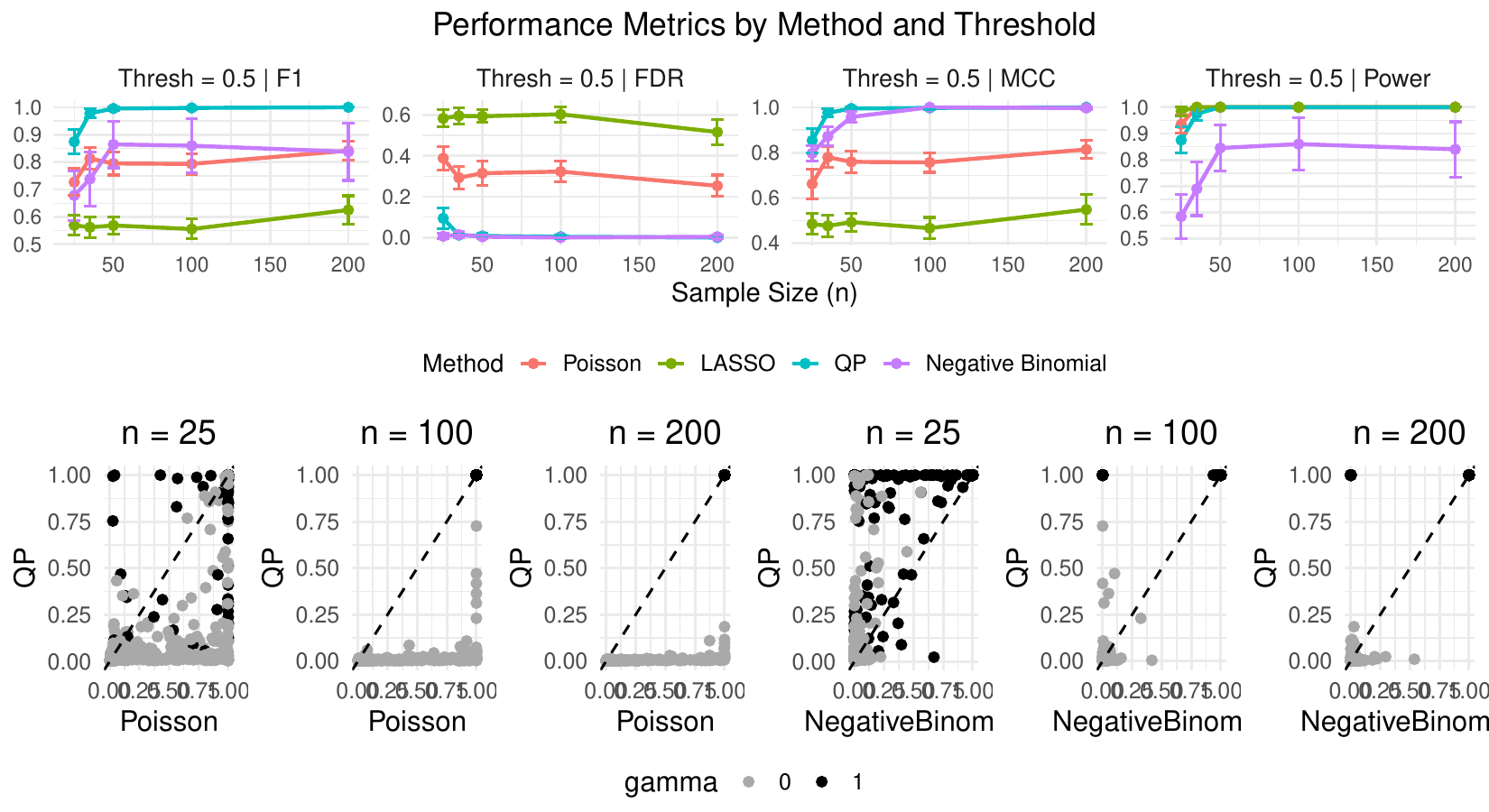}
    \caption{
    Overdispersed count data. 
Top: Variable selection performance of quasi–posterior (QP), negative binomial (NB), and Poisson (Pois) models across sample sizes using fixed cutoff at $\pi_j \geq 0.5$. Error bars denote $\pm 2$ standard errors across simulation replicates. 
Bottom: Estimated posterior probabilities of inclusion  across repeats, coloured according to whether the generating $\beta^\ast$ was zero (grey) or non-zero (black).}
    \label{fig:coutsim-metrics-0.5}
\end{figure}

\subsubsection{Linear Regression with non-Gaussian errors}


\paragraph{Heavy-tailed linear regression}

Figure \ref{fig:multcomp_probs_homo_student0.5} presents additional results for the heavy tailed linear regression simulated example from Section \ref{Sec:sim_heavy}. The top panel shows that the median model selection procedure achieves higher power for both \texttt{mombf} and the QP, relative to Figure \ref{fig:multcomp_probs_homo_student}, at the cost of a inflated FDR at the smallest sample size. This plot also demonstrates that while the power of LASSO procedure is better than the both Bayesian methods, this comes at the cost of an FDR of around 0.5 which does not appear to decrease as the sample size increases. The comparison of the probabilities in the bottom panel shows that for $n \geq 125$ the inclusion probabilities estimated by the QP and \texttt{mombf} are largely the same, with the QP assigning larger probability to truly active covariates in a small number of cases.

\begin{figure}[ht!]
    \centering

    \includegraphics[width=0.95\linewidth]{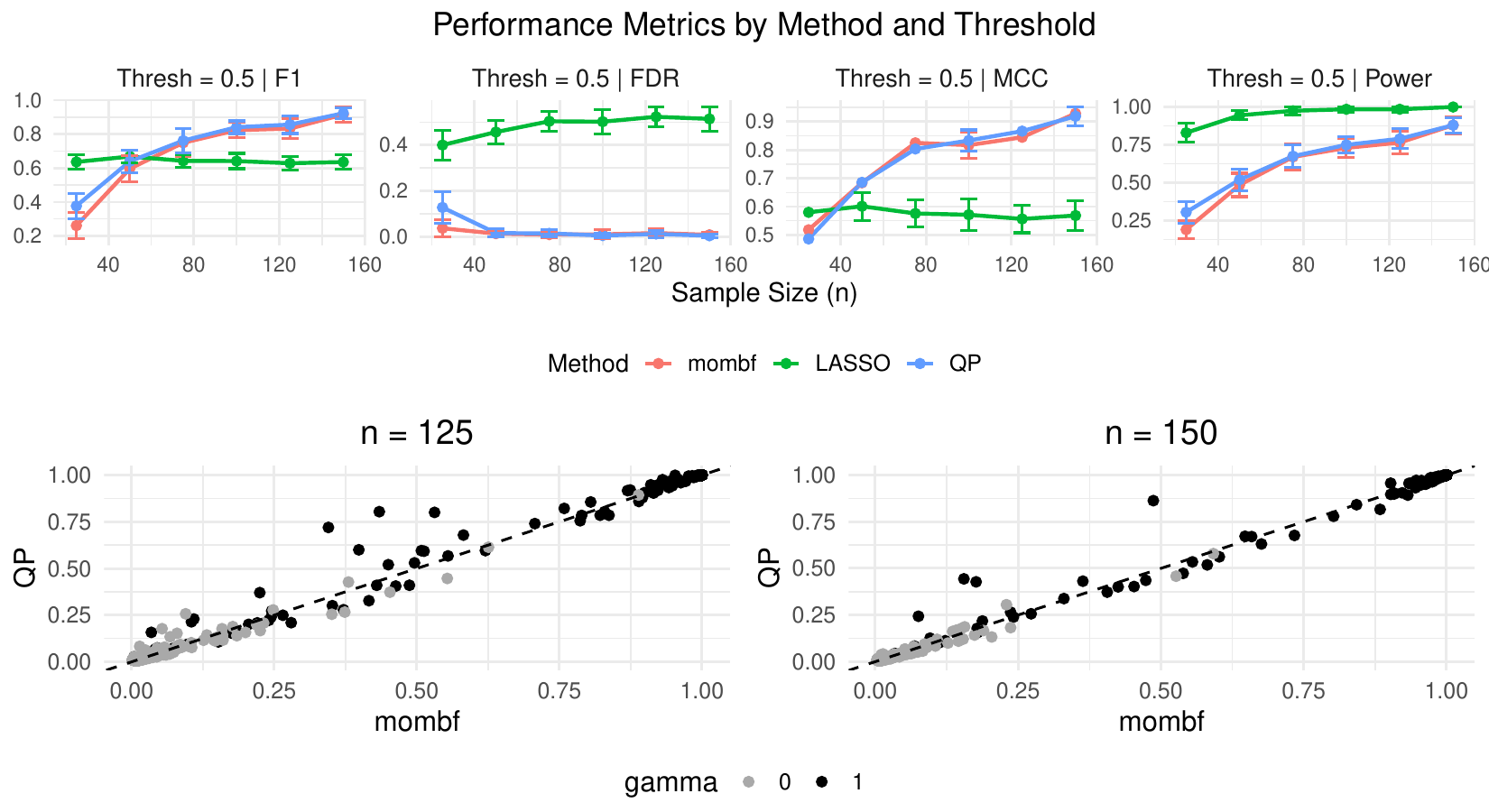}
    \caption{
    Heavy-tailed linear regression. 
Top: Variable selection performance of quasi–posterior (QP), traditional posterior (mombf) and frequentist LASSO models across sample sizes using fixed cutoff at $\pi_j \ge 0.5$. Error bars denote $\pm 2$ standard errors across simulation replicates. 
Bottom: Estimated posterior probabilities of inclusion  across repeats, coloured according to whether the generating $\beta^\ast$ was zero (grey) or non-zero (black).}
    \label{fig:multcomp_probs_homo_student0.5}
\end{figure}

\paragraph{Linear regression with inliers}

Figure \ref{fig:multcomp_selestion_homo_inlier0.5} and \ref{fig:multcomp_metrics_homo_inlier_psiLASSO} presents additional results for the inlier linear regression simulated example from Section \ref{Sec:inliers}. 
The top panel of Figure \ref{fig:multcomp_selestion_homo_inlier0.5} shows that relative to Figure \ref{fig:multcomp_selestion_homo_inlier}, the QP and \texttt{mombf} both have increased FDR and power under the median model selection procedure. The probability comparison in the bottom panel shows that for $n\leq 170$ the QP generally assigns much higher inclusion probability to inactive covariates than \texttt{mombf}, which results in its inflated FDR for small $n$. However, the opposite is the case for $n\geq 750$ which demonstrates why the QP has a lower FDR in this setting.

Figure \ref{fig:multcomp_metrics_homo_inlier_psiLASSO} presents a comparison of the model selection performance of the QP with \texttt{mombf} when a LASSO penalised estimate of $\beta$ is used in place of $\hat{\beta}$ in \eqref{Equ:psi_est} to calculate $\hat{\psi}$. This massively reduced the FDR of the QP for small $n$ at the cost of a small reduction in power. The behaviour for large $n$ is similar to what is observed in Figures \ref{fig:multcomp_selestion_homo_inlier} and \ref{fig:multcomp_selestion_homo_inlier0.5}.

\begin{figure}[ht!]
    \centering

    \includegraphics[width=0.95\linewidth]{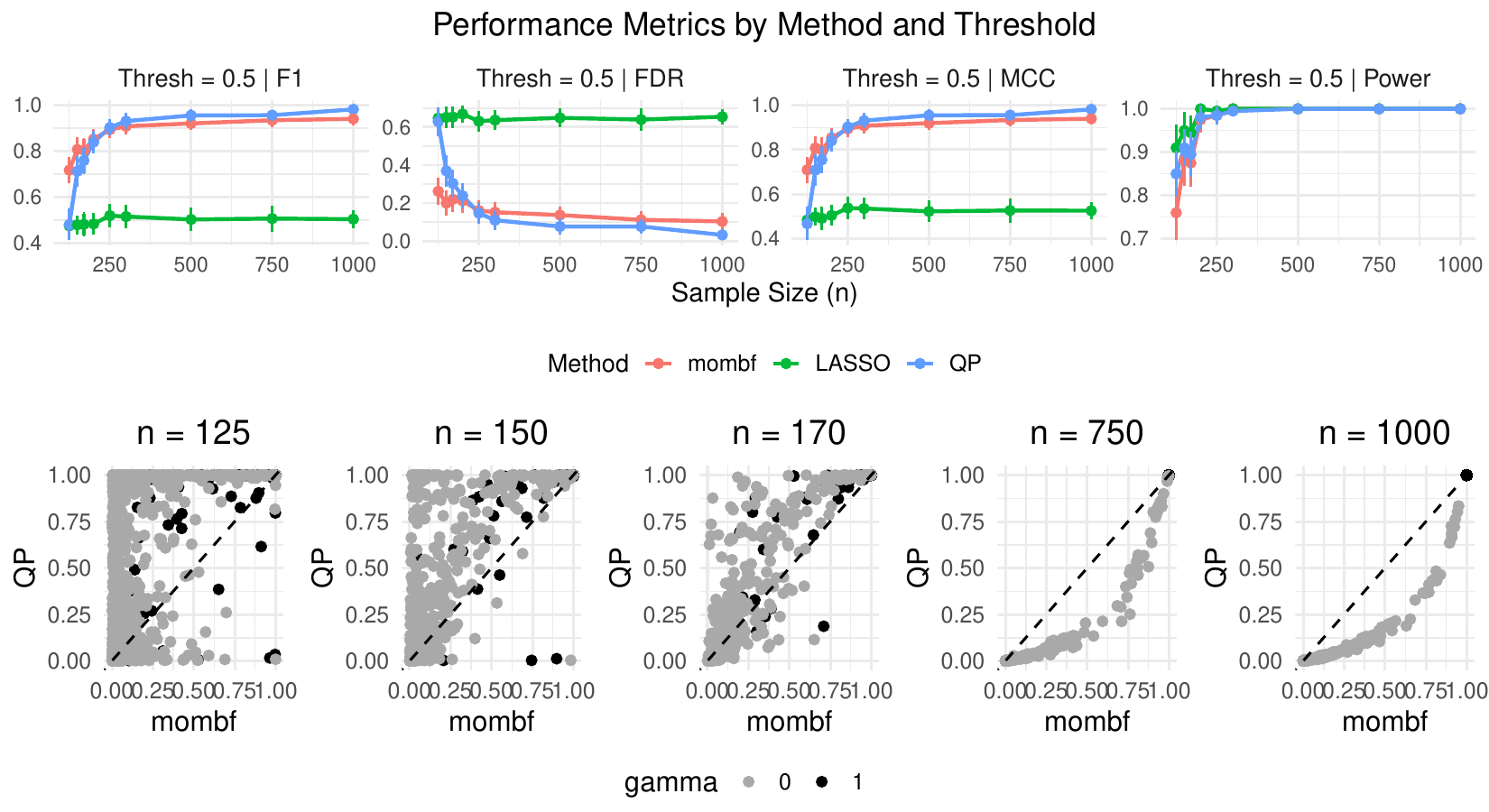}
    \caption{
    Linear regression with inliers. 
    Top: Variable selection performance of quasi–posterior (QP), traditional posterior (mombf) and frequentist LASSO models across sample sizes using fixed cutoff at $\pi_j \ge 0.5$. Error bars denote $\pm 2$ standard errors across simulation replicates. 
Bottom: Estimated posterior probabilities of inclusion  across repeats, coloured according to whether the generating $\beta^\ast$ was zero (grey) or non-zero (black).}
    \label{fig:multcomp_selestion_homo_inlier0.5}
\end{figure}

\begin{figure}[ht!]
    \centering
    \includegraphics[width=0.95\linewidth]{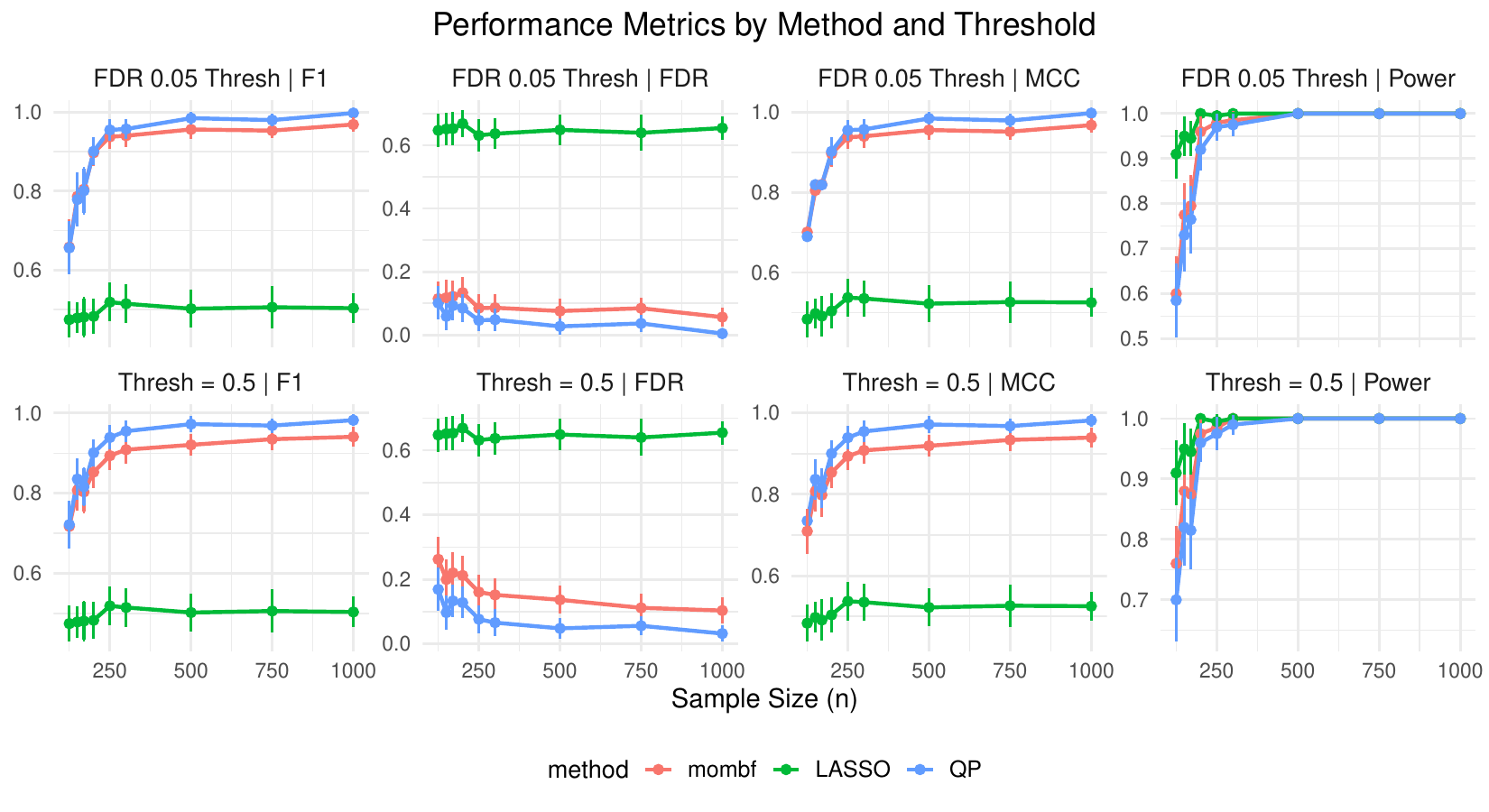}
    \caption{
    Linear regression with inliers.
    Variable selection performance of quasi–posterior (QP), traditional posterior (mombf) and frequentist LASSO models across sample sizes.  
    Results are shown under two selection rules: (top) fixed cutoff at $\pi_j \ge 0.5$ and (bottom) Bayesian FDR control at $\alpha = 0.05$. 
    Error bars denote $\pm 2$ standard errors across simulation replicates. For the QP $\psi$ estimated here using LASSO regularisation to estimate $\hat{\beta}$ in \eqref{Equ:psi_est}.}
    \label{fig:multcomp_metrics_homo_inlier_psiLASSO}
\end{figure}

\subsection{Additional Results for Real Data}
In this section, we present additional analyses for the real data applications, providing further diagnostic and predictive assessments of the competing models.

\subsubsection{NMES 1988 dataset}

The NMES dataset is available in the \texttt{AER} package in \textsf{R}. All continuous covariates were standardized to have mean zero and unit variance prior to analysis. 

\noindent \textbf{Posterior inlcusion probabilities} Figure \ref{fig:NMES1988_ppi_comparison}  reports the posterior inclusion probabilities of the QP, Poisson and negative-binomial (NB) models for the NMES data. The QP and NB results closely align, while the Poisson model highlights additional predictors that are not supported under the others.

\begin{figure}[ht!]
  \centering
  \includegraphics[width=0.8\textwidth]{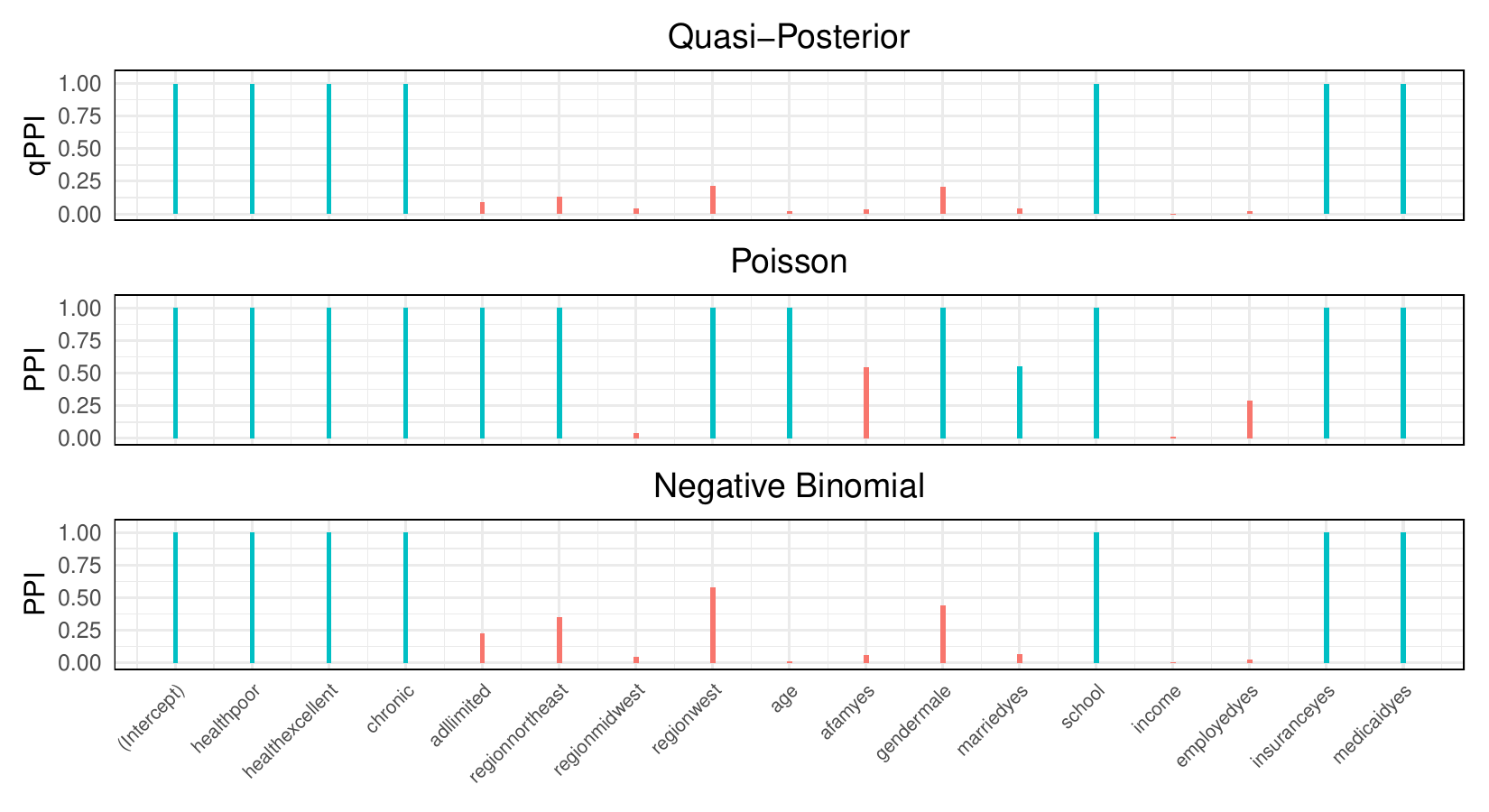}
  \caption{
  NMES data.
Each vertical line represents the marginal posterior probability of inclusion for a covariate. Blue lines correspond to selected predictors using Bayesian FDR control $\alpha = 0.05$.
The three panels show: (top)  QP model, (center) Poisson, (bottom) Negative-binomial.}
  \label{fig:NMES1988_ppi_comparison}
\end{figure}

\noindent \textbf{Mean-Variance Function Diagnostics} \label{sec:variance_function_diagnostics}
To evaluate whether the proposed mean and variance assumption adequately describes the NMES data, we adopt a diagnostic similar to that of \citet{agnoletto2025bayesian} (Supplement~S5), which compares empirical means and variances within bins to model--implied mean and variance functions.   All model parameters were estimated using the maximum a posteriori (MAP) values of the parameters, conditional on the selected $\gamma$ for each model,  with the selected set allowed to differ across models.
Unlike the example in \citet{agnoletto2025bayesian}, which involves a single covariate and allows binning directly along that variable, our models include multiple predictors.  
Accordingly, rather than conditioning on one covariate at a time, we use the fitted mean  
\(\mu_i = \exp(x_i^{\top}\beta)\)  
as a one--dimensional summary index that captures the combined covariate effect.  
To avoid favoring any single model in defining the bins, we construct a neutral index
\[
\bar{\mu}_i = \tfrac{1}{3}\!\left(\hat{\mu}_i^{\text{Pois}} + \hat{\mu}_i^{\text{NB}} + \hat{\mu}_i^{\text{QP}}\right),
\]
take deciles of \(\bar{\mu}_i\), and within each bin \(b\) (retaining those with at least 20 observations) compute the response $y$'s empirical mean \(\bar{y}_b\) and variance \(s_b^2\).

For each fitted model $m\in\{\text{Poisson},\text{NB},\text{QP}\}$ we also compute the bin--averaged fitted mean $\bar\mu^{(m)}_b$ and the corresponding theoretical variance
\[
v_b^{\text{Pois}} \,=\, \bar\mu_b^{\text{Pois}},\qquad
v_b^{\text{NB}} \,=\, \bar\mu_b^{\text{NB}} + \frac{\big(\bar\mu_b^{\text{NB}}\big)^2}{\hat{\theta}},\qquad
v_b^{\text{QP}} \,=\, \hat{\psi}\,\bar\mu_b^{\text{QP}},
\]
where $\hat{\theta}$ is the negative binomial dispersion  and $\hat{\psi}$ is the quasi--Poisson dispersion, both estimated from the respective selected models (each uses its own subset $\gamma$ of predictors).

The left panel of Figure \ref{fig:mean_variance-adequacy_NMES1988} compared empirical \(\bar{y}_b\) with model implied $\bar{\mu}_i^{\text{Pois}}$, $\bar{\mu}_i^{\text{NB}}$, and $\bar{\mu}_i^{\text{QP}}$ and shows that the three models yield very similar mean fits, reproducing the empirical bin averages with small binned errors. 
The right panel of Figure \ref{fig:mean_variance-adequacy_NMES1988} compares the relationship between empirical \(\bar{y}_b\) and $s_b^2$ with the relationships between model-implied $\bar{\mu}_i^{(m)}$ and $v_b^{(m)}$ for each model. The Poisson model markedly underestimates variability across the range of means, consistent with strong overdispersion. The negative binomial captures the increasing variance more faithfully but tends to overshoot at the largest means. The Bayesian quasi--Poisson curve tracks the empirical mean--variance pattern most closely over all bins. 

Table~\ref{tab:adequacy-scores}(a) quantifies the quality of the binned mean estimation by calculating
\[
\text{MSE}^{(m)}=\tfrac{1}{K}\sum_{b=1}^K \big(\bar{y}_b-\bar{\mu}_b^{(m)}\big)^2,\qquad
\text{MAE}^{(m)}=\tfrac{1}{K}\sum_{b=1}^K \big|\bar{y}_b-\bar{\mu}_b^{(m)}\big|,
\]
confirming that all three methods appear to accurately capture the mean structure in the data.
Table~\ref{tab:adequacy-scores}(b) analogously quantifies how well each method capture the binned variances by calculting
\[
\text{MSE}^{(m)}=\tfrac{1}{K}\sum_{b=1}^K \big(s_b^2-v_b^{(m)}\big)^2,\qquad
\text{MAE}^{(m)}=\tfrac{1}{K}\sum_{b=1}^K \big|s_b^2-v_b^{(m)}\big|,
\]
confirming substantial overdispersion in this dataset and indicating that the model quasi-posterior provides the most accurate description of the empirical variability among the candidates considered.



\begin{table}[htbp]
\centering
\small
\begin{subtable}[t]{0.47\textwidth}
  \centering
  \begin{tabular}{lcc}
    \toprule
    \textbf{Model} & \textbf{MSE} & \textbf{MAE} \\
    \midrule
    Poisson  & \textbf{0.246} & \textbf{0.406} \\
    NegBin   & 0.436 & 0.499 \\
    Bayes QP & 0.281 & 0.429 \\
    \bottomrule
  \end{tabular}
  \caption{Mean-function.}
  \label{tab:adequacy-mean}
\end{subtable}
\hfill
\begin{subtable}[t]{0.47\textwidth}
  \centering
  \begin{tabular}{lcc}
    \toprule
    \textbf{Model} & \textbf{MSE} & \textbf{MAE} \\
    \midrule
    Poisson  & 1570.0 & 35.4 \\
    NegBin   &  263.0 & 10.6 \\
    Bayes QP & \textbf{62.3} & \textbf{6.67} \\
    \bottomrule
  \end{tabular}
  \caption{Variance-function.}
  \label{tab:adequacy-var}
\end{subtable}
\caption{Binned error metrics for (a) mean- and (b) variance-function adequacy, computed over deciles of a neutral average of the fitted means across models.  
Entries report mean-squared error (MSE) and mean-absolute error (MAE) between empirical quantities and model-based expectations.  
Each method uses its own selected subset $\gamma$.  
Smaller values indicate better fit; best results are shown in bold.}
\label{tab:adequacy-scores}
\end{table}

While our diagnostic is inspired by the graphical checks proposed by \citet{agnoletto2025bayesian}, there are several important methodological differences that are essential in the present variable--selection and multi--predictor setting. In \citet{agnoletto2025bayesian}, the diagnostic is developed for models with a single covariate, allowing observations to be binned directly along that covariate. In contrast, our models include multiple predictors and distinct selected subsets $\gamma$ across competing specifications, making direct conditioning on any single covariate both arbitrary and potentially misleading. To address this, we introduce a one--dimensional summary index based on the fitted mean, which aggregates the combined effect of all active predictors under the log link and provides a natural scale on which both mean and variance are defined. Moreover, to avoid favoring any particular model when defining the bins, we construct a neutral indexing variable by averaging fitted means across the Poisson, negative binomial, and Bayesian quasi--Poisson models. This neutral binning strategy is new and ensures that the empirical summaries are not implicitly optimized for any single variance specification. Finally, unlike the purely graphical checks in \citet{agnoletto2025bayesian}, we complement the visual diagnostics with quantitative binned error metrics (MSE and MAE), enabling direct numerical comparison of variance and mean adequacy across competing models. Together, these extensions adapt the original diagnostic framework to high--dimensional model--selection settings and provide a principled, model--agnostic assessment of both mean and variance specification.

\vspace{0.2cm}

\noindent \textbf{Gibbs sampling mixing diagnostic: } We assess the mixing behavior of the Gibbs sampler by tracking cumulative posterior inclusion probabilities across iterations, which stabilize rapidly for all three models (Figure~\ref{fig:NMES_diags}).

\begin{figure}[htbp]
  \centering
  \includegraphics[width=\linewidth]{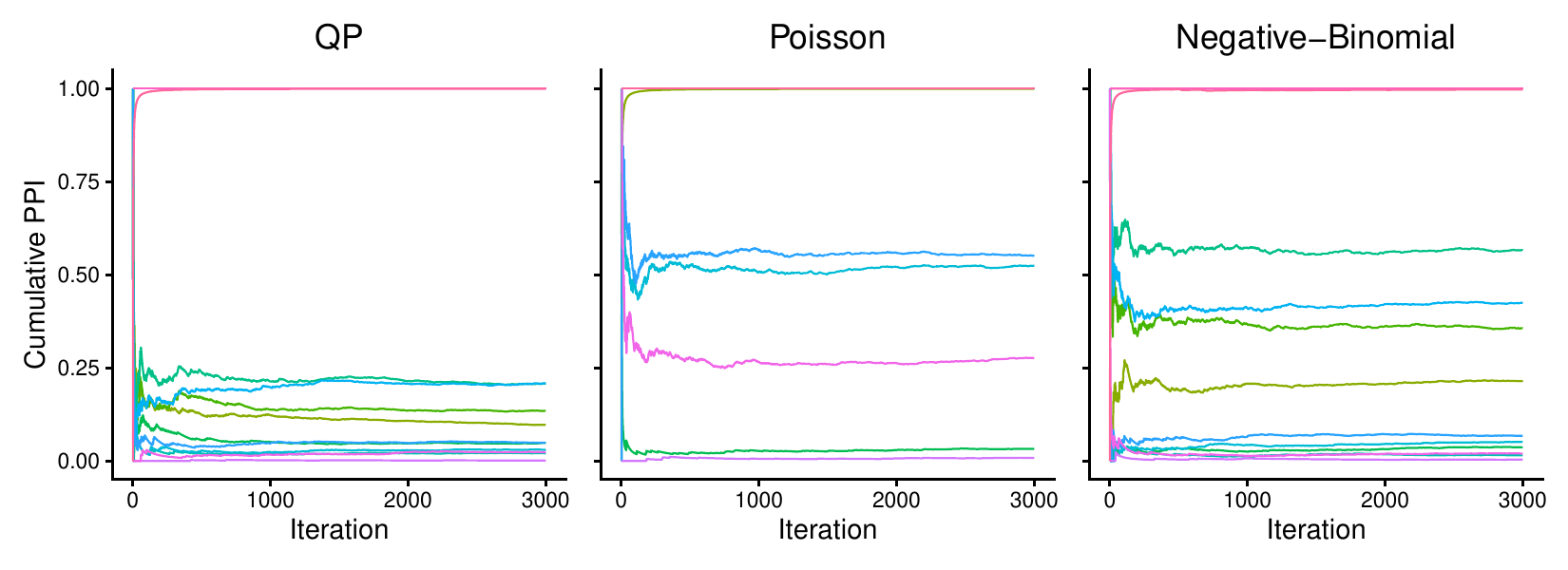}
  \caption{Cumulative estimate of posterior inclusion probabilities with increasing Gibbs sampler iterations for the quasi-posterior, Poisson and negative binomial models.}
  \label{fig:NMES_diags}
\end{figure}

\vspace{0.2cm}

\noindent \textbf{Out of Sample-Predictive Performances:} To evaluate out--of--sample performance, we performed a $10$--fold cross--validation using a variance--weighted mean squared error (WMSE) criterion.  
In each fold,  $10\%$ of the data were held out as a test set, and the remaining observations were used to estimate parameters.  
Within the training set of data, variable selection proceeded as before, and all parameters were obtained as maximum a posteriori (MAP) estimates conditional on the selected subset~$\gamma$.  
For each held--out observation $i\in\mathcal{I}_{\text{test}}$, the predicted mean $\hat\mu_i^{(m)}$ and variance $\hat V_i^{(m)}$ were computed under model $m\in\{\text{Poisson},\text{NB},\text{QP}\}$ using their theoretical variance forms,
\[
\hat V_i^{\text{Pois}} = \hat\mu_i^{\text{Pois}}, \qquad
\hat V_i^{\text{NB}} = \hat\mu_i^{\text{NB}} + \frac{(\hat\mu_i^{\text{NB}})^2}{\hat\theta}, \qquad
\hat V_i^{\text{QP}} = \hat\psi\,\hat\mu_i^{\text{QP}}.
\]
Predictive accuracy was then summarized by the variance--weighted mean squared error,
\[
\mathrm{WMSE}^{(m)} = 
\frac{1}{n_{\text{test}}}
\sum_{i\in\mathcal{I}_{\text{test}}}
\frac{(y_i-\hat\mu_i^{(m)})^2}{\hat V_i^{(m)}},
\]
where larger predictive residuals are penalized more heavily when the model assigns low variance.  
Lower values correspond to better calibrated predictive performance.

\begin{table}[htbp]
\centering
\small
\begin{tabular}{lcc}
  \toprule
  \textbf{Model} & \textbf{WMSE} & \textbf{SE} \\
  \midrule
  QP & \textbf{1.04} & 0.10 \\
  NegBin   & 1.25 & 0.15 \\
  Poisson  & 6.78 & 0.58 \\
  \bottomrule
\end{tabular}
\caption{Ten--fold cross--validation results based on the variance--weighted mean squared error (WMSE).  
Entries report fold--averaged means and standard errors.  
Smaller values indicate higher predictive accuracy; the best performance is shown in bold.}
\label{tab:wmse-results}
\end{table}

The results in Table~\ref{tab:wmse-results} show that the model quasi-posterior achieves the lowest WMSE across folds, indicating superior predictive calibration and robustness to overdispersion.

\subsubsection{DLD dataset}




\noindent \textbf{Assessing the mean-variance Assumption: }{\label{Sec:homo_diagnostics}} Analogously to the mean/variance–function assessment described in Section~\ref{sec:variance_function_diagnostics} for the count models, we examined the adequacy of the mean and variance assumptions for the DLD dataset using a diagnostic similar to that of \citet{agnoletto2025bayesian}. 
The left panel of Figure~\ref{fig:DLD_homo_vs_hetero} demonstrates that the QP model provides an excellent approximation of the binned average responses. The right panel of Figure~\ref{fig:DLD_homo_vs_hetero} compares the empirical binned mean and variance relationship with the QP's homoskedastic assumption. 
This indicates that the constant–variance assumption is adequate for these data, a conclusion further supported by the Goldfeld–Quandt test ($p = 0.99$).

\begin{figure}[ht!]
  \centering
  \includegraphics[width=0.95\linewidth]{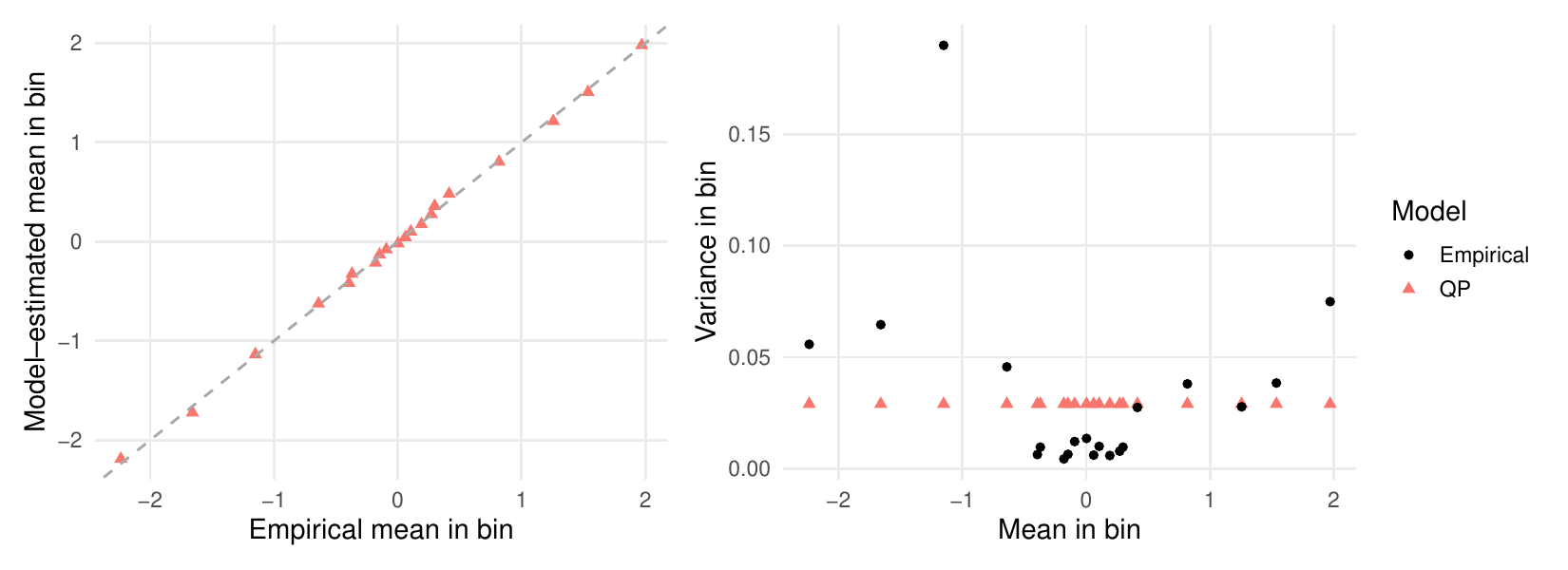}
  \caption{Empirical variances (black points) and model–implied variance functions under the homoscedastic (blue dashed) and heteroscedastic (red solid) quasi–posterior models for the DLD dataset.}
  \label{fig:DLD_homo_vs_hetero}
\end{figure}

\noindent \textbf{Gibbs sampling convergence diagnostics} Convergence of the Gibbs sampler for estimating posterior inclusion probabilities is verified by Figure \ref{fig:DLD_diags}

\begin{figure}[ht!]
  \centering
  \includegraphics[width=0.8\linewidth, trim=0.cm 0.cm 8.75cm  0.cm,clip]{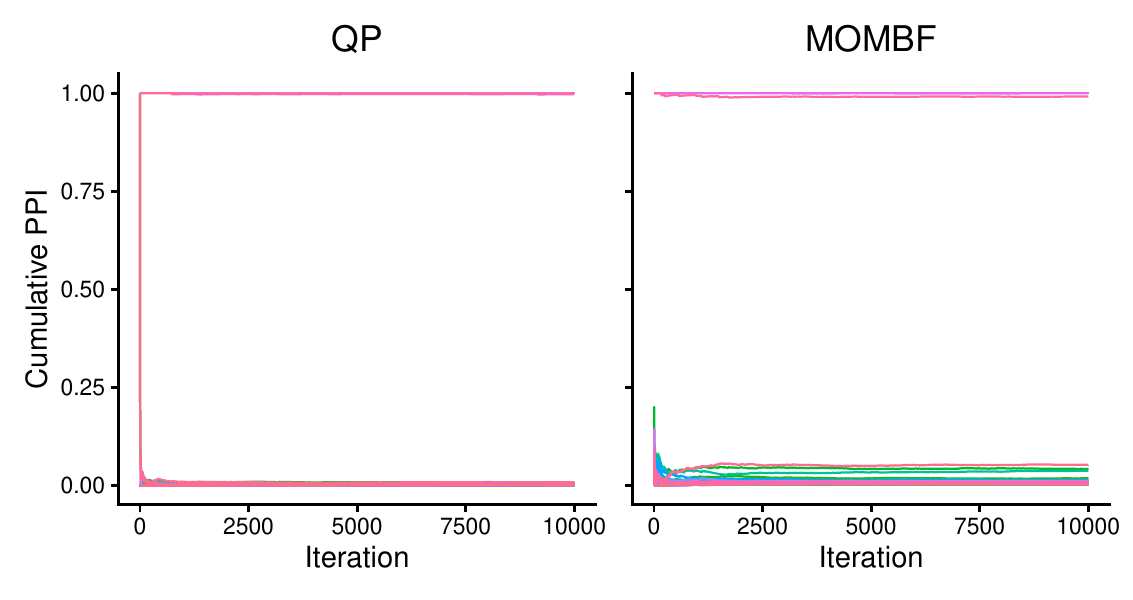}
  \caption{Cumulative estimate of posterior inclusion probabilities with increasing Gibbs sampler iterations for the quasi-posterior and Gaussian (mombf).
  }
  \label{fig:DLD_diags}
\end{figure}

\end{document}